\documentclass[leqno,11pt]{article}
\usepackage[english]{babel}
\usepackage{amsfonts}
\usepackage{amsmath} \usepackage{tikz,graphicx,subfigure,overpic,verbatim}
\usepackage{color} \usepackage{amssymb} \usepackage{amsthm}
\usepackage{hyperref}
\usepackage{a4wide}
\usepackage[fixlanguage]{babelbib}
\usetikzlibrary{arrows,decorations.markings}
\usetikzlibrary{decorations.pathmorphing}

\def\R{\mathbb{R}}
\def\C{\mathbb{C}}
\def\P{\mathbb{P}}
\def\A{\mathcal{A}}
\def\D{D \left(0, n^{-\delta} \right)}

\def\Pinf{P^{(\infty)}}
\def\P0{P^{(0)}}
\def\Tr{\mathop{\mathrm{Tr}}\nolimits}
\def\Re{\mathop{\mathrm{Re}}\nolimits}
\def\Im{\mathop{\mathrm{Im}}\nolimits}

\def\supp{\mathop{\mathrm{supp}}\nolimits}

\def\diag{\mathop{\mathrm{diag}}\nolimits}

\def\Kcr{\mathop{K_{\mathrm{cr}}}\nolimits}

\numberwithin{equation}{section}

\newtheorem{theorem}{Theorem}[section]
\newtheorem{lemma}[theorem]{Lemma}
\newtheorem{proposition}[theorem]{Proposition}

\newtheorem{rhp}[theorem]{RH problem}

\theoremstyle{definition}

\theoremstyle{remark}

\newtheorem{remark}[theorem]{Remark}

\newcommand{\ud}{\,\mathrm{d}}

\begin{document}
\title{A critical phenomenon in the two-matrix model in the quartic/quadratic case}
\author{Maurice Duits\footnotemark[1] \and Dries Geudens\footnotemark[2] }

\date{}
\renewcommand{\thefootnote}{\fnsymbol{footnote}}
\footnotetext[1]{Department of Mathematics, Royal Institute of Technology (KTH), Lindstedtsv\"agen 25, SE-10044 Stockholm, Sweden, email: duits\symbol{'100}kth.se. Partially supported by the grant KAW 2010.0063 from the
Knut and Alice Wallenberg Foundation.}
\footnotetext[2]{Department of Mathematics, Katholieke
Universiteit Leuven, Celestijnenlaan 200B, B-3001 Leuven, Belgium.
email: dries.geudens\symbol{'100}wis.kuleuven.be. Research Assistant of the Fund for Scientific
Research - Flanders (Belgium). }

\maketitle

\begin{abstract}
We study a critical behavior for the eigenvalue statistics in the two-matrix model in the quartic/quadratic case. For certain parameters, the eigenvalue distribution for one of the matrices has a limit that  vanishes like a square root in the \emph{interior} of the support. The main result of the paper is  a new kernel that describes the local eigenvalue correlations near that critical point. The kernel is expressed in terms of a $4\times 4$ Riemann-Hilbert problem related to the Hastings-McLeod solution of the Painlev\'e II equation.  We then compare the new kernel with two other critical phenomena that appeared in the literature before. First, we show that the critical kernel that appears in case of quadratic vanishing of the limiting eigenvalue distribution can be retrieved from the new kernel by means of a  double scaling limit. Second, we briefly discuss the relation with the tacnode singularity in non-colliding Brownian motions that was recently analyzed. Although  the limiting density in that model also vanishes like a square root  at a certain interior point, the process  at the local scale is different from the process that we obtain in the two-matrix model. \end{abstract}

\tableofcontents

 \section{Introduction}
The two-matrix model in random matrix theory is  the probability measure
\[
\frac{1}{Z_n}\exp \left(-n\Tr \left(V(M_1)+W(M_2)-\tau M_1 M_2\right) \right)\ud M_1 \ud M_2,
\]
defined on couples $(M_1,M_2)$ of $n\times n$ Hermitian matrices. Here $\tau>0$ is the coupling constant, $Z_n$ is a normalizing constant and $V$ and $W$ are two polynomials such that the density is integrable. In this paper we are interested in the asymptotic behavior of the eigenvalues as $n\to \infty$.

The two-matrix model is commonly assumed to be an effective model to generate new  critical phenomena for the eigenvalue statistics. In particular, it is believed to generate all $(p,q)$ minimal conformal models, whereas in one-matrix models one only obtains the $(p,2)$ minimal models \cite{DaKaKo,Dou,Eyn2}. However, even for non-critical bulk universality few rigorous results have been obtained and a general description of the asymptotic behavior of the spectrum remains an important open question.

Recent progress has been made  in \cite{DGK,DK2, DKM,Mo} where the authors considered  the special case   $W(y)=y^4/4+\alpha y^2/2$ and analyzed the asymptotic behavior of the eigenvalues of $M_1$ when averaged over $M_2$.  It turns out that the mean eigenvalue density for the matrix $M_1$ has a  limit as $n\to \infty$, which can be characterized by a vector equilibrium problem. In particular, in \cite{DGK} the authors showed that for the special case
\begin{align}\label{eq:defVW}
W(y)=y^4/2+\alpha y^2/2, \qquad \text{and} \qquad V(x)=\tfrac12 x^2,
\end{align}
and the special values 
\[
\alpha=-1 \textrm{ and }\tau=1,
\]
the limiting eigenvalue density for $M_1$ vanishes like a square root near the origin, which lies in the bulk of the spectrum, see also Figure \ref{fig: density of critical equilibrium measure}. As a consequence, near the origin the local correlations are no longer in the universality class of the sine kernel and critical behavior is expected.  It is the purpose of this paper to characterize that critical behavior.

The eigenvalues of $M_1$, when averaged over $M_2$, form a determinantal point process with a kernel that is expressed in terms of certain biorthogonal polynomials \cite{EyM,MS}. Our  main result  is a triple scaling limit of that kernel near the critical point. The limiting kernel is expressed in terms of the unique solution to  a $4\times 4$ model Riemann-Hilbert problem for which the associated ODE can be expressed in terms of  the Hastings-McLeod solution to the Painlev\'e II equation.  This is the solution to the Painlev\'e II equation 
\begin{align}\label{eq:PIIeq}q''(\sigma)=2q(\sigma)^3+\sigma q(\sigma),\end{align}
characterized by
\begin{align}\label{eq:HScharach}
q(\sigma)= {\rm Ai}(\sigma)(1+o(1)), \qquad \text{as }\sigma \to +\infty.
\end{align}
In the proof, we follow the approach of \cite{DKM} and perform a Deift/Zhou steepest descent analysis for the Riemann-Hilbert problem characterizing the biorthogonal polynomials that integrate the two-matrix model. The crucial point in our proof is the construction of the local parametrix  around the critical point in terms of the   $4\times 4$ Riemann-Hilbert problem that defines the kernel. To deal with the triple scaling limit, we will also need to modify the $\lambda$-functions that are used in the normalization of the Riemann-Hilbert problem.

The  $4\times 4$ Riemann-Hilbert problem that characterizes the critical kernel is an extension (we have an additional parameter) of a Riemann-Hilbert problem that appeared before in \cite{DKZ}.   In that paper Delvaux, Kuijlaars, and Zhang analyze the local process near a tacnode  in a model of non-colliding Brownian motions (for alternative approaches to the same problem see \cite{AFM,J}). In this situation, the limiting distribution of particles also vanishes like a square root in the interior of its support. However, we will show that at the local scale this process is different from the one we obtain for the two-matrix model. More precisely, the kernel that is of importance in the two-matrix model is constructed in a different way out of the extended (but for a certain choice of parameters same) Riemann-Hilbert problem. Note that   this means in particular that the universality class is not determined by the vanishing exponent only.

We recall that the Hastings-McLeod solution to the Painlev\'e II equation also plays an important role in a different critical phenomenon in random matrix theory.  It  appears  in the one-matrix model when the limiting density of eigenvalues vanishes quadratically in the interior of its support \cite{BI,CK}. In that case the local behavior is expressed in terms of $\Psi$-functions that solve a $2\times 2$ Riemann-Hilbert problem characterizing the Hastings-McLeod solution to the Painlev\'e II equation. However, we stress that this $2\times 2$ Riemann-Hilbert problem is essentially different from the $4\times 4$ Riemann-Hilbert problem that we will use.

It was already observed in \cite{DKZ} that the underlying Riemann-Hilbert problem for the kernel is intimately connected to the Hastings-McLeod solution of the Painlev\'e II equation.  In the context of the tacnode singularity in the Brownian motion model, this remarkable connection agrees with earlier results in \cite{DelvauxK}.  In the critical case for the two-matrix model, we have an even stronger connection. Indeed, the $\alpha \tau$- phase diagram tells us that it is possible to arrive at that situation, when we move away from the critical point along a certain curve. The second main result of the paper is that we can indeed retrieve the kernel associated with quadratic vanishing, by a double scaling limit of the new critical kernel.

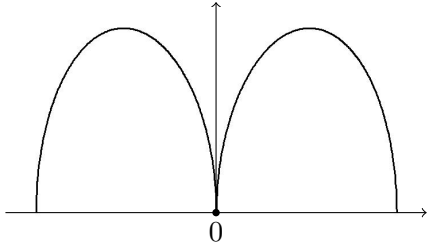
\begin{figure}[t]
\centering
\begin{tikzpicture}[xscale=0.7,yscale=0.7]
\draw (0,1.75) node {\includegraphics[width=140pt]{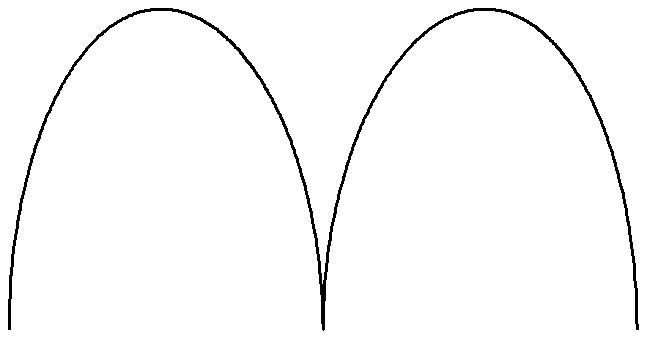}};
\draw[->] (-4,0)--(4,0);
\draw[->] (0,0)--(0,4);
\fill    (0,0) circle (2pt) node[below]{$0$};           
\end{tikzpicture}
\caption{The limiting mean eigenvalue density for critical values $\alpha=-1$, $\tau=1$.}
\label{fig: density of critical equilibrium measure}
\end{figure}

\paragraph{Acknowledgments}
The authors wish to thank  Arno Kuijlaars and Steven Delvaux for many interesting discussions. In particular we are grateful to Steven Delvaux for suggesting an alternative way of defining the conformal map and the parameters in Section \ref{sec:P0} that is more elegant than our original one. The main part of the research was carried out while the first author was an Olga Taussky--John Todd instructor at the California Institute of Technology. The second author was partially supported by the Belgian Interuniversity Attraction Pole P06/02. The authors are also grateful to MSRI in Berkeley for the hospitality and support during the fall of 2010.

\section{Statement of results}

In this section we state our main results. The proofs will be given in the subsequent sections. We start by exploring some preliminary results on the two-matrix model that we will need. See also \cite{DKM} and the references therein.

\subsection{Biorthogonal polynomials}

The two-matrix model can be integrated in terms of biorthogonal polynomials. These are defined as the two families of polynomials $\{p_{k,n}\}_{k=0}^\infty$ and $\{q_{j,n}\}_{j=0}^\infty$, where $p_{k,n}$ and $q_{j,n}$ are monic polynomials of degree $k$ and $j$ respectively, satisfying the orthogonality relation 
\begin{align} \label{eq:biorthogonal}
\int_{\R^2} p_{k,n}(x) q_{j,n}(y) e^{-n\left(V(x)+W(y)-\tau x y \right)} \ \ud x\ud y=\delta_{jk} h_{k,n}^2.
\end{align}
One can show that these polynomials exist and are unique \cite{EMc}. Furthermore, they have properties that are typical for orthogonal polynomials: the zeros are real and simple \cite{EMc} and they satisfy an interlacing property \cite{DGK}. The integrable structure of the polynomials also has been intensively investigated in \cite{BEH1,BEH2,BEH}.
 
The correlation functions for the eigenvalues of $M_1$ and $M_2$ have a determinantal structure and can be expressed in terms of these polynomials \cite{EyM,MS}.  Here we are only interested in the eigenvalues of $M_1$ averaged over $M_2$ and hence we content ourselves by describing the eigenvalues of $M_1$ only. To this end, we define the transformed function $Q_{j,n}(x)$ as
 \begin{align*}
 Q_{j,n}(x)=e^{-n V(x)} \int_\R q_{j,n}(y) e^{-n(W(y)-\tau xy)} \ud y.
 \end{align*} 
 Then it is clear from \eqref{eq:biorthogonal} that $p_{k,n}$ and $Q_{j,n}$  satisfy the orthogonality condition
 \begin{align*}
 \int_\R p_{k,n}(x) Q_{j,n}(x) \ud x=0, \qquad j\neq k.
 \end{align*}
 The fact of the matter now is that the eigenvalues of $M_1$, when averaged over $M_2$, form a determinantal point process  with kernel
 \begin{align} \label{eq:defKn}
 K_n(x_1,x_2)=\sum_{k=0}^{n-1}  \frac{1}{h_{k,n}^2}\, p_{k,n}(x_1) Q_{k,n}(x_2).
 \end{align}
For more details on determinantal point processes we refer to \cite{BorDet,HKPV,Jdet,K,L,Sosh}. 
 Note that in the language of \cite{Bor}, the eigenvalues of $M_1$ form a biorthogonal ensemble. As the eigenvalue statistics are fully described by the kernel $K_n$, it remains to analyze the asymptotic behavior of this kernel near the critical point. 
 
Riemann-Hilbert techniques offer a natural approach to find the asymptotic behavior of $K_n$. We recall that the asymptotic behavior of orthogonal polynomials describing the one-matrix models (unitary ensembles), can be effectively computed using the Deift/Zhou steepest descent method \cite{DZ} on an associated Riemann-Hilbert problem introduced in \cite{FIK}. This procedure was carried out in two important papers \cite{DKMVZ1,DKMVZ2} in the regular situation. See also \cite{Deift} for a detailed discussion. This development had a strong impact on the orthogonal polynomial literature. In many subsequent papers asymptotic results for (multiple) orthogonal polynomials were obtained through Riemann-Hilbert techniques. In particular, this method proved to be very successful to establish various universality results in random matrix theory. 

Also the biorthogonal polynomials describing the two-matrix model can be characterized by a Riemann-Hilbert problem \cite{BEH,EMc,Kapaev,KMcL}. Although the general situation is still open, for the special case of even $V$ and $W(y)=y^4/4+\alpha y^2/2$, the steepest descent method has been successfully carried out for the Riemann-Hilbert problem \cite{KMcL} under the assumption that the limiting eigenvalue density is regular. See \cite{DK2,Mo} for the case $\alpha=0$ and \cite{DKM} for general $\alpha$. We will proceed by discussing some of these results.

\subsection{Limiting mean density} \label{sec: limiting mean density}

From this point we assume that $V(x)$ is an even polynomial with positive leading coefficient and 
\[
W(y)=\frac{y^4}{4}+\alpha \frac{y^2}{2}.
\]
For technical reasons we will also assume that $n\equiv 0 \mod 3$.
It then follows from the results in \cite{DKM} that the mean eigenvalue density for $M_1$ has a  limit  as $n \to \infty$, i.e., there exists an absolutely continuous measure $\mu_1$ on $\R$ such that
\begin{align}\label{eq:averagedensitylimit}
\lim_{n\to \infty} \frac{1}{n} K_n(x,x)=\frac{\ud \mu_1}{\ud x}.
\end{align} 
Moreover, from the results in \cite{DKM} one also deduces  that $\mu_1$ is the limiting normalized counting measure on the zeros of $p_{n,n}$
\begin{align}\label{eq:polynomiallimit}
 \frac{1}{n} \sum_{x \ : \  p_{n,n}(x)=0}  \delta_x \to \mu_1.
\end{align}
This was also established in \cite{DGK} for the special case $V(x)=x^2/2$ using different methods.  A crucial point in the analysis of \cite{DKM} is that the limiting measure $\mu_1$ can be characterized as the first component of a vector of measures $(\mu_1,\mu_2,\mu_3)$ that is the unique minimizer of an energy  functional $E$, which we will now discuss. We need some basic notions from potential theory, for which we refer to the standard work \cite{SaffTotik}.

Let us first introduce the logarithmic energy for a measure $\mu$ on $\C$
\[
I(\mu)= \iint \log \frac{1}{|x-y|} \ud \mu(x) \ud \mu(y).
\]
For two measures $\mu$ and $\nu$ the mutual logarithmic energy is defined as
\[
I(\mu,\nu) = \iint \log \frac{1}{|x-y|} \ud \mu(x) \ud \nu(y).
\]
Then the equilibrium problem is to minimize the energy functional
\begin{multline} \label{eq:energyfunctional}
E(\nu_1,\nu_2,\nu_3)= \sum_{j=1}^3 I(\nu_j)-\sum_{j=1}^2 I(\nu_j,\nu_{j+1}) 
+ \int V_1(x) \ud \nu_1(x) + \int V_3(x) \ud \nu_3(x),
\end{multline}
among all measures $\nu_1,$ $\nu_2,$ and $\nu_3$ that satisfy
\begin{itemize}
\item[(a)] the measures have finite logarithmic energy;
\item[(b)] $\nu_1$ is a measure on $\R$ with $\nu_1(\R)=1$;
\item[(c)] $\nu_2$ is a measure on $i\R$ with $\nu_2(i\R)=2/3$;
\item[(d)] $\nu_3$ is a measure on $\R$ with $\nu_3(\R)=1/3$;
\item[(e)] $\nu_2 \leq \sigma_2$ where $\sigma_2$ is a certain measure on the imaginary axis. 
\end{itemize}
The vector equilibrium problem clearly depends on the input data $V_1$, $V_3$, and $\sigma_2$ that we describe next. We will only deal with the situation $\alpha<0$, which is the situation relevant to us. For $\alpha\geq 0$ we refer to \cite{DKM}.  The external field $V_1$ that acts on $\nu_1$ is given by
\[
V_1(x) = V(x) + \min_{s \in \R}(W(s)-\tau x s). 
\]
Now define 
\begin{equation} \label{eq: def xstar}
x^*(\alpha) = 
\frac{2}{\tau} \left( \frac{-\alpha}{3}\right)^{3/2}, \end{equation}
Then, the function $\R \to \R: s \mapsto W(s) - \tau x s$ has a local maximum and a local minimum besides the global minimum if and only if $x \in (-x^*(\alpha),x^*(\alpha))$. In that case the external field $V_3$ is defined as the positive difference between these local extrema. For all other real values of $x$ we define $V_3(x)=0$. Finally, to describe the measure $\sigma_2$ that acts as a constraint on $\nu_2$ we consider the equation
\[
s^3+\alpha s=\tau z, \qquad \textrm{ with } z \in i\R,
\]
and let $s(z)$ denote its solution with the largest real part for $z \in i\R$. Then
\[
\frac{\textrm d \sigma_2(z)}{|\textrm d z|}= \frac{\tau}{\pi} \Re s(z).
\]
and the support of $\sigma_2$ is $i\R$.

In \cite{DKM} it is proved that the above vector equilibrium problem has a unique minimizer $(\mu_1,\mu_2,\mu_3)$ (see also \cite{HK}) and that \eqref{eq:averagedensitylimit}
 holds. It is also proved that each $\mu_j$ has an analytic density with respect to the Lebesgue measure.

\subsection{Phase diagram}

From now on we also assume that $V(x)=\tfrac12 x^2$, so that we are in the situation given by \eqref{eq:defVW}. In order to get a better understanding of the nature of the critical point that we will consider in the present paper, we will first discuss the behavior of the supports of the measures $\mu_1$, $\sigma_2-\mu_2$, and $\mu_3$ and how they depend on $\alpha$ and $\tau$. As proved in \cite{DGK,DKM}, the supports of the measures $\mu_1$, $\sigma_2-\mu_2$, and $\mu_3$ have the following form
\begin{align*}
    \supp (\mu_1) & =[-\beta_0,-\beta_1] \cup [\beta_1, \beta_0], \\
    \supp(\sigma_2-\mu_2) & =i\R \setminus (-i\beta_2,i\beta_2), \\
    \supp(\mu_3) &= \R \setminus (-\beta_3,\beta_3),
\end{align*}
for some $\beta_0 > \beta_1 \geq 0$, $\beta_2,\beta_3 \geq 0$ that all depend on the values of $\alpha \in \mathbb R$ and $\tau > 0$. We distinguish a number of cases, depending on whether $\beta_1$, $\beta_2$, or $\beta_3$ are equal to
zero, or not. At least one of these is zero, and generically, no two consecutive ones are zero.
The situation in summarized in the phase diagram in the $\alpha \tau$-plane shown in
Figure \ref{fig: phase diagram}.
\begin{description}
\item[Case I:] $\beta_1=0$, $\beta_2>0$, and $\beta_3=0$.
Thus, in this case there are no gaps in the supports of the measures $\mu_1$ and $\mu_3$ on
the real line. The constraint $\sigma_2$ is active along an interval $[-i\beta_2, i\beta_2]$
on the imaginary axis.
\item[Case II:] $\beta_1>0$, $\beta_2 > 0$, and $\beta_3=0$.
In Case II there is a gap in the support of $\mu_1$,  but there is no gap in the support of $\mu_3$,
which is again the full real line. The constraint is active
along an interval along the imaginary axis.
\item[Case III:] $\beta_1>0$, $\beta_2 = 0$, and $ \beta_3>0$.
In Case III there is a gap in the supports of $\mu_1$ and $\mu_3$, but the constraint
on the imaginary axis is not active.
\item[Case IV:] $\beta_1=0$, $\beta_2>0$, and $\beta_3>0$.
In this case the measure $\mu_1$ is still supported on one interval. However there
is a gap $(-\beta_3, \beta_3)$ in the support of $\mu_3$. As in Case I, the constraint
$\sigma_2$ is  active along an interval $[-i\beta_2, i\beta_2]$ on the imaginary axis.
\end{description}
In Figure \ref{fig: phase diagram} we plotted a phase diagram that shows which values of $(\alpha,\tau)$ correspond to the different cases. The different cases are separated by the curves  given by the equations
\begin{align*}
    \tau = \sqrt{\alpha + 2}, \quad -2 \leq \alpha < \infty, \quad \text{and} \quad
    \tau = \sqrt{- \frac{1}{\alpha}}, \quad -\infty < \alpha < 0. \end{align*}
On these curves two of the numbers $\beta_1, \beta_2,$ and $\beta_3$
are equal to zero. For example, on the curve between Case III and Case IV, we
have $\beta_1=\beta_2=0$, while $\beta_3>0$.
Finally, note the multi-critical point $(\alpha,\tau)=(-1,1)$
in the phase diagram, where $\beta_1=\beta_2=\beta_3 = 0$.
All four cases come together at this point in the $\alpha \tau$-plane. This point has our main interest.

As long as we consider points $(\alpha,\tau)$ that are not on the  curves, the local correlations are governed by the sine kernel in the bulk of the spectrum and the Airy kernel at the edge of the spectrum. Critical phenomena occur at the curves that separate the different cases.  

When we cross the line $\tau=\sqrt{\alpha+2}$ for $\alpha>-1$, the support of $\mu_1$ turns from two intervals into one interval. On the critical curve, the intervals meet at the origin and  there the density  vanishes quadratically. This indicates that in a double scaling limit the kernel converges to a kernel that is constructed out of $\Psi$-functions for the Hastings-McLeod solution of the Painlev\'e II equation \cite{BI,CK,CKV}.

When crossing the line $\tau=\sqrt{-1/\alpha}$ for $\alpha<-1$, we again have a transition of two intervals merging at the origin. However, due to the fact that $\sigma-\mu_2$ also closes simultaneously, the vanishing at the origin occurs with an exponent ${1/3}$. This indicates that the local correlations are governed by the Pearcey kernel \cite{BK,BH,BH1,OR,TW}.

The other transitions, represented by the dashed lines in Figure \ref{fig: phase diagram}, do not concern $\mu_1$. They take place on the non-physical sheets of the spectral curve and, therefore, they do not influence the local correlations of the eigenvalues of $M_1$, which are again described by the sine and Airy kernels.

In the next paragraph, we will describe  the limiting process near the point $(\alpha,\tau)=(-1,1)$, where there is a simultaneous transition in the supports of all three measures $\mu_1$, $\sigma-\mu_2$ and $\mu_3$.  
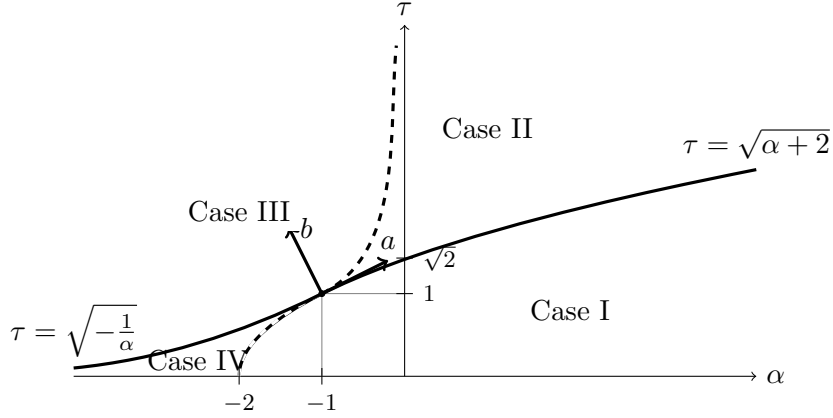
\begin{figure}[t]
\begin{center}
\begin{tikzpicture}[scale=1.1]
\draw[->](0,0)--(0,4.25) node[above]{$\tau$};
\draw[->](-4,0)--(4.25,0) node[right]{$\alpha$};
\draw[help lines] (-1,0)--(-1,1)--(0,1);
\draw[very thick,rotate around={-90:(-2,0)}] (-2,0) parabola (-4.5,6.25)   node[above]{$\tau=\sqrt{\alpha+2}$};
\draw[very thick,dashed,rotate around={-90:(-2,0)},color=white] (-2,0) parabola (-3,1) ;
\draw[very thick,dashed]
(-1,1)..controls (0,1.5) and (-0.2,3).. (-0.1,4);
\draw[very thick]              (-1,1)..controls (-2,0.5) and (-3,0.2).. (-4,0.1) node[above]{$\tau=\sqrt{-\frac{1}{\alpha}}$};
\filldraw  (-1,1) circle (1pt);
\draw[very thick,->] (-1,1)--(-0.2,1.4) node[above]{$a$};
\draw[very thick,->] (-1,1)--(-1.4,1.8) node[right]{$b$};
\draw (0.1,1) node[font=\footnotesize,right]{$1$}--(-0.1,1);
\draw (-1,0.1)--(-1,-0.1) node[font=\footnotesize,below]{$-1$};
\draw (-2,0.1)--(-2,-0.1) node[font=\footnotesize,below]{$-2$};
\draw (0.1,1.43) node [font=\footnotesize,right]{$\sqrt 2$}--(-0.1,1.43);
\draw[very thick] (2,0.8) node[fill=white]{Case I}
                  (-2.5,0.2) node{Case IV}
                  (-2,2) node[fill=white]{Case III}
                  (1,3) node[fill=white]{Case II};
\end{tikzpicture}
\end{center}
\caption{The phase diagram in the $\alpha\tau$-plane: the critical curves $\tau=\sqrt{\alpha+2}$
and $\tau=\sqrt{-\frac{1}{\alpha}}$ separate the four cases.
The cases are distinguished by the fact whether $0$ is in the support of the measures
$\mu_1$, $\sigma_2-\mu_2$, and $\mu_3$, or not.}
\label{fig: phase diagram}
\end{figure}

\subsection{Limiting kernel} \label{subsec: tacnode rhp}

We now come to the main results of the paper. We  study the process near the critical parameters $\tau=1$ and $\alpha=-1$, by means of a triple scaling limit. To this end, we rescale $\alpha$ and $\tau$ near the critical values in the following way
\begin{equation} \label{eq: scaling t tau}
\begin{pmatrix} \alpha \\ \tau \end{pmatrix} = \begin{pmatrix} -1 \\ 1 \end{pmatrix} + a n^{-1/3}\begin{pmatrix} 2 \\ 1 \end{pmatrix} + b n^{-2/3}\begin{pmatrix} -1 \\ 2 \end{pmatrix}, 
\end{equation}
for  $a,b\in \R$.  Note that the vectors $\begin{pmatrix} 2 & 1 \end{pmatrix}^T$ and $\begin{pmatrix} -1 & 2 \end{pmatrix}^T$ are respectively tangent and normal to both critical curves in the point $(\alpha,\tau)=(-1,1)$. We  also scale the space variables with 
\begin{align*}
x=u n^{-2/3}, \qquad \text{and} \qquad y=v n^{-2/3},
\end{align*}
and compute the limiting behavior of $K_n(x,y)$ as $n\to \infty$. 

The limiting kernel  is characterized by the solution to a Riemann-Hilbert problem ($\equiv$ RH problem) that we will first discuss. 
\begin{rhp} \label{rhp: tacnode rhp} Fix parameters $r_1,r_2,s_1,s_2$, and $t$. We search for a $4 \times 4$ matrix-valued function $M(\zeta)$ satisfying
\begin{itemize}
\item[\rm (1)] $M$ is analytic for  $\zeta \in \C \setminus \Sigma_M$;
\item[\rm (2)] $M_+(\zeta)=M_-(\zeta)J_k$, for $\zeta \in \Gamma_k,$ $k=0,\ldots,9$;
\item[\rm (3)] As $\zeta \to \infty$ with $\zeta \in \C \setminus \Sigma_M$ we have
\begin{multline}
M(\zeta)=\left( I+\mathcal O(\zeta^{-1}) \right)B(\zeta) A\\\times
\diag \left( e^{-\psi_2(\zeta)+t\zeta}, e^{-\psi_1(\zeta)-t \zeta}, e^{\psi_2(\zeta)+t \zeta},e^{\psi_1(\zeta)-t \zeta} \right),
\end{multline}
where
\begin{align}
\psi_1(\zeta)&=\frac23 r_1\zeta^{3/2} +2 s_1 \zeta^{1/2}, \label{eq: def psi1}\\
\psi_2(\zeta)&=\frac23 r_2(-\zeta)^{3/2} +2 s_2 (-\zeta)^{1/2} \label{eq: def psi2},
\end{align}
and
\begin{equation} \label{eq: A}
A=\frac{1}{\sqrt 2} \begin{pmatrix} 1 & 0 & -i & 0 \\ 0 & 1& 0&i \\ -i&0&1&0\\0&i&0&1 \end{pmatrix},
\end{equation}
and, finally,
\begin{equation} \label{eq: B}
B(\zeta)=\diag \left((-\zeta)^{-1/4},\zeta^{-1/4},(-\zeta)^{1/4},\zeta^{1/4} \right);
\end{equation}
\item[\rm (4)] $M(\zeta)$ is bounded near $\zeta=0$.
\end{itemize}

The fractional powers in   $\zeta\mapsto \zeta^{3/2}$, $\zeta\mapsto \zeta^{1/2}$ and $\zeta\mapsto \zeta^{\pm 1/4}$ are chosen such that these maps are analytic in $\C\setminus (-\infty,0]$ and taking positive values on the positive part of the real line. The fractional powers in  $\zeta\mapsto (-\zeta)^{3/2}$, $\zeta\mapsto (-\zeta)^{1/2}$ and $\zeta\mapsto (-\zeta)^{\pm 1/4}$ are chosen such that these maps are analytic in $\C\setminus [0,\infty)$ and taking positive values on the negative part of the real line.

The contour $\Sigma_M$ is shown in Figure \ref{fig: contour M} and consists of 10 rays emanating from the origin. The function $M(\zeta)$ makes constant jumps $J_k$ on each of the rays $\Gamma_k$. These rays are determined by two angles $\varphi_1$ and $\varphi_2$ satisfying $0<\varphi_1<\varphi_2<\pi/2$. The half-lines $\Gamma_k,$ $k=0,\ldots,9$, are defined by 
\begin{align*}
\Gamma_0 &= \R^+, & \Gamma_1 &= e^{i\varphi_1} \R^+, & \Gamma_2 &= e^{i\varphi_2} \R^+, \\
&&                 \Gamma_3 &= e^{i(\pi-\varphi_2)} \R^+, & \Gamma_5 &= e^{i(\pi-\varphi_1)} \R^+,    
\end{align*}
and
\[
\Gamma_{5+k}=-\Gamma_k, \qquad k=0,\ldots,4.
\]
All rays are oriented towards infinity.
\end{rhp}
\begin{figure}[t]
\centering
\begin{tikzpicture}[scale=.9]
\begin{scope}[decoration={markings,mark= at position 0.5 with {\arrow{stealth}}}]
\draw[postaction={decorate}]      (0,0)--node[near end, above]{$\Gamma_0$}(4,0) node[right]{$\begin{pmatrix} 0&0&1&0\\0&1&0&0\\-1&0&0&0\\0&0&0&1 \end{pmatrix}$};
\draw[postaction={decorate}]      (0,0)--node[near end, above]{$\Gamma_1$}(3,1) node[above right]{$\begin{pmatrix} 1&0&0&0\\0&1&0&0\\1&0&1&0\\0&0&0&1 \end{pmatrix}$};
\draw[postaction={decorate}]      (0,0)--node[near end, right]{$\Gamma_2$}(1.5,2.5) node[above]{$\begin{pmatrix} 1&0&0&0\\-1&1&0&0\\0&0&1&1\\0&0&0&1 \end{pmatrix}$};
\draw[postaction={decorate}]      (0,0)--node[near end, right]{$\Gamma_3$}(-1.5,2.5) node[above]{$\begin{pmatrix} 1&1&0&0\\0&1&0&0\\0&0&1&0\\0&0&-1&1 \end{pmatrix}$};
\draw[postaction={decorate}]      (0,0)--node[near end, above]{$\Gamma_4$}(-3,1) node[above left]{$\begin{pmatrix} 1&0&0&0\\0&1&0&0\\0&0&1&0\\0&-1&0&1 \end{pmatrix}$} ;
\draw[postaction={decorate}]      (0,0)--node[near end, above]{$\Gamma_5$}(-4,0) node[left]{$\begin{pmatrix} 1&0&0&0\\0&0&0&-1\\0&0&1&0\\0&1&0&0 \end{pmatrix}$};
\draw[postaction={decorate}]      (0,0)--node[near end, above]{$\Gamma_6$}(-3,-1) node[below left]{$\begin{pmatrix} 1&0&0&0\\0&1&0&0\\0&0&1&0\\0&-1&0&1 \end{pmatrix}$} ;
\draw[postaction={decorate}]      (0,0)--node[near end, right]{$\Gamma_7$}(-1.5,-2.5)node[below]{$\begin{pmatrix} 1&-1&0&0\\0&1&0&0\\0&0&1&0\\0&0&1&1 \end{pmatrix}$};
\draw[postaction={decorate}]      (0,0)--node[near end, right]{$\Gamma_8$}(1.5,-2.5)node[below]{$\begin{pmatrix} 1&0&0&0\\1&1&0&0\\0&0&1&-1\\0&0&0&1 \end{pmatrix}$};
\draw[postaction={decorate}]      (0,0)--node[near end, above]{$\Gamma_9$}(3,-1)node[below right]{$\begin{pmatrix} 1&0&0&0\\0&1&0&0\\1&0&1&0\\0&0&0&1 \end{pmatrix}$};
\end{scope}
\end{tikzpicture}
\caption{The jump contour $\Sigma_M$ in the complex $\zeta$-plane and the constant jump matrices $J_k$ on each of the rays $\Gamma_k$, $k=0, \ldots,9$.}
\label{fig: contour M}
\end{figure}
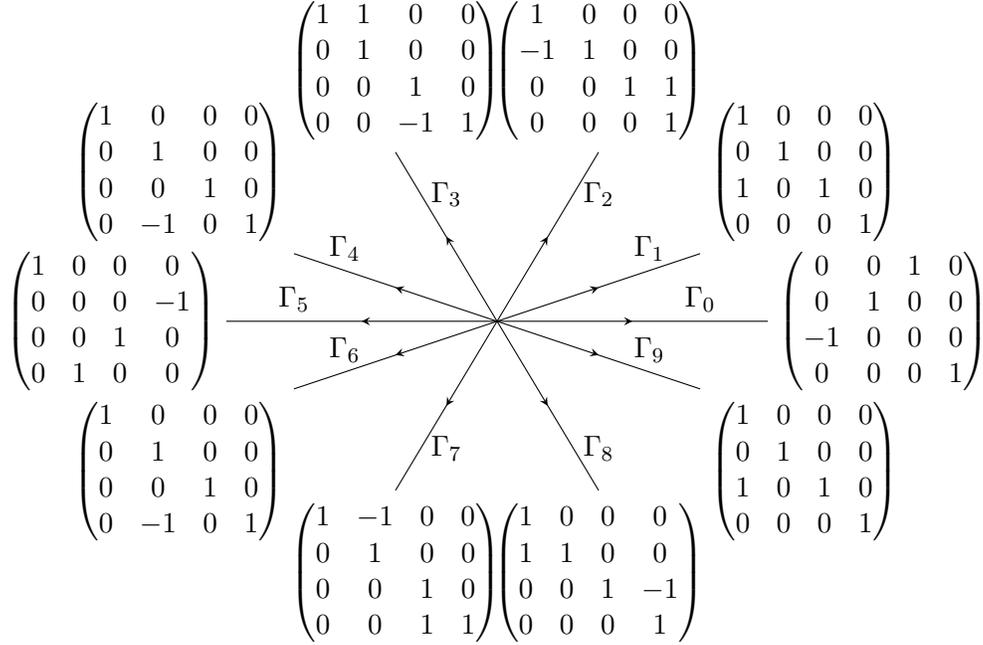

Notations like $M_\pm(z)$ are used throughout the paper and should be understood as follows. Suppose that $z$ belongs to an oriented contour (here $z \in \Sigma_M$). Locally the contour divides the complex plane into two parts. The part on the left (right) side when traversing the contour according to its orientation will be called the + (-) side. Then $M_+(z)$ ($M_-(z)$) denotes the limiting value of $M(\tilde z)$ as $\tilde z$ approaches $z$ from the + (-) side of the contour. 

In case $t=0$, the RH problem \ref{rhp: tacnode rhp} has already appeared before in connection with the tacnode singularity in the model of non-colliding Brownian motions \cite{DKZ}. The following result on the solvability of the RH problem for $M$ was already obtained in \cite{DKZ} for the special situation $t=0$. Here we will also need the statement for $t\in \R$. 

\begin{theorem}\label{th: tacnode rhp}
Assume $r_1=r_2>0$, $s_1=s_2\in \R$, and $t\in \R$. Then RH problem \ref{rhp: tacnode rhp} has a unique solution.
\end{theorem}
The proof of this theorem will be given in Section \ref{sec:existence}. We remark that we do not prove the existence by a vanishing lemma, which does not seem to work for $t\in \R$. Instead, we prove the existence by combining the existence for $t=0$, as obtained in \cite{DKZ}, with standard considerations from the theory of isomonodromy deformations.

As mentioned in the introduction RH problem \ref{rhp: tacnode rhp} is connected to the Hastings-McLeod solution to the Painlev\'e II equation, see \eqref{eq:PIIeq} and \eqref{eq:HScharach}. This will be made more precise in Section \ref{sec: lax pair} when we discuss the Lax pair behind the RH problem. For now, we only mention that it is possible to retrieve the Hastings-McLeod solution by considering the following limit
\begin{multline*}
 \lim_{\zeta\to \infty} \left(\zeta M(\zeta)\diag \left( e^{\psi_2(\zeta)-t\zeta}, e^{\psi_1(\zeta)+t \zeta}, e^{-\psi_2(\zeta)-t \zeta},e^{-\psi_1(\zeta)+t \zeta} \right)A^{-1}B(\zeta)^{-1}\right)_{1,4} \\=i2^{-1/3}q\left(2^{2/3}(2s-t^2)\right),
\end{multline*}
where $q$ stands for the Hastings-McLeod solution and we assumed $r_1=r_2=1$ and $s_1=s_2=s$.

Now that we have established the existence of the solution of the RH problem, we will define the limiting kernel.  It is convenient to first transform $M$ in the following way
\begin{equation} \label{eq: Mbar definition}
\widetilde M(\zeta)=M(i\zeta)^{-T} C_\pm, \qquad \pm \Im \zeta>0,
\end{equation}
where
\begin{equation} \label{eq: def C+-}
C_+=\begin{pmatrix}  1 & 0&0&0 \\ 0&-1 &0&0 \\ 0&0 &0 &- 1 \\ 0&0&-1&0 \end{pmatrix} \textrm{  and   }
C_-=\begin{pmatrix}  0 &-1&0&0 \\ -1& 0&0&0 \\ 0&0 &-1 & 0 \\ 0&0& 0&1 \end{pmatrix}.
\end{equation}
Naturally, the function $\widetilde M$ is also characterized by a RH problem, that can be easily obtained out of the RH problem for $M$. The precise formulation of that RH problem will be given in Section \ref{sec:modelRHP}, where it is used in the construction of the local parametrix.

We now define $K_{\rm cr}$ by
\begin{align}\label{eq:defKcr}
K_{\rm cr}(u,v;s,t)=
\frac{1}{2 \pi i (u-v)} \begin{pmatrix} -1 & 1 & 0 & 0 \end{pmatrix}  \widetilde M_+^{-1}(v)   \widetilde M_+(u) \begin{pmatrix} 1\\1\\0\\0 \end{pmatrix},
\end{align}
where $\widetilde M_+$ is the $+$-boundary value of $\widetilde M$, as in \eqref{eq: Mbar definition}, with $M(u)$ the unique solution $M$ to RH problem \ref{rhp: tacnode rhp} with parameters 
\begin{align}\label{eq:specialparam}
\begin{cases}
r_1=r_2=1>0,\\
s_1=s_2=s\in \R,\\
t\in \R. 
\end{cases}
\end{align} 
A straightforward calculation shows that kernel $K(u,v)$ depends analytically on $u,v$, also in the origin.

The following theorem is the main result of this paper.
\begin{theorem} \label{th: main theorem}
Let $K_{n}$ be the kernel in \eqref{eq:defKn} describing the eigenvalues of $M_1$ when averaged over $M_2$ with $V$ and $W$ as in \eqref{eq:defVW}. Set  
\begin{equation*} 
\begin{pmatrix} \alpha \\ \tau \end{pmatrix} = \begin{pmatrix} -1 \\ 1 \end{pmatrix} + a n^{-1/3}\begin{pmatrix} 2 \\ 1 \end{pmatrix} + b n^{-2/3}\begin{pmatrix} -1 \\ 2 \end{pmatrix}, 
\end{equation*}
for  $a,b\in \R$. Then for $n\to \infty$ and $n\equiv 0 \mod 6$
\begin{equation*}
\lim_{n \to \infty} \frac{1}{n^{2/3}}K_n\left(\frac{u}{n^{2/3}},\frac{v}{n^{2/3}}\right) =K_{\rm cr}\left(u,v;\tfrac14 (a^2-5b),-a\right),
\end{equation*}
uniformly for $u,v$ in compact subsets of $\R$, where $K_{\rm cr}$ is as in \eqref{eq:defKcr} and \eqref{eq:specialparam}.
\end{theorem} 
To the best of our knowledge, the kernel $K_{\rm cr}$ and the associated process have not appeared in the literature before.

\begin{remark}
The assumption $n\equiv 0 \mod 6$  in Theorem \ref{th: main theorem} is mainly for technical reasons. We recall that  also in \cite{DK2,DKM,Mo} the assumption $n\equiv 0\mod 3$ was used. Here we will need the additional assumption that $n$ is even. However, we believe that these assumptions are for technical reasons only and that with some additional argument this condition can be dropped entirely.
\end{remark}

\subsection{A double scaling limit for $K_{{\rm cr}}$}

The nature of the critical point suggests that it is possible to take interesting limits of the kernel. Indeed, if we walk away to the right starting from the critical point and staying on the curve $\tau=\sqrt{\alpha+2}$, we expect to retrieve the Painlev\'e II kernel that describes the transition from Case I to Case II. In the local variables $a,b$ this corresponds to setting  $b=-a^2/5$ and taking the limit $a\to +\infty$. This leads to $t=-a$ and $s=a^2/2$ as parameters for $K_{\rm cr}$. In fact, to obtain the limit in the most general form  one  needs to perform a double scaling limit. Hence we set
\begin{align}
\begin{cases}
s=\tfrac{1}{2} a^2,\\
t=-a\left(1-\frac{\sigma}{a^2}\right),
\end{cases}
\end{align}
for fixed $\sigma\in \R$ and take the limit $a\to +\infty$.

Consider the following RH problem.
\begin{rhp} \label{rhp: PII rhp}
We look for a $2 \times 2$ matrix-valued function $\Psi(\zeta)$ satisfying
\begin{itemize}
\item[\rm (1)] $\Psi(\zeta)$ is analytic for  $\zeta \in \C \setminus \Sigma_\Psi$;
\item[\rm (2)] $\Psi_+(\zeta)=\Psi_-(\zeta)J_k$, for $\zeta \in \Gamma_k,$ $k=1,\ldots,4$;
\item[\rm (3)] As $\zeta \to \infty$ we have
\[
\Psi(\zeta)=\left( I+ \mathcal O(\zeta^{-1}) \right) \begin{pmatrix} e^{-i \frac43 \zeta^3-i \nu \zeta} & 0 \\ 0& e^{i\frac43 \zeta^3+i \nu \zeta} \end{pmatrix};
\]
\item[\rm (4)] $\Psi(\zeta)$ is bounded near $\zeta=0$.
\end{itemize}
The contour $\Sigma_\Psi=\Gamma_1 \cup \Gamma_2 \cup \Gamma_3 \cup \Gamma_4$, where
\[
\Gamma_1=e^{\pi i/6}\R^+, \quad \Gamma_2=e^{5\pi i/6}\R^+, \quad \Gamma_3=-\Gamma_1, \quad \Gamma_4=-\Gamma_2.
\] 
All rays are oriented towards infinity. See also Figure \ref{fig: contour Psi}.
\end{rhp}
This RH problem was introduced by Flaschka and Newell in \cite{FN}. They showed that from this RH problem one can retrieve the Hastings-McLeod solution for the Painlev\'e II equation. More precisely, define $q(\nu)$ by 
\begin{align}
q(\nu)=\lim_{\zeta\to \infty} \zeta \Psi_{12}(\zeta;\nu) e^{-i \frac43 \zeta^3-i \nu \zeta}, 
\end{align}
where $\Psi_{12}$ is the $12$-entry of $\Psi$, then $q$ solves the Painlev\'e II equation \eqref{eq:PIIeq} and satisfies \eqref{eq:HScharach}. One can show that there exists a unique solution $\Psi$ to RH problem \ref{rhp: PII rhp} if and only if the Hastings-McLeod solution $q$ has no pole at $\nu$. Since it is known that this solution has no poles on the real axis \cite{HMcL}, it follows that $\Psi$ exists for all $\nu \in \R$.

We now define the kernel $K_{\rm PII}$ by 
\begin{equation} \label{eq: PII kernel}
K_{\rm PII}(u,v;\nu)=\frac{1}{2\pi i(u-v)} \begin{pmatrix} 1 & -1 \end{pmatrix} \Psi^{-1}(u;\nu)\Psi(v;\nu) \begin{pmatrix}1 \\ 1\end{pmatrix},
\end{equation}\marginpar{Are $u,v$ interchanged?}
where $\Psi(\zeta,\nu)$ is the unique solution to  RH problem \ref{rhp: PII rhp}. The kernel $K_{{\rm PII}}$ is the kernel that appears as universal object in random matrix theory, when the limiting mean density vanishes quadratically at an interior point of the support, see \cite{BI,CK}.
\begin{figure}[t]
\centering
\begin{tikzpicture}[scale=1]
\begin{scope}[decoration={markings,mark= at position 0.5 with {\arrow{stealth}}}]
\draw[postaction={decorate}]      (0,0)--node[midway, above]{$\Gamma_1$}(1.73,1) node[right]{$J_1=\begin{pmatrix} 1&0 \\ 1&1 \end{pmatrix}$};
\draw[postaction={decorate}]      (0,0)--node[midway, above]{$\Gamma_2$}(-1.73,1) node[left]{$J_2=\begin{pmatrix} 1&0 \\ -1&1 \end{pmatrix}$};
\draw[postaction={decorate}]      (0,0)--node[midway, below]{$\Gamma_3$}(-1.73,-1) node[left]{$J_3=\begin{pmatrix} 1&1 \\ 0&1 \end{pmatrix}$};
\draw[postaction={decorate}]      (0,0)--node[midway, below]{$\Gamma_4$}(1.73,-1) node[right]{$J_4=\begin{pmatrix} 1&-1 \\ 0&1 \end{pmatrix}$};
\end{scope}
\end{tikzpicture}
\caption{The jump contour $\Sigma_\Psi$ in the complex $\zeta$-plane and the constant jump matrices $J_k$ on each of the rays $\Gamma_k$, $k=1, \ldots,4$.}
\label{fig: contour Psi}
\end{figure}
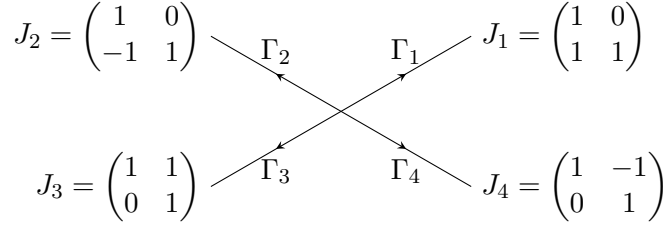

Then we have the following limiting behavior of the critical kernel. 
\begin{theorem} \label{th: PII}
Let $\Kcr$ be as defined in \eqref{eq:defKcr} and $K_{\rm PII}$ as in \eqref{eq: PII kernel}. There exists a function $h$ such that 
\[
\lim_{ a \to +\infty} 2^{5/3} a \frac{h(x,a)}{h(y,a)}\Kcr \left(2^{5/3} ax,2^{5/3} ay;\frac{a^2}{2},-a\left(1-\frac{\sigma}{a^2} \right)\right)=K_{\rm PII}(x,y;2^{5/3}\sigma),
\]
uniformly for $x,y$ in compact sets. 
\end{theorem}

\subsection{Comparison with tacnode situation in non-colliding Brownian motions}
As mentioned in the introduction, the phenomenon of square root vanishing of a limiting density at an interior point of its support can also be obtained in a model of non-colliding Brownian motions \cite{AFM,DKZ,J}. 

Consider $2n$ particles on the real line, half of which start at time $S=-1$ in position $a$ and the other half start in position $-a$. The walkers perform Brownian motions and are conditioned to return to their starting positions at time $S=1$. Moreover, we condition the particles never to collide. For $n\to \infty$ and a proper scaling of $a$, the paths of the particles in the space-time diagram fill out two ellipses that touch at $S=0$, thus forming a tacnode. The limiting density of particles at $S=0$ is described by two touching semi-circles. Hence, we see square root vanishing at the origin, precisely as in the multi-critical situation in the two-matrix model that we just discussed. In \cite{DKZ} the authors show  (in a more general setup) that in a double scaling limit the process at the local scale near the tacnode is governed by the kernel \begin{align}\label{eq:tacnodekernel}
K_{\rm tac}(u,v;r_1,r_2,s_1,s_2)=\frac{1}{2\pi i(u-v)}\begin{pmatrix} -1 & 0 & 1 & 0 \end{pmatrix}M_+^{-1}(v)M_+(u)\begin{pmatrix} 1\\0\\1\\0 \end{pmatrix},
\end{align}
for $u,v>0$, where $M$ is the solution of the RH problem \ref{rhp: tacnode rhp}. Here $r_j$ and $s_j$ are expressed in terms of the relevant parameters of the problem and $t=0$. The kernel for general $u,v\in \R$ is defined by analytic continuation. 

Although the critical situations in the two different models share the same vanishing exponent, the kernels $K_{\rm tac}$ and $K_{\rm cr}$ differ. This may not be immediately obvious from the representations \eqref{eq:defKcr} and \eqref{eq:tacnodekernel}, but can for example be seen by computing the asymptotic behavior along their diagonals. Based on RH problem \ref{rhp: tacnode rhp} one shows that this asymptotic behavior is different for both kernels, which proves that the processes are indeed different. Indeed, for the kernel in the tacnode situation we have
\begin{align}\label{eq:tacnodekernelexp}
K_{\rm tac}(u,u;r_1,r_2,s_1,s_2)= \frac{r_2\sqrt u}{\pi}-\frac{s_2}{\pi \sqrt u}-\frac{1}{4\pi u}\Re e^{2\psi_2(u)}+\mathcal O(u^{-3/2})
\end{align}
as $u\to +\infty$, where we recall the definition of $\psi_2$ \eqref{eq: def psi2}.
Note that there is no constant term in the expansion and that $1/u$ has a highly oscillatory coefficient of modulus 1. For the kernel appearing in the multi-critical situation in the two-matrix model we have
\begin{align}\label{eq:Kcrexp}
K_{\rm cr}(u,u;s,t)= \frac{\sqrt u}{\sqrt 2 \pi}+\frac{t}{\pi}+\frac{s}{\sqrt 2 \pi \sqrt u}+\mathcal O(u^{-3/2}),
\end{align}
as $u\to +\infty$. Here we do have a constant contribution, but there is no $1/u$ term. Therefore, even for well-chosen parameters, these two expansions cannot be the same and hence the kernels are essentially different. 

The computations that lead to \eqref{eq:tacnodekernelexp} and \eqref{eq:Kcrexp} follow by rather straightforward computations on the asymptotic behavior of the function $M$ and are given in Appendix A for completeness.

\subsection{Overview of  the rest of the paper}

The rest of the paper is devoted to the proof of our main results.  The proof of Theorem \ref{th: main theorem} is based on a Deift/Zhou steepest descent analysis on a RH problem characterizing the biorthogonal polynomials. This RH problem will be introduced in Section 4. The steepest descent analysis will be given in the same section. In this analysis we will make use of certain meromorphic functions on a Riemann surface that we will define first in Section 3.   In Section 5 we prove Theorem \ref{th: tacnode rhp}. The proof of Theorem \ref{th: PII} can be found in Section 6.

\section{Meromorphic functions and  Riemann surfaces} \label{sec: Riemann surface}

In this section we define several meromorphic functions and derive some of their properties  that we will need later on. In particular, we define the $\lambda$-functions that we use in the normalization step in the steepest descent analysis. The definitions as presented here differ at some point from \cite{DKM}, except in the critical situation $\alpha=-1$ and $\tau=1$. To emphasize this difference we speak about modified Riemann surface and modified $\xi$,$\lambda$-functions. These modifications are typical when dealing with double scaling limits, see for example \cite{BK,DKZ,DK,KMW} for the case of more than two sheets.

Throughout this section, we will introduce new functions ($\theta$, $w,$ $\xi$, $\lambda$, $h$, $g$, $H$, $G, \ldots$) and constants ($\gamma$, $c,\ldots$).  Although it is not indicated in the notation, these functions and numbers depend on $\alpha$ and $\tau$. We will always assume that the parameters $\tau$ and $\alpha$ are close to their critical values $\tau=1$ and $\alpha=-1$. When dealing with functions or numbers associated with the critical parameter values we add a star to the notation, thus, we write $w^*$, $\xi^*$, $\lambda^*$, $\gamma^*$, $c^*,\ldots$. 

\subsection{Definition of the functions $\theta_j$} \label{sec: theta}
We start with defining a function $\theta$ on  a three-sheeted Riemann surface. The definition is exactly the same as in \cite{DKM}. 

Consider the equation 
\[s^3+\alpha s=\tau z,\]
for $\alpha <0$. If $z=x\in \R$ and $|x|<x^*(\alpha)=\frac{2}{\tau}\left(\frac{-\alpha}{3}\right)^{3/2}$, then the equation has three real solutions $s_j=s_j(x)$ that we take to be ordered such that
\[W(s_1)-\tau x s_1 \leq W(s_2)-\tau xs_2\leq W(s_3)-\tau x s_3.\]
In other words, the function $W(s)-\tau x z$ has a global minimum at $s_1$, a local minimum at $s_2$ and a local maximum at $s_3$. The functions $x\mapsto s_j(x)$ extend to a meromorphic function in the following way. 

Consider the three-sheeted Riemann surface $\mathcal S$ with sheets
\begin{align}
\left\{
\begin{array}{ll}\mathcal S_1&=\C\setminus i \R\\
\mathcal S_2&=\C\setminus \left((-\infty,-x^*(\alpha)]\cup[x^*(\alpha),\infty)\cup i \R\right)\\
\mathcal S_3&=\C\setminus \left((-\infty,-x^*(\alpha)]\cup[x^*(\alpha),\infty)\right)
\end{array} \right.
\end{align} 
Then the function $s_j$ has an analytic continuation to $S_j$, which we also denote $s_j$. It is straightforward to check that $s_2$ and $s_3$ can also be obtained by analytic continuation of $s_1$ onto $\mathcal S_2$ and $\mathcal S_3$. 

Finally, we define the functions $\theta_j$ by 
\begin{equation} \label{eq: def thetaj}
\theta_j(z)=-W(s_j(z))+\tau z s_j(z), \quad j=1,2,3.
\end{equation}
We will need the asymptotic behavior as $z\to \infty$ that is given by \cite[Lem. 2.4]{DKM}
\begin{align}
\theta_1(z)&= \begin{cases} 
\frac34 (\tau z)^{4/3}-\frac{\alpha}{2}(\tau z )^{2/3}+\frac{\alpha^2}{6}-\frac{\alpha^3}{54}(\tau z )^{-2/3}+\mathcal O\left(z^{-4/3} \right), & \text{in }I \cup IV, \\
\frac34 \omega(\tau z)^{4/3}-\frac{\alpha}{2}\omega^2(\tau z )^{2/3}+\frac{\alpha^2}{6}-\frac{\alpha^3}{54}\omega(\tau z )^{-2/3}+\mathcal O\left(z^{-4/3} \right), & \text{in }II, \\
\frac34 \omega^2(\tau z)^{4/3}-\frac{\alpha}{2}\omega(\tau z )^{2/3}+\frac{\alpha^2}{6}-\frac{\alpha^3}{54}\omega^2(\tau z )^{-2/3}+\mathcal O\left(z^{-4/3} \right), & \text{in }III,
\end{cases} \label{eq: asym theta 1}\\
\theta_2(z)&= \begin{cases} 
\frac34 \omega(\tau z)^{4/3}-\frac{\alpha}{2}\omega^2(\tau z )^{2/3}+\frac{\alpha^2}{6}-\frac{\alpha^3}{54}\omega(\tau z )^{-2/3}+\mathcal O\left(z^{-4/3} \right), & \text{in }I, \\
\frac34 (\tau z)^{4/3}-\frac{\alpha}{2}(\tau z )^{2/3}+\frac{\alpha^2}{6}-\frac{\alpha^3}{54}(\tau z )^{-2/3}+\mathcal O\left(z^{-4/3} \right), & \text{in }II \cup III, \\
\frac34 \omega^2(\tau z)^{4/3}-\frac{\alpha}{2}\omega(\tau z )^{2/3}+\frac{\alpha^2}{6}-\frac{\alpha^3}{54}\omega^2(\tau z )^{-2/3}+\mathcal O\left(z^{-4/3} \right), & \text{in }IV,
\end{cases} \label{eq: asym theta 2}\\
\theta_3(z)&= \begin{cases} 
\frac34 \omega^2(\tau z)^{4/3}-\frac{\alpha}{2}\omega(\tau z )^{2/3}+\frac{\alpha^2}{6}-\frac{\alpha^3}{54}\omega^2(\tau z )^{-2/3}+\mathcal O\left(z^{-4/3} \right), & \text{in }I \cup II, \\
\frac34 \omega(\tau z)^{4/3}-\frac{\alpha}{2}\omega^2(\tau z )^{2/3}+\frac{\alpha^2}{6}-\frac{\alpha^3}{54}\omega(\tau z )^{-2/3}+\mathcal O\left(z^{-4/3} \right), & \text{in }III \cup IV.
\end{cases} \label{eq: asym theta 3}
\end{align}
Here, $I$, $II$, $III$, and $IV$ denote the four open quadrants of the complex plane and $$\omega=e^{2 \pi i /3}.$$

\subsection{Modified Riemann surface $\mathcal R$}

Let us first introduce an auxiliary parameter $\gamma$ that is completely determined by $\alpha$ and $\tau$ but will prove to be convenient for notation. We define $\gamma=\gamma(\alpha,\tau)$ as the solution of
\begin{equation} \label{eq: gamma}
\alpha \tau^{2/3}=\frac{3}{\gamma}-9\gamma^2+5\tau^{4/3}\gamma,
\end{equation}
that tends to 1 as $\tau \to 1$ and $\alpha \to -1$. In the triple scaling limit, i.e. we let $\alpha$ and $\tau$ depend on $n$ as in \eqref{eq: scaling t tau} while $n \to \infty$, we have
\begin{equation} \label{eq: scaling gamma}
\gamma=1+\tfrac13 a n^{-1/3} + \left(  \tfrac{11}{144}a^2 + \tfrac{47}{48}b \right) n^{-2/3} + \mathcal O (n^{-1}).
\end{equation}

Next we introduce a four-sheeted Riemann surface $\mathcal R$ with sheets
\begin{align*}
\mathcal R_1 &= \C \setminus [-c,c], & \mathcal R_2 &= \C \setminus ([-c,c]\cup i\R), \\
\mathcal R_3 &= \C \setminus (\R \cup i \R), & \mathcal R_4 &= \C \setminus \R, 
\end{align*}
where
\begin{equation} \label{eq: c}
c=\frac{16}{3 \sqrt 3}\gamma^{3/2}.
\end{equation}
We connect the sheets $\mathcal R_j$ to each other in the usual crosswise manner along the cuts $[-c,c]$, $\R$, and $i \R$. For example, $\mathcal R_1$ is connected to $\mathcal R_2$ along the cut $[-c,c]$ and so on.  Moreover, the Riemann surface is compactified by adding two points at infinity. The first point $\infty_1$ is added to the first sheet, while the second point $\infty_2$ connects the other sheets. This Riemann surface $\mathcal R$ has genus zero. It is shown in Figure \ref{fig: Riemann surface}. 

Let us now consider the equation 
\begin{equation} \label{eq: algebraic equation for w}
(w^2+\gamma^3)^2=zw^3.
\end{equation}
We claim that this equation defines an algebraic function on the Riemann surface $\mathcal R$.  Indeed, for each $z\in \C$ equation \eqref{eq: algebraic equation for w} has four solutions,  denoted $w_j(z)$, $j=1,2,3,4$. We order the solutions such that
\begin{equation} \label{eq: order}
|w_1(z)| \geq |w_2(z)| \geq |w_3(z)| \geq |w_4(z)|.
\end{equation}
For certain $z\in \C$ this can be done in different ways. In that case we arbitrarily pick a particular order. It follows from the next lemma that the functions $w_j(z)$ can be interpreted as a meromorphic function on the Riemann surface $\mathcal R$. 

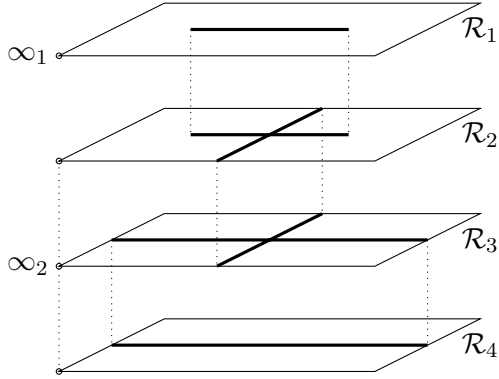
\begin{figure}[t]
\centering
\begin{tikzpicture}[scale=0.7]
\draw (0,0)--(6,0)--(8,1)--(2,1)--cycle 
      (0,-2)--(6,-2)--(8,-1)--(2,-1)--cycle
      (0,-4)--(6,-4)--(8,-3)--(2,-3)--cycle
      (0,-6)--(6,-6)--(8,-5)--(2,-5)--cycle;
\draw[very thick] (2.5,0.5)--(5.5,0.5)
                  (2.5,-1.5)--(5.5,-1.5)
                  (3,-2)--(5,-1)
                  (3,-4)--(5,-3)
                  (1,-3.5)--(7,-3.5)
                  (1,-5.5)--(7,-5.5);
\draw[dotted] (2.5,0.5)--(2.5,-1.5)
              (3,-2)-- (3,-4)
              (1,-3.5)--(1,-5.5)
              (5.5,0.5)--(5.5,-1.5)
              (5,-1)--(5,-3)
              (7,-3.5)--(7,-5.5)
              (0,-2)--(0,-6);
\draw (8,0.5) node{$\mathcal R_1$}
      (8,-1.5) node{$\mathcal R_2$}
      (8,-3.5) node{$\mathcal R_3$}
      (8,-5.5) node{$\mathcal R_4$};
\draw (0,0) circle (1.5pt) node[left]{$\infty_1$}
      (0,-2) circle (1.5pt)
      (0,-4) circle (1.5pt) node[left]{$\infty_2$}
      (0,-6) circle (1.5pt);
\end{tikzpicture}
\caption{Modified Riemann surface for values of the parameters $\alpha$ and $\tau$ close to $\alpha=-1$ and $\tau=1$. Note the points at infinity: $\infty_1$ belongs to $\mathcal R_1$, while $\infty_2$ is  common to the other sheets.}
\label{fig: Riemann surface}
\end{figure}

\begin{lemma} \label{lemma: w meromorf}
The function $w_j(z)$ is analytic on $\mathcal R_j,$ $j=1,2,3,4$, and the following relations hold
\begin{align*}
w_{j,+}(x) &= \overline{w_{j,-}(x)}, && x \in (-c,0)\cup(0,c), \qquad j=1,2, \\
w_{1,\pm}(x) &= w_{2,\mp}(x), && x \in (-c,0)\cup(0,c), \\
w_{j,+}(z) &= -\overline{w_{j,-}(z)}, && z \in i\R \setminus \{0\}, \qquad j=2,3, \\
w_{2,\pm}(z) &= w_{3,\mp}(z), && z \in i\R \setminus \{0\}, \\
w_{j,+}(x) &= \overline{w_{j,-}(x)}, && x \in \R \setminus \{0\}, \qquad j=3,4,\\
w_{3,\pm}(x) &= w_{4,\mp}(x), && x \in \R \setminus \{0\}.
\end{align*}
Here $[-c,c]$ and $\R$ are oriented from left to right and $i\R$ from $-i\infty$ to $i\infty$. Hence, the function $\bigcup_{j=1}^4 \mathcal R_j \to \C:  \mathcal R_j \ni z \mapsto w_j(z)$ has an analytic continuation to a meromorphic function $w: \mathcal R \to \overline \C$. Moreover, we have the symmetry condition
\begin{equation} \label{eq: symmetry w}
w_j(z)=-w_j(-z), \qquad z \in \C \setminus (\R \cup i\R).
\end{equation}
\end{lemma}
\begin{proof}
 For the case $\gamma=1$ the relations on the cuts follow from \cite[Th. 5.3]{DGK} and the proof of that theorem. A rescaling proves the relations for general $\gamma$.
 
Due to symmetry both $\{w_j(z), j=1,2,3,4 \}$ and  $\{-w_j(-z),j=1,2,3,4\}$ solve the algebraic equation \eqref{eq: algebraic equation for w}. The fact that both sets of solutions have to satisfy \eqref{eq: order} proves \eqref{eq: symmetry w}.
\end{proof}

The following lemmas on the asymptotic behavior of $w_j(z)$ for $z\to 0$ and $z\to \infty$ will be important to us.  Recall the notation $\omega=e^{2 \pi i /3}$.

\begin{lemma} \label{lemma: asymptotics w around 0}
The meromorphic function $w$ on $\mathcal R$ has the following asymptotic behavior as $z \to 0$
\begin{align*}
w_1(z)&=
 i\gamma^{3/2}+\tfrac{1}{2}e^{\pi i/4}\gamma^{3/4} z^{1/2}+\tfrac{1}{4}z+\tfrac{7}{64}e^{-\pi i/4} \gamma^{-3/4}z^{3/2} + \mathcal O(z^2), \qquad \text{ in } I \cup II, \\
w_2(z)&=\begin{cases}
 i\gamma^{3/2}+\frac{1}{2}e^{-3\pi i/4} \gamma^{3/4}z^{1/2}+\frac{1}{4}z+\frac{7}{64}e^{3\pi i/4} \gamma^{-3/4}z^{3/2} + \mathcal O(z^2), & \text{ in } III, \\
 i\gamma^{3/2}+\frac{1}{2}e^{\pi i/4} \gamma^{3/4}z^{1/2}+\frac{1}{4}z+\frac{7}{64}e^{-\pi i/4} \gamma^{-3/4}z^{3/2} + \mathcal O(z^2), & \text{ in } IV,
\end{cases} \\
w_3(z)&=\begin{cases}
 i\gamma^{3/2}+\frac{1}{2}e^{\pi i/4} \gamma^{3/4}z^{1/2}+\frac{1}{4}z+\frac{7}{64}e^{-\pi i/4} \gamma^{-3/4}z^{3/2} + \mathcal O(z^2), & \text{ in } III, \\
 i\gamma^{3/2}+\tfrac{1}{2}e^{-3\pi i/4} \gamma^{3/4}z^{1/2}+\tfrac{1}{4}z+\tfrac{7}{64}e^{3\pi i/4} \gamma^{-3/4}z^{3/2} + \mathcal O(z^2), & \text{ in } IV,
\end{cases} \\
w_4(z)&=
 i\gamma^{3/2}+\frac{1}{2}e^{-3\pi i/4} \gamma^{3/4}z^{1/2}+\frac{1}{4}z+\frac{7}{64}e^{3\pi i/4} \gamma^{-3/4}z^{3/2} + \mathcal O(z^2), \qquad \text{ in } I \cup II.
\end{align*}
In the remaining quadrants the expansion is determined by \eqref{eq: symmetry w} and from above.
\end{lemma}

\begin{lemma} \label{lemma: asymptotics w around infinity}
The meromorphic function $w$ on $\mathcal R$ has the following asymptotic behavior as $z \to \infty$
\begin{align*}
w_1(z)&=z-2 \gamma^3 \tfrac{1}{z}-\tfrac{5\gamma^6}{z^3}+\mathcal O(z^{-5}), \qquad \text{ in }\C, \\
w_2(z)&=\begin{cases}
\gamma^2 z^{-1/3}+\frac{2}{3} \gamma^3 z^{-1}+\frac{7}{9} \gamma^4 z^{-5/3}+ \mathcal O(z^{-7/3}), & \text{ in } I\cup IV, \\
\omega^2 \gamma^2z^{-1/3}+\frac{2}{3}\gamma^3z^{-1}+\frac{7}{9}\omega \gamma^4z^{-5/3}+ \mathcal O(z^{-7/3}), & \text{ in } II, \\
\omega \gamma^2z^{-1/3}+\frac{2}{3}\gamma^3z^{-1}+\frac{7}{9}\omega^2 \gamma^4z^{-5/3}+ \mathcal O(z^{-7/3}), & \text{ in } III, 
\end{cases} \\
w_3(z)&=\begin{cases}
\omega^2 \gamma^2z^{-1/3}+\frac{2}{3}\gamma^3z^{-1}+\frac{7}{9}\omega \gamma^4z^{-5/3}+ \mathcal O(z^{-7/3}), & \text{ in } I, \\
\gamma^2 z^{-1/3}+\frac{2}{3}\gamma^3z^{-1}+\frac{7}{9} \gamma^4z^{-5/3}+ \mathcal O(z^{-7/3}), & \text{ in } II\cup III, \\
\omega \gamma^2z^{-1/3}+\frac{2}{3}\gamma^3z^{-1}+\frac{7}{9}\omega^2 \gamma^4z^{-5/3}+ \mathcal O(z^{-7/3}), & \text{ in } IV, 
\end{cases} \\
w_4(z)&=\begin{cases}
\omega \gamma^2z^{-1/3}+\frac{2}{3}\gamma^3z^{-1}+\frac{7}{9}\omega^2 \gamma^4z^{-5/3}+ \mathcal O(z^{-7/3}), & \text{ in } I \cup II, \\
\omega^2 \gamma^2z^{-1/3}+\frac{2}{3}\gamma^3z^{-1}+\frac{7}{9}\omega \gamma^4z^{-5/3}+ \mathcal O(z^{-7/3}), & \text{ in } III \cup IV.
\end{cases}
\end{align*}
\end{lemma}

\subsection{Modified $\xi$-functions} \label{subsec: xi functions}

We define the modified $\xi$-functions for $j=1,2,3,4$ as
\begin{equation} \label{eq: xi in terms of w}
\xi_j(z) = \frac{w_j(z)^4+(3\gamma^3-1)w_j(z)^2+\tau^{4/3}\gamma^5}{w_j(z)(w_j(z)^2+\gamma^3)}, \qquad z \in \mathcal R_j.
\end{equation}
Then, the analytic extension $\xi: \mathcal R \to \overline \C$ of the function 
\begin{equation} \label{eq: definition xi}
\bigcup_{j=1}^4 \mathcal R_j \to \C:  \mathcal R_j \ni z \mapsto \xi_j(z),
\end{equation}
to the full Riemann surface $\mathcal R$ is meromorphic on $\mathcal R$. This is the content of the following lemma.

\begin{lemma} \label{lemma: xi meromorf}
The function $\xi_j(z)$ is analytic on $\mathcal R_j,$ $j=1,2,3,4$, and the following relations hold
\begin{align}
\label{eq: sym xi 1} \xi_{j,+}(z) &= \overline{\xi_{j,-}(z)}, && z \in (-c,0)\cup(0,c), \qquad j=1,2,\\
\label{eq: sym xi 3} \xi_{1,\pm}(z) &= \xi_{2,\mp}(z), && z \in (-c,0)\cup(0,c), \\
\nonumber \xi_{2,+}(z) &= -\overline{\xi_{2,-}(z)}, && z \in i\R \setminus \{0\}, \qquad j=2,3,\\
\nonumber \xi_{2,\pm}(z) &= \xi_{3,\mp}(z), && z \in i\R \setminus \{0\}, \\
\nonumber \xi_{3,+}(z) &= \overline{\xi_{3,-}(z)}, && z \in \R \setminus \{0\}, \qquad j=3,4,\\
\nonumber \xi_{3,\pm}(z) &= \xi_{4,\mp}(z), && z \in \R \setminus \{0\}.
\end{align}
Moreover, we have the symmetry condition
\begin{equation} \label{eq: symmetry xi}
\xi_j(z)=-\xi_j(-z), \qquad z \in \C \setminus (\R \cup i\R).
\end{equation}
\end{lemma}
\begin{proof}
This is immediate from Lemma \ref{lemma: w meromorf} and \eqref{eq: xi in terms of w}.
\end{proof}

Combining \eqref{eq: xi in terms of w} with Lemma \ref{lemma: asymptotics w around 0} and Lemma \ref{lemma: asymptotics w around infinity} leads to the asymptotic behavior of the $\xi$-functions. 

\begin{lemma}\label{lem:asymptxinear0}
For $z$ in a neighborhood of zero, we have
\begin{equation} \label{eq: xi near zero}
\begin{aligned}
\xi_1(z) &= C z^{-1/2}+g(z)+z^{1/2}h(z) \qquad \text{ in } I \cup II,\\
\xi_2(z) &= \begin{cases}
-C z^{-1/2}+g(z)-z^{1/2}h(z) & \text{ in } III, \\
C z^{-1/2}+g(z)+z^{1/2}h(z) & \text{ in } IV,
\end{cases}\\
\xi_3(z) &= \begin{cases}
C z^{-1/2}+g(z)+z^{1/2}h(z) & \text{ in } III,\\
-C z^{-1/2}+g(z)-z^{1/2}h(z) & \text{ in } IV,
\end{cases}\\
\xi_4(z) &= -C z^{-1/2}+g(z)-z^{1/2}h(z) \qquad \text{ in } I \cup II,
\end{aligned}
\end{equation}
where
\begin{equation} \label{eq: C}
C=e^{3\pi i/4} \gamma^{1/4}\left( -2 \gamma^2 + \tfrac{1}{\gamma} + \tau^{4/3}\gamma \right),
\end{equation}
and $g$ and $h$ are analytic functions in a neighborhood of zero with
\begin{align*}
g(0) &= i \gamma^{-1/2}\left(\tfrac32 \gamma^2 - \tfrac{1}{4\gamma} -\tfrac54 \tau^{4/3} \gamma \right), \\
h(0) &= e^{\pi i/4} \gamma^{-5/4}\left(\tfrac34 \gamma^2 - \tfrac{1}{8\gamma} +\tfrac38 \tau^{4/3} \gamma \right).
\end{align*}
In the remaining quadrants the behavior is determined by \eqref{eq: symmetry xi} and \eqref{eq: xi near zero}.
\end{lemma}

\smallskip

\begin{lemma} \label{lemma: xi at infinity}
The meromorphic $\xi$-function has the following asymptotic behavior as $z \to \infty$
\begin{align*}
\xi_1(z)&=z-\frac{1}{z}+\frac{\tau^{4/3}\gamma^5-\gamma^3-3\gamma^6}{z^3}+\mathcal O(z^{-5}), \qquad \text{ in }\C, \\
\xi_2(z)&=\begin{cases}
\tau^{4/3}z^{1/3}-\frac{1}{3}\alpha \tau^{2/3}  z^{-1/3}+\frac{1}{3}  z^{-1}+ \mathcal O(z^{-5/3}), & \text{ in } I\cup IV, \\
\tau^{4/3}\omega z^{1/3}-\frac{1}{3}\alpha \tau^{2/3} \omega^2 z^{-1/3}+\frac{1}{3}  z^{-1}+ \mathcal O(z^{-5/3}), & \text{ in } II, \\
\tau^{4/3}\omega^2 z^{1/3}-\frac{1}{3}\alpha \tau^{2/3} \omega z^{-1/3}+\frac{1}{3}  z^{-1}+ \mathcal O(z^{-5/3}), & \text{ in } III, 
\end{cases} \\
\xi_3(z)&=\begin{cases}
\tau^{4/3}\omega z^{1/3}-\frac{1}{3}\alpha \tau^{2/3} \omega^2 z^{-1/3}+\frac{1}{3}  z^{-1}+ \mathcal O(z^{-5/3}), & \text{ in } I, \\
\tau^{4/3}z^{1/3}-\frac{1}{3}\alpha \tau^{2/3}  z^{-1/3}+\frac{1}{3}  z^{-1}+ \mathcal O(z^{-5/3}), & \text{ in } II\cup III, \\
\tau^{4/3}\omega^2 z^{1/3}-\frac{1}{3}\alpha \tau^{2/3} \omega z^{-1/3}+\frac{1}{3}  z^{-1}+ \mathcal O(z^{-5/3}), & \text{ in } IV, 
\end{cases} \\
\xi_4(z)&=\begin{cases}
\tau^{4/3}\omega^2 z^{1/3}-\frac{1}{3}\alpha \tau^{2/3} \omega z^{-1/3}+\frac{1}{3}  z^{-1}+ \mathcal O(z^{-5/3}), & \text{ in } I \cup II, \\
\tau^{4/3}\omega z^{1/3}-\frac{1}{3}\alpha \tau^{2/3} \omega^2 z^{-1/3}+\frac{1}{3}  z^{-1}+ \mathcal O(z^{-5/3}), & \text{ in } III \cup IV.
\end{cases}
\end{align*}
\end{lemma}

\begin{remark}
Note that this behavior coincides with that of the unmodified $\xi$-function in \cite[Lem. 4.10]{DKM} up to and including the $1/z$ term. Actually this motivates the definitions of $\xi$ \eqref{eq: xi in terms of w} and $\gamma$ \eqref{eq: gamma}.
\end{remark}

As a corollary to both lemmas we see that $\xi$ has simple poles in both points at infinity. $\xi$ also has at most simple poles in the two points at zero. Only if $C=C(\alpha,\tau)=0$, the $\xi$-functions are bounded in zero. Note that in the multi-critical point $\alpha=-1$, $\tau=1$ the function $\xi$ has simple zeros in the two points in zero. 

We will also need the following integrals of $\xi_j$.
\begin{lemma} \label{lemma: integrals}
We have the following definite integrals
\begin{align}
\Im \int_{-c}^c \xi_{1\pm}(x) \ud x &= \pm \pi, \label{eq: integral 1}  \\ 
\Im \int_{-c}^c \xi_{2\pm}(x) \ud x &= \mp \pi,    \label{eq: integral 2}
\end{align}
with $c$ as in \eqref{eq: c}.
\end{lemma}

\begin{proof}
In view of \eqref{eq: sym xi 3} the expressions \eqref{eq: integral 1} and \eqref{eq: integral 2} are equivalent. Hence, we only prove the first one. Note that by \eqref{eq: sym xi 1}
\[
\Im \int_{-c}^c \xi_{1,\pm}(x) \ud x= \pm \frac{1}{2i}\int_{-c}^c (\xi_{1,+}-\xi_{1,-})(x) \ud x= \mp \frac{1}{2i}\int_\Gamma \xi_1(z) \ud z,
\]
where $\Gamma$ is a closed curve encircling $[-c,c]$ once in counterclockwise directing. Since $\xi_1$ is analytic at infinity, this integral can be evaluated by means of the residue theorem. Computing the residue of $\xi_1$ at infinity using Lemma \ref{lemma: xi at infinity} then leads to \eqref{eq: integral 1}.
\end{proof}

\begin{lemma} \label{lemma: xi on cuts}
In the multi-critical case $\alpha=-1$, $\tau=1$ the following inequalities hold
\begin{align}
\Im \left( \xi^*_{1,+}(x)-\xi^*_{2,+}(x) \right) &>0, && x \in (-c^*,0) \cup (0,c^*), \label{eq: est xi 1} \\
\Im \left( \xi^*_{1,-}(x)-\xi^*_{2,-}(x) \right) &<0, && x \in (-c^*,0) \cup (0,c^*), \nonumber \\
\Im \left( \xi^*_{3,-}(z)-\xi^*_{2,-}(z) \right) &<0, && z \in i\R \setminus \{0\}, \nonumber \\
\Im \left( \xi^*_{3,+}(z)-\xi^*_{2,+}(z) \right) &>0, && z \in i\R \setminus \{0\}, \nonumber \\
\Im \left( \xi^*_{3,+}(x)-\xi^*_{4,+}(x) \right) &>0, && x \in \R \setminus \{0\},  \nonumber \\
\Im \left( \xi^*_{3,-}(x)-\xi^*_{4,-}(x) \right) &<0, && x \in \R \setminus \{0\}.  \nonumber
\end{align}
\end{lemma}
\begin{proof}
In the multi-critical case the $\xi$-function has the simple form
\begin{equation} \label{eq: xicriticalpoint}
\xi_j^*(z)=w_j^*(z)+\frac{1}{w_j^*(z)}, \qquad j=1,2,3,4.
\end{equation}
Exploiting this relation we obtain the following algebraic equation for the $\xi$-functions in the multi-critical point
\begin{equation} \label{eq: spectral curve}
\xi^4-z\xi^3+z^2=0.
\end{equation}
We look for values of $z$ for which this equation has coalescing solutions. Computing the discriminant
\[
z^6(256-27z^2)
\]
we find out that this can only be the case in the points $z=0$ and $z=\pm c^*$. For the case $z=0$ we clearly have $\xi^*_j(0)=0$, $j=1,2,3,4$. For the case $z=\pm c^*$ only two solutions of \eqref{eq: spectral curve} coincide. It follows from Lemma \ref{lemma: xi meromorf} that we must have $\xi_1^*(\pm c^*)=\xi_2^*(\pm c^*)$.

We will now prove \eqref{eq: est xi 1}. Observe from \eqref{eq: sym xi 1}--\eqref{eq: sym xi 3}
\[
\xi^*_{1,\pm}(x)=\overline{\xi^*_{1,\mp}(x)}=\overline{\xi^*_{2,\pm}(x)}, \qquad x \in (-c^*0) \cup (0,c^*),
\]
and hence
\begin{equation} \label{eq: xi proof 1}
\Im \left( \xi^*_{1,+}(x)-\xi^*_{2,+}(x) \right)= 2 \Im \xi^*_{1,+}(x) \neq 0,\qquad x \in (-c^*0) \cup (0,c^*), 
\end{equation}
where we used that $\xi_{1,+}^*(x)$ cannot coincide with $\xi_{2,+}^*$. Then \eqref{eq: est xi 1} follows from \eqref{eq: xi proof 1} and Lemma \ref{lem:asymptxinear0} by continuity. The other inequalities can be proven likewise.
\end{proof}

\subsection{Modified $\lambda$-functions} \label{subsec: lambda functions}

The modified $\lambda$-functions are defined as the following antiderivatives of the modified $\xi$-functions
\begin{equation} \label{eq: definition lambda 1}
\lambda_j(z)=\int_0^z \xi_j(s) \ud s,
\end{equation}
for $j\in \{1,4\}$ and $z \in I \cup II$, and for $j\in \{2,3\}$ and $z \in III \cup IV$. In the other quadrants, i.e. for $j\in \{1,4\}$ and $z \in III \cup IV$, and for $j\in \{2,3\}$ and $z \in I \cup II$,  $\lambda_j$ is defined by
\begin{equation} \label{eq: definition lambda 2}
\lambda_j(z)=\int_0^z \xi_j(s) \ud s +i\pi.
\end{equation}
In both cases, the complete path of integration must stay in the same quadrant. 

The following lemma describes the jump behavior of the $\lambda$-functions.
\begin{lemma} \label{lemma: lambda meromorf}
The functions $\lambda_j$, $j=1,2,3,4$, are analytic within each quadrant and satisfy these jumps
\begin{align}
\label{eq: jumps lambda 1} \lambda_{1,+}(x) &= \lambda_{1,-}(x)-2\pi i  && x<-c \\
\label{eq: jumps lambda 2} \lambda_{2,+}(x) &= \lambda_{2,-}(x)+2\pi i, && x<-c \\
\label{eq: jumps lambda 3} \lambda_{j,+}(x)&= \lambda_{j,-}(x), && x>c, \quad j=1,2, \\
\label{eq: jumps lambda 4} \lambda_{1,\pm}(x)&= \lambda_{2,\mp}(x), && -c<x<c, \\
\label{eq: jumps lambda 5} \lambda_{j,+}(z)&= \lambda_{j,-}(z), && z \in i\R\setminus \{0\}, \quad j=1,4, \\
\label{eq: jumps lambda 6} \lambda_{2,\pm}(z)&= \lambda_{3,\mp}(z), && z \in i\R\setminus \{0\},  \\
\label{eq: jumps lambda 7} \lambda_{3,\pm}(x)&= \lambda_{4,\mp}(x), && x \in \R\setminus \{0\}.
\end{align}
Moreover, we have the following symmetry relations
\begin{align}
\label{eq: sym lambda 1} \lambda_{1,+}(x) &= \overline{\lambda_{1,-}(x)}+\pi i  && x \in (-c,0) \cup (0,c), \\
\label{eq: sym lambda 2} \lambda_{2,+}(x) &= \overline{\lambda_{2,-}(x)}-\pi i  && x \in (-c,0) \cup (0,c), \\
\label{eq: sym lambda 3} \lambda_{2,+}(z) &= \overline{\lambda_{2,-}(z)}  && z \in i\R \setminus \{0\}, \\
\label{eq: sym lambda 4} \lambda_{3,+}(z) &= \overline{\lambda_{3,-}(z)}  && z \in i\R \setminus \{0\}, \\
\label{eq: sym lambda 5} \lambda_{3,+}(x) &= \overline{\lambda_{3,-}(x)}+\pi i  && x \in \R \setminus \{0\}, \\
\label{eq: sym lambda 6} \lambda_{4,+}(x) &= \overline{\lambda_{4,-}(x)}-\pi i  && x \in \R \setminus \{0\}.
\end{align}
\end{lemma}
\begin{proof}
The analyticity within each quadrant and the equalities \eqref{eq: jumps lambda 4}--\eqref{eq: sym lambda 6} are immediate from \eqref{eq: definition lambda 1}--\eqref{eq: definition lambda 2} and Lemma \ref{lemma: xi meromorf}.  Equations \eqref{eq: jumps lambda 1}--\eqref{eq: jumps lambda 3} are a corollary of \eqref{eq: definition lambda 1}--\eqref{eq: definition lambda 2} and Lemma \ref{lemma: integrals}. 
\end{proof}

The behavior of the $\lambda$-functions as $z$ tends to $\infty$ follows by combining Lemma \ref{lemma: xi at infinity} with \eqref{eq: definition lambda 1}--\eqref{eq: definition lambda 2}. The asymptotics are expressed in terms of functions $\theta_j(z)$, $j=1,2,3$, introduced in \eqref{eq: def thetaj}.

\begin{lemma} \label{lemma: asymptotics lambda around infinity}
As $z \to \infty$, we have that
\begin{align*}
\lambda_1(z) &= 
\frac{z^2}{2}-\log z+ \ell_1 + \mathcal O(z^{-2}), \\
\lambda_2(z) &= \theta_1(z)+\begin{cases}
\frac13 \log z + \ell_2 + D z^{-2/3}+\mathcal O(z^{-4/3}) & \text{ in } I \cup IV, \\
\frac13 \log z + \ell_3 + D \omega z^{-2/3}+\mathcal O(z^{-4/3}) & \text{ in } II,\\
\frac13 \log z + \ell_4 + D \omega^2 z^{-2/3}+\mathcal O(z^{-4/3}) & \text{ in } III,
\end{cases}\\
\lambda_3(z) &= \theta_2(z)+\begin{cases}
\frac13 \log z + \ell_3 + D \omega z^{-2/3}+\mathcal O(z^{-4/3}) & \text{ in } I,\\
\frac13 \log z + \ell_2 + D z^{-2/3}+\mathcal O(z^{-4/3}) & \text{ in } II \cup III,\\
\frac13 \log z + \ell_4 + D \omega^2 z^{-2/3}+\mathcal O(z^{-4/3}) & \text{ in } IV,
\end{cases}\\
\lambda_4(z) &= \theta_3(z)+\begin{cases}
\frac13 \log z + \ell_4 + D \omega^2 z^{-2/3}+\mathcal O(z^{-4/3}) & \text{ in } I \cup II,\\
\frac13 \log z + \ell_3 + D \omega z^{-2/3}+\mathcal O(z^{-4/3}) & \text{ in } III \cup IV,
\end{cases}
\end{align*}
where $D \in \R$ and $\ell_j,$ $j=1,2,3,4$, are constants satisfying
\[
\ell_3-\ell_2=\ell_2-\ell_4=\frac23 \pi i.
\]
\end{lemma}
Note that this behavior resembles very much the behavior of the original $\lambda$-functions in \cite[Lem. 4.14]{DKM}. This will be essential in Section \ref{subsec: transformation gfunctions} to perform the second transformation in the steepest descent analysis. In fact, the modified $\xi$-functions were  precisely defined in order to achieve this.

We will also need the behavior of the $\lambda$-functions around zero.
\begin{lemma}\label{lemma: asymptotics lambda around 0}
In a neighborhood of zero, we have
\begin{align}
\lambda_1(z) &= z^2K(z)+\begin{cases}
F(z)z^{1/2}+G(z)z+H(z)z^{3/2} & \text{ in } I \cup II,\\
F(z)(-z)^{1/2}-G(z)z+H(z)(-z)^{3/2}+\pi i & \text{ in } III \cup IV,
\end{cases}\label{eq:lambda1inI} \\
\lambda_2(z) &= z^2K(z)+\begin{cases}
-F(z)(-z)^{1/2}-G(z)z-H(z)(-z)^{3/2}+\pi i & \text{ in } I, \\
F(z)(-z)^{1/2}-G(z)z+H(z)(-z)^{3/2}+\pi i & \text{ in } II,\\
-F(z)z^{1/2}+G(z)z-H(z)z^{3/2} & \text{ in } III,\\
F(z)z^{1/2}+G(z)z+H(z)z^{3/2} & \text{ in } IV ,
\end{cases} \nonumber \\
\lambda_3(z) &= z^2K(z)+\begin{cases}
F(z)(-z)^{1/2}-G(z)z+H(z)(-z)^{3/2}+\pi i & \text{ in } I, \\
-F(z)(-z)^{1/2}-G(z)z-H(z)(-z)^{3/2}+\pi i & \text{ in } II, \\
F(z)z^{1/2}+G(z)z+H(z)z^{3/2} & \text{ in } III,\\
-F(z)z^{1/2}+G(z)z-H(z)z^{3/2} & \text{ in } IV,
\end{cases}\nonumber \\
\lambda_4(z) &= z^2K(z)+\begin{cases}
-F(z)z^{1/2}+G(z)z-H(z)z^{3/2} & \text{ in } I \cup II,\\
-F(z)(-z)^{1/2}-G(z)z-H(z)(-z)^{3/2}+\pi i & \text{ in } III \cup IV.
\end{cases} \nonumber
\end{align}
where $F,G,H,K$ are {even analytic functions} satisfying
\begin{align}
F(0) &= 2 e^{3\pi i/4} \gamma^{1/4} \left( -2 \gamma^2+\tfrac{1}{\gamma}+\tau^{4/3}\gamma \right),  &F^*(0)&=0, \label{eq: F} \\
G(0) &= i \gamma^{-1/2}\left( \tfrac32 \gamma^2 - \tfrac{1}{4\gamma}-\tfrac54 \tau^{4/3} \gamma \right), &G^*(0)&=0,\label{eq: G} \\
H(0) &= e^{\pi i/4} \gamma^{-5/4} \left( \tfrac12 \gamma^2- \tfrac{1}{12\gamma}+\tfrac14 \tau^{4/3} \gamma \right), &H^*(0)&=\tfrac23 e^{\pi i/4} \label{eq: H}.
\end{align}
\end{lemma}
\begin{proof} 
The function $\lambda_1(z^2)$ is clearly analytic in the first quadrant. In fact, following the analytic continuation in counterclockwise direction we see that $\lambda_1(z^2)$ can be extended to a function that is analytic in a neighborhood of the origin and vanishes at the origin. This means that there exists even analytic functions $F,G,H,K$ such that 
\begin{equation} \label{eq: decomposition lambda}
\lambda_1(z)=z^{1/2} F(z)+z G(z)+z^{3/2} H(z) + z^2 K(z)
\end{equation}
for $z$ in  $I\cup II$.  Hence we have the top equality in \eqref{eq:lambda1inI}. By analytic continuation we also obtain the expressions  for $\lambda_2(z)$ in $III$, $\lambda_3(z)$ in $IV$, $\lambda_4(z)$ in $I\cup II$, $\lambda_3$ in $III$, and finally $\lambda_2$ in $IV$. Hence we found half of the expressions in the lemma.

To get the other half of the expressions, we note that we must have $\lambda_j(z)=\lambda_j(-z)+c_0$ for some constant $c_0$, depending on $j \in \{1,2,3,4\}$ and the quadrant to which $z$ belongs. This follows from \eqref{eq: symmetry xi} and the fact that $\lambda_j$ is an anti-derivative of $\xi_j$. It is not difficult to evaluate the constants by taking the limit $z\to 0$. This gives the other half of the expressions.

Finally, the values $F(0)$, $G(0)$, and $H(0)$ can be found by comparing the asymptotics for $\xi_j(z)$ as $z\to 0$ as given in Lemma \ref{lem:asymptxinear0}.
\end{proof}

As $n\to \infty$, the functions $\lambda_j$  converge to $\lambda_j^*$, which are the $\lambda$-functions associated with the critical parameters $\alpha=-1$ and $\tau=1$. In the following lemma we give an estimate on the rate of convergence, which will be useful later on. 

\begin{lemma} \label{lemma: estimates on lambda minus lambdastar}
There exist constants $c_1,c_2>0$ such that  
\begin{align*}
|\lambda_1(z)-\lambda_1^*(z)| & \leq c_1 n^{-1/3} |z|^{1/2}, \\
|\lambda_j(z)-\lambda_j^*(z)| & \leq c_2 n^{-1/3} \max(|z|^{1/2},|z|^{4/3}), \qquad j=2,3,4,
\end{align*}
hold for $z \in \C$ and $n$ sufficiently large.
\end{lemma}
\begin{proof} We start by giving an estimate for the functions $\xi_j$. If we stay away from the singularities $z=0,\infty$ it is clear from \eqref{eq: scaling gamma}, \eqref{eq: algebraic equation for w}, and \eqref{eq: xi in terms of w} that $$|\xi_j(z)-\xi_j^*(z)|=\mathcal O(n^{-1/3})$$ uniformly as $n \to \infty$. Combining this with the asymptotic behavior of $\xi_j(z)$ as $z\to 0$ and $z\to \infty$ given in Lemma \ref{lem:asymptxinear0} and Lemma \ref{lemma: xi at infinity},  we find there exists constant $\tilde c_1$ and $\tilde c_2$ such that 
\[|\xi_1(z)-\xi_1^*(z)|\leq \tilde c_1 n^{-1/3} |z|^{-1/2},\]
and 
\[|\xi_j(z)-\xi_j^*(z)|\leq \tilde c_2 n^{-1/3} \max(|z|^{-1/2},|z|^{1/3})\]
for $z\in \C \setminus \{0\}$. The statement now follows integrating these expressions.
\end{proof}

\subsection{Modified equilibrium problem } \label{sec:modequiproblem}

Before we continue, we end this section with a discussion  on the interpretation of the modified functions. This discussion is only to clarify the motivation behind the definition and is not necessary for the upcoming analysis.

Inspired by \cite[eq. (4.27)]{DKM}, we define measures by
\begin{align} \label{eq:modeqmeasure}
\begin{cases}
\frac{\mathrm d \mu_1}{\mathrm d x}(x) = \frac{1}{2\pi i} \left(\xi_{1,+}(x)-\xi_{1,-}(x)\right), & x \in [-c,c], \\
\frac{\mathrm d \mu_2}{\mathrm d z}(z) = \frac{1}{2\pi i} \left(\xi_{2,+}(z)-\xi_{2,-}(z)\right)-\frac{1}{2\pi i} \left(\theta_{1,+}'(z)-\theta_{1,-}'(z)\right), & z \in i\R, \\
\frac{\mathrm d \mu_3}{\mathrm d x}(x) = \frac{1}{2\pi i} \left(\xi_{3,+}(x)-\xi_{3,-}(x)\right)-\frac{1}{2\pi i} \left(\theta_{2,+}'(x)-\theta_{2,-}'(x)\right), & x \in \R.
\end{cases}
\end{align}
Although \eqref{eq:modeqmeasure} is exactly the same as \cite[eq. (4.27)]{DKM}, the measures here are different because our definition of $\xi$ is different. Only in the critical case $\alpha=-1$ and $\tau=1$,  the definitions agree.  In that situation, the vector of measures $(\mu_1,\mu_2,\mu_3)$ is the unique solution to the equilibrium problem in Section \ref{sec: limiting mean density}. See Figure \ref{fig: densities critical equilibrium measures} for a plot of these critical measures. For other values of $\alpha$ and $\tau$, the measures defined by \eqref{eq:modeqmeasure} do not solve that equilibrium problem, but a modification of it that we will briefly discuss.

Let us first make some observations. It can be checked that the total masses of $\mu_j$, $j=1,2,3,$ are given by 
\begin{align*}
\mu_1([-c,c])=1, \qquad \mu_2(i\R)=\frac23, \qquad \textrm{ and } \mu_3(\R)=\frac13, 
\end{align*}
e.g. for $\mu_1$ this follows from Lemma \ref{lemma: integrals}. So the measures still satisfy conditions (b), (c), and (d) of the original equilibrium problem. On the other hand, these  measures are not necessarily positive. In a neighborhood of zero, and only there, they can have negative densities. The measure $\mu_2$ can also break the constraint locally around zero, violating condition (e) of the original equilibrium problem. Moreover, generically the densities in \eqref{eq:modeqmeasure} blow up around zero, to either plus or minus infinity. 

Without proof we claim that $(\mu_1,\mu_2,\mu_3)$ solves the following modified equilibrium problem. We seek to minimize the energy functional \eqref{eq:energyfunctional} under the conditions (a), (b), (c), and (d), but in a small neighborhood of the origin we allow the measures to be negative and $\mu_2$ to break the constraint $\sigma_2$.  These modifications lead to violations of the original equilibrium problem, and as a result they will cause trouble when opening the lenses in the steepest descent analysis. However, these problems only take place near the origin, around which we are going to make a special parametrix anyway.

\begin{figure}[t]
\centering
\begin{tikzpicture}[xscale=0.7,yscale=0.7]
\begin{scope}[shift={(9,0)}]
\draw (0,2.9) node {\includegraphics[width=112pt]{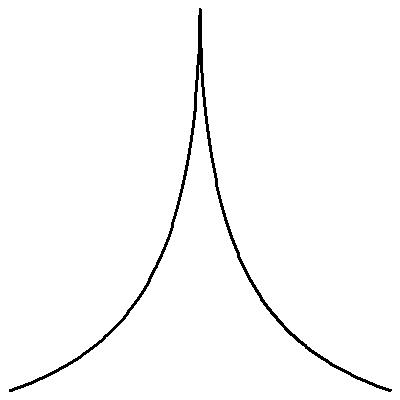}};
\draw[->] (-3,0)--(3,0) node[right]{$i\R$};
\draw[->] (0,0)--(0,6);
\draw (1,4) node {$\mu_2^*$};
\fill (0,5.5) circle (2pt) node[left]{$\frac{1}{\pi}$};
\end{scope}

\draw (0,2.5) node {\includegraphics[width=140pt]{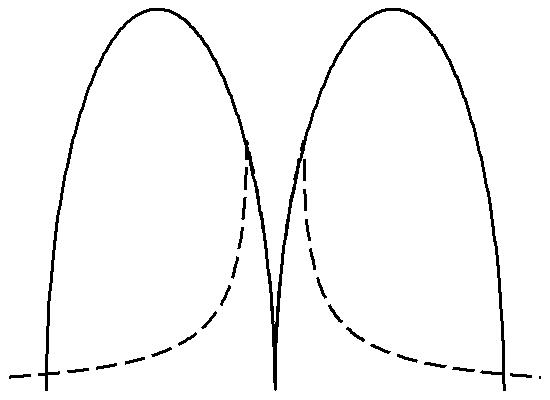}};
\draw[->] (-4,0)--(4,0) node[right]{$\R$};
\draw[->] (0,0)--(0,6);
\draw[help lines] (0.35,3)--(0.35,0)  (-0.35,3)--(-0.35,0);
\draw (3,4.5)node{$\mu_1^*$}
      (1,2.5)node{$\mu_3^*$};
\fill (0.35,0) circle (2pt) node[below]{$x^*$}
      (-0.35,0) circle (2pt) node[below]{$-x^*$}
      (2.9,0) circle (2pt) node[below]{$c^*$}
      (-2.9,0) circle (2pt) node[below]{$-c^*$};                   
\end{tikzpicture}
\caption{Densities of the equilibrium measures $\mu_1^*$, $\mu_2^*$ (solid), and $\mu_3^*$ (dashed) for critical values $\alpha=-1$, $\tau=1$.}
\label{fig: densities critical equilibrium measures}
\end{figure}
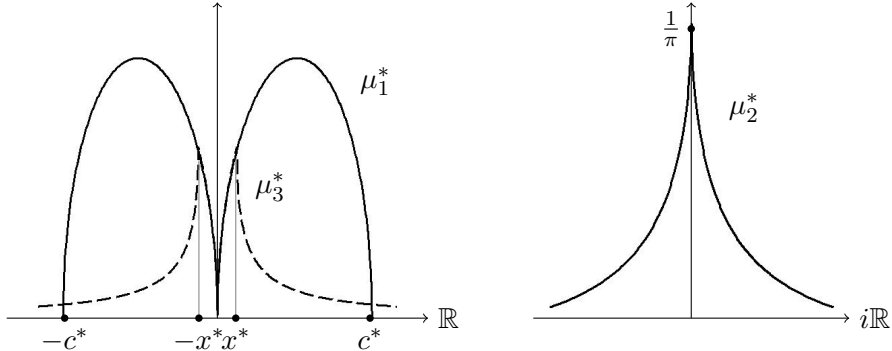

\section{Riemann-Hilbert steepest descent analysis}

In this section we prove Theorem \ref{th: main theorem}. The proof starts with a characterization of the kernel $K_n$ in terms  of a $4\times 4$ RH problem. We then perform a Deift/Zhou steepest descent analysis on the RH problem to obtain the asymptotic behavior near the critical point.  

Many definitions and transformations are analogous to the ones in \cite{DKM}. For this reason we will allow us to be brief at some points and simply refer to \cite{DKM} for more details. The most important differences between the analysis here and the one given in \cite{DKM}, are the use of the modified $\lambda$-functions from the previous section and the special parametrix near the origin based on RH problem \ref{rhp: tacnode rhp}.

\subsection{Riemann-Hilbert problem for $Y$}

We start by expressing the kernel $K_n$ in \eqref{eq:defKn} in terms of the solution to a RH problem.  This RH problem first appeared in \cite{KMcL} to characterize the biorthogonal polynomials in  \eqref{eq:biorthogonal}.

\begin{lemma} \label{rhp: Y}
Suppose $ n \equiv 0 \mod 3$. There exists a unique $4 \times 4$ matrix-valued function $Y(z)$, also depending on the parameters $\alpha$ and $\tau$, satisfying
\begin{itemize}
\item[\rm (1)] $Y(z)$ is analytic for  $z \in \C \setminus \R$;
\item[\rm (2)] for $x \in \R$ the function $Y$ makes the jump
\[
Y_+(x)=Y_-(x) \begin{pmatrix} 1 & w_{0,n}(x) & w_{1,n}(x) &w_{2,n}(x) \\ 0&1&0&0 \\ 0&0&1&0 \\ 0&0&0&1 \end{pmatrix}, \qquad x \in \R,
\]
where \[
w_{j,n}(x)= e^{-nV(x)}\int_{-\infty}^\infty y^j e^{-n(W(y)-\tau x y)} \ud y, \qquad j=0,1,2;
\]
\item[\rm (3)] as $z \to \infty$, we have
\[
Y(z)=\left( I+\mathcal O(1/z) \right) \diag \left(z^n,z^{-n/3},z^{-n/3},z^{-n/3}\right).
\]
\end{itemize}
 Moreover, the kernel $K_n$ is given by
\begin{equation} \label{eq: K in terms of Y}
K_n(x,y)=\frac{1}{2\pi i (x-y)}\begin{pmatrix} 0 & w_{0,n}(y)  & w_{1,n}(y) & w_{2,n}(y) \end{pmatrix} Y_+^{-1}(y)Y_+(x) \begin{pmatrix} 1\\0\\0\\0 \end{pmatrix}.
\end{equation}

\end{lemma}
\begin{proof}
For details see \cite[Sec. 1.7]{DKM} and the references therein.
\end{proof}

In the rest of this section we  perform a Deift/Zhou steepest descent analysis on $Y$, i.e.  a series of explicit transformations
\[
Y \mapsto X \mapsto U \mapsto T \mapsto S \mapsto R,
\]
such that the final RH problem for $R$ is simple in the sense  that all jump matrices tend to the identity matrix as $n \to \infty$, both uniformly and in $L^2$-sense, and the behavior at infinity is normalized.  Then, by standard considerations, we obtain an asymptotic expansion for $R$ as $n \to \infty$. By tracing back the transformations, we then obtain the asymptotic behavior of the kernel $K_n$. 

The successive transformations are: a preliminary transformation (exactly the same as in \cite{DKM}), normalization at infinity by means of the modified $\lambda$-functions that were established in Section \ref{sec: Riemann surface}, opening of unbounded lenses, opening of local lenses, and finally the matching with local and global parametrices. 

\subsection{First transformation: $Y \mapsto X$}

The first transformation  $Y \mapsto X$ is identical to the first transformation in \cite[Sec. 5]{DKM}. For that reason, we will not go into details here but only state the new RH problem for $X$ and the formula that expresses the kernel in terms of the solution to this RH problem.

\begin{lemma} \label{rhp: X}
Suppose that $n \equiv 0 \mod 3$. Let $\theta_j$ be as defined in Section \ref{sec: theta}. There exists a  unique $4 \times 4$ matrix-valued function $X(z)$, also depending on the parameters $\alpha<0$ and $\tau$, satisfying
\begin{itemize}
\item[\rm (1)] $X(z)$ is analytic for  $z \in \C \setminus (\R \cup i\R)$;
\item[\rm (2)] $X_+(z)=X_-(z) J_X$, for $z \in \R \cup i\R$, where the matrices $J_X$ are given as follows
\begin{itemize}
\item[$\bullet$] On the real line we have for $x \in (-\infty,-x^*(\alpha)]\cup [x^*(\alpha),\infty)$
\[
J_X= \begin{pmatrix} 1 & e^{-n(x^2/2-\theta_1(x))} & 0 & 0 \\ 0&1&0&0 \\ 0&0&e^{n(\theta_{2,+}-\theta_{3,+})(x)} &1 \\ 0&0&0&e^{n(\theta_{3,+}-\theta_{2,+})(x)} \end{pmatrix},
\]
and for $x \in (-x^*(\alpha),x^*(\alpha))$
\[
J_X= \begin{pmatrix} 1 & e^{-n(x^2/2-\theta_1(x))} & 0 & 0 \\ 0&1&0&0 \\ 0&0&1&e^{-n(\theta_{2}-\theta_{3})(x)}  \\ 0&0&0&1 \end{pmatrix}.
\]
See \eqref{eq: def xstar} for the definition of $x^*(\alpha)$.
\item[$\bullet$] On the imaginary axis we have
\[
J_X=\begin{pmatrix} 1 & 0 & 0 & 0 \\ 0&e^{n(\theta_{1,+}-\theta_{2,+})(z)}&0&0\\0&1&e^{n(\theta_{1,-}-\theta_{2,-})(z)}&0\\0&0&0&1 \end{pmatrix}, \qquad z \in i\R.
\]
\end{itemize}
\item[\rm (3)] As $z \to \infty$ in the $j$-th quadrant, we have
\[
X(z)=\left( I+\mathcal O(z^{-2/3})\right)\diag (z^n,z^{-\frac{n-1}{3}},z^{-\frac{n}{3}},z^{-\frac{n+1}{3}}) \begin{pmatrix} 1 & 0 \\ 0 & A_j \end{pmatrix},
\]
where the matrices $A_j$ are given by
\begin{align}
A_1 &=  \frac{-i}{\sqrt 3}\begin{pmatrix} -1 & \omega & \omega^2 \\ -1 & 1 & 1 \\ -1 & \omega^2 & \omega \end{pmatrix}, &
A_2 &=  \frac{-i}{\sqrt 3}\begin{pmatrix} \omega & 1 & \omega^2 \\ 1 & 1& 1 \\ \omega^2 & 1 & \omega \end{pmatrix}, \label{eq: A12}\\
A_3 &=  \frac{i}{\sqrt 3} \begin{pmatrix} -\omega^2 & -1 & \omega \\ -1 & -1& 1 \\ -\omega & -1 & \omega^2 \end{pmatrix}, &
A_4 &=  \frac{i}{\sqrt 3} \begin{pmatrix} 1 & -\omega^2 & \omega \\ 1&-1&1 \\ 1&-\omega&\omega^2 \end{pmatrix}, \label{eq: A34}
\end{align}
and $\omega= e^{2 \pi i /3}$.
\item[\rm (4)] $X(z)$ remains bounded as $z\to 0$. 
\end{itemize}
Moreover, the kernel \eqref{eq: K in terms of Y} is given by 
\begin{equation} \label{eq: K in terms of X}
K_n(x,y)=\frac{1}{2\pi i (x-y)}\begin{pmatrix} 0 & e^{-n(y^2/2-\theta_1(y))}  & 0 & 0 \end{pmatrix} X_+^{-1}(y)X_+(x) \begin{pmatrix} 1\\0\\0\\0 \end{pmatrix}.
\end{equation}
\end{lemma}
\begin{proof}
The uniqueness for $X$ is standard and the solution can be constructed out of $Y$ by the transformation $Y\mapsto X$ in \cite[Sec. 5]{DKM}. For \eqref{eq: K in terms of X} see \cite[eq. (9.28)]{DKM}.
\end{proof}

\subsection{Second transformation: $X \mapsto U$} \label{subsec: transformation gfunctions}

In this transformation the modified $\lambda$-functions of Section 3 come into play to normalize the behavior at infinity. As in \cite[eq. (6.1)]{DKM}, but now with the modified $\lambda$-functions, we define
\begin{equation} \label{eq: X to U}
U(z)=(I-nDE_{2,4})e^{-nL(z)}X(z)e^{n\Lambda(z)}\begin{pmatrix} e^{-nz^2/2} & 0 \\ 0 & e^{-n \Theta(z)}\end{pmatrix}, \qquad z \in \C \setminus (\R \cup i\R),
\end{equation}
where
\begin{align*}
L(z) &= \begin{cases}
\diag(\ell_1,\ell_2,\ell_3,\ell_4) & \text{ for } z \text{ in } I,   \\
\diag(\ell_1,\ell_3,\ell_2,\ell_4) & \text{ for } z \text{ in } II,  \\
\diag(\ell_1,\ell_4,\ell_2,\ell_3) & \text{ for } z \text{ in } III, \\
\diag(\ell_1,\ell_2,\ell_4,\ell_3) & \text{ for } z \text{ in } IV,
\end{cases}, \\
\Lambda(z)&=\diag \left( \lambda_1(z),\lambda_2(z),\lambda_3(z),\lambda_4(z)\right),
\end{align*}
and
\[
\Theta(z)= \diag \left( \theta_1(z),\theta_2(z),\theta_3(z)\right).
\]
Here, $D$ is the constant from Lemma \ref{lemma: asymptotics lambda around infinity} and $E_{2,4}$ is the $4 \times 4$ matrix with a 1 in the $(2,4)$ entry and zeros in all other entries. 

Then $U$ satisfies the following RH problem.
\begin{lemma} \label{rhp: U}
Suppose that $n \equiv 0 \mod 3$. Then $U$ as defined in \eqref{eq: X to U}  has the following properties \begin{itemize}
\item[\rm (1)] $U(z)$ is analytic for  $z \in \C \setminus (\R \cup i\R)$;
\item[\rm (2)] $U_+(z)=U_-(z) J_U$, for $z \in \R \cup i\R$, where the matrices $J_U$ are specified as follows.
\begin{itemize}
\item[$\bullet$] On the real line we have for $x \in (-c,c)$
\[
J_U(x)=\begin{pmatrix} e^{n(\lambda_{1,+}-\lambda_{1,-})(x)} & 1 & 0 & 0 \\ 0 & e^{n(\lambda_{2,+}-\lambda_{2,-})(x)}&0&0\\0&0&e^{n(\lambda_{3,+}-\lambda_{3,-})(x)} & 1\\0&0&0& e^{n(\lambda_{4,+}-\lambda_{4,-})(x)} \end{pmatrix},
\]
and for $x \in (-\infty,-c] \cup [c,\infty)$
\[
J_U(x)=\begin{pmatrix} 1 & e^{n(\lambda_{2,+}-\lambda_{1,-})(x)} & 0 & 0 \\ 0&1&0&0\\ 0&0&e^{n(\lambda_{3,+}-\lambda_{3,-})(x)} & 1\\0&0&0 & e^{n(\lambda_{4,+}(x)-\lambda_{4,-})(x)}\end{pmatrix}.
\]
\item[$\bullet$] On the imaginary axis $z\in i \R$ we have
\[
J_U(z)=\begin{pmatrix} 1 & 0 & 0 & 0 \\ 0&e^{n(\lambda_{2,+}-\lambda_{2,-})(z)} & 0 &0\\0&1 & e^{n(\lambda_{3,+}-\lambda_{3,-})(z)}&0\\0&0&0&1 \end{pmatrix}.
\]
\end{itemize}
\item[\rm (3)] As $z \to \infty$ in the $j$-th quadrant, we have
\[
U(z)=\left( I+\mathcal O(z^{-1/3})\right)\diag (1,z^{1/3},1,z^{-1/3}) \begin{pmatrix} 1 & 0 \\ 0 & A_j \end{pmatrix},
\]
with $A_j$, $j=1,2,3,4$, as in \eqref{eq: A12}--\eqref{eq: A34}.
\end{itemize}
The jump contour and the jump matrices  are shown in Figure \ref{fig: contour for U}. 
\end{lemma}
\begin{proof}
The proof follows by straightforward computations and is similar to the analysis in  \cite[Sec. 6]{DKM}.
\end{proof}
 
Note that by  Lemma \ref{lemma: lambda meromorf} we have that the diagonal entries of the jump matrices $J_U$ are either $1$ or highly oscillating. In the next transformations we will open lenses  to replace the jump matrices containing oscillating entries by jumps for which the jump matrices are either constant or exponentially close to the identity as $n\to \infty$.

\begin{figure}[t]
\centering
\small{
\begin{tikzpicture}
\draw[->] (-6,0)--(6,0) node[right]{$\R$};
\draw[->] (0,-4)--(0,4) node[above]{$i\R$};
\filldraw (-3,0) circle (1pt) node[below]{$-c$};
\filldraw (3,0) circle (1pt) node[below]{$c$};
\draw (5,2) node{$\begin{pmatrix} e^{n(\lambda_{1,+}-\lambda_{1,-})} & 1 & 0 & 0 \\ 0 & e^{n(\lambda_{2,+}-\lambda_{2,-})}&0&0\\0&0&e^{n(\lambda_{3,+}-\lambda_{3,-})} & 1\\0&0&0& e^{n(\lambda_{4,+}-\lambda_{4,-})} \end{pmatrix}$};
\draw[->] (2.3,1)--(1,0); \draw[->] (2.3,1)--(-1.5,0);
\draw (-3.5,2) node{$\begin{pmatrix} 1 & 0 & 0 & 0 \\ 0&e^{n(\lambda_{2,+}-\lambda_{2,-})} & 0 &0\\0&1 & e^{n(\lambda_{3,+}-\lambda_{3,-})}&0\\0&0&0&1 \end{pmatrix}$};
\draw[->] (-1,1)--(0,1); \draw[->] (-1,1)--(0,-1);
\draw (4.5,-2) node{$\begin{pmatrix} 1 & e^{n(\lambda_{2,+}-\lambda_{1,-})} & 0 & 0 \\ 0&1&0&0\\ 0&0&e^{n(\lambda_{3,+}-\lambda_{3,-})} & 1\\0&0&0 & e^{n(\lambda_{4,+}-\lambda_{4,-})}\end{pmatrix}$};
\draw[->] (4,-1)--(4,0); \draw[->] (0.2,-2)--(-5,0);
\end{tikzpicture}}
\caption{Jump contour for $U$, consisting of the axes $\R$ and $i\R$, shown together with their respective jumps.}
\label{fig: contour for U}
\end{figure}
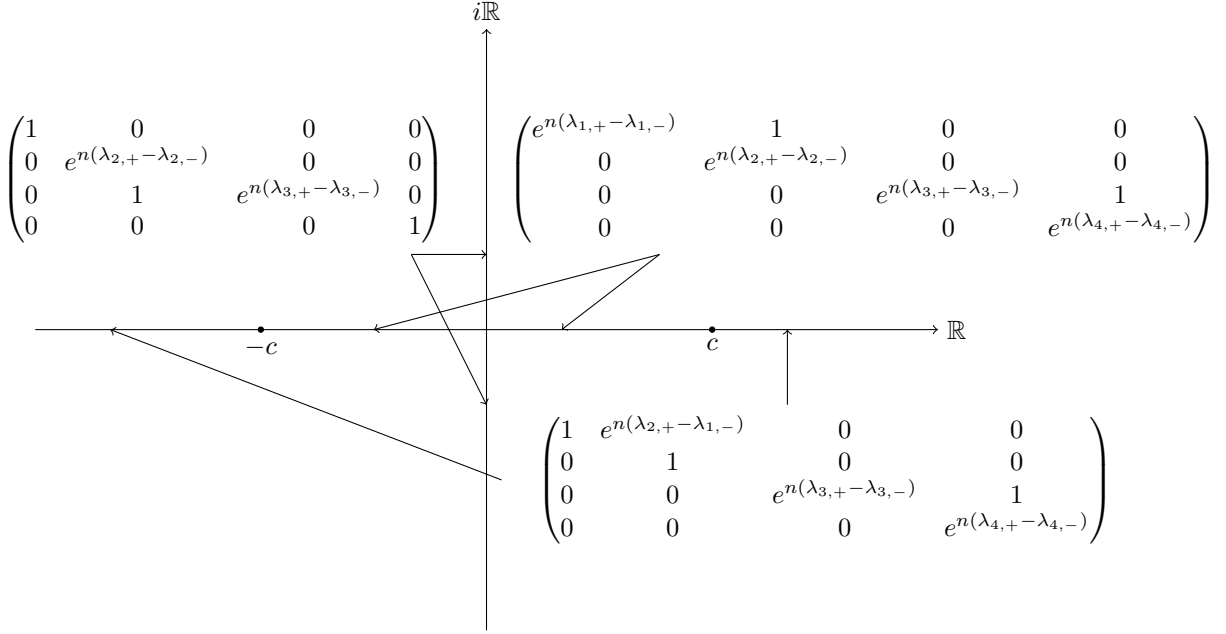

\subsection{Third transformation $U \mapsto T$: opening of unbounded lenses }

 In this section we deal with the oscillatory entries of the jump matrix on $i\R$ and in the lower right block of the jump matrix on $\R$. We postpone the handling of the jump in the upper left block of the jump matrix on the compact set $[-c,c]$ until the next section.

We open unbounded lenses around $i\R$ and $\R$ that we denote with $L_2$ and $L_3$ respectively.  The lips of both lenses emanate from zero and do not intersect elsewhere. The lenses are symmetric with respect to both the real and imaginary axes. The unbounded lenses are shown in Figure \ref{fig: global lenses}. We will be more precise on their definition later on. The transformation $U \mapsto T$ and the RH problem for $T$ are essentially the same as in \cite[Sec. 7.1]{DKM}, with the only difference that here we use the modified $\lambda$-functions. 

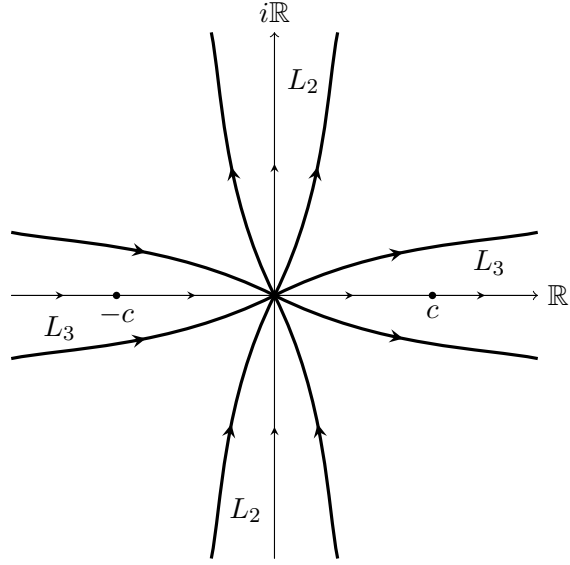
\begin{figure}
\centering
\begin{tikzpicture}[scale=0.7]
\draw[->] (-5,0)--(5,0) node[right]{$\R$};
\draw[->]	(0,-5)--(0,5) node[above]{$i\R$};
\begin{scope}[decoration={markings,mark= at position 0.5 with {\arrow{stealth}}}]
\draw[postaction={decorate}]   (-5,0)--(-3,0);
\draw[postaction={decorate}]   (-3,0)--(0,0);
\draw[postaction={decorate}]   (0,0)--(3,0);
\draw[postaction={decorate}]   (3,0)--(5,0);
\draw[postaction={decorate}]   (0,-5)--(0,0);  
\draw[postaction={decorate}]   (0,0)--(0,5);  
\draw[very thick, postaction={decorate}]   (0,0) .. node[left,near end]{$L_2$} controls (1,2) and (1,4) .. (1.2,5);
\draw[very thick, postaction={decorate}]   (0,0) .. controls (-1,2) and (-1,4) .. (-1.2,5);
\draw[very thick, postaction={decorate}]   (1.2,-5) .. controls (1,-4) and (1,-2) .. (0,0);
\draw[very thick, postaction={decorate}]   (-1.2,-5) ..node[right,near start]{$L_2$} controls (-1,-4) and (-1,-2) .. (0,0);
\draw[very thick, postaction={decorate}]   (0,0) .. node[below,near end]{$L_3$} controls (2,1) and (4,1) .. (5,1.2);
\draw[very thick, postaction={decorate}]   (0,0) .. controls (2,-1) and (4,-1) .. (5,-1.2);
\draw[very thick, postaction={decorate}]   (-5,1.2) .. controls (-4,1) and (-2,1) .. (0,0);
\draw[very thick, postaction={decorate}]   (-5,-1.2) ..node[above,near start]{$L_3$} controls (-4,-1) and (-2,-1) .. (0,0);
\end{scope}
\fill (3,0) circle (2pt) node[below]{$c$}
			(-3,0) circle (2pt) node[below]{$-c$};
\fill (0,0) circle (2pt);
\end{tikzpicture}
\caption{The jump contour $\Sigma_T$ after opening unbounded lenses around the real and imaginary axes.}
\label{fig: global lenses}
\end{figure}

The definition of $T$ is based on the diagonalization of blocks in $J_U$, as explained in \cite[Sec. 7.1]{DKM}, and given by
\begin{align}\label{eq:defT1}
T=U \times \begin{cases}
\begin{pmatrix}
1 &0&0&0 \\
0& 1 & -e^{n(\lambda_3-\lambda_2)} &0 \\
0&0&1&0\\
0&0&0&1
\end{pmatrix} & \textrm{\parbox[t]{0.3\textwidth}{in the left part of the lens around $i\R$,}} \\
\begin{pmatrix}
1 &0&0&0 \\
0& 1 & e^{n(\lambda_3-\lambda_2)} & 0\\
0&0&1&0\\
0&0&0&1
\end{pmatrix} & \textrm{\parbox[t]{0.3\textwidth}{in the right part of the lens around $i\R$,}} \\
\end{cases}
\end{align}
\begin{align} \label{eq:defT2}
T=U \times \begin{cases}
\begin{pmatrix}
1 &0&0&0 \\
0& 1 & 0 & 0\\
0&0&1&0\\
0&0&-e^{n(\lambda_3-\lambda_4)}&1
\end{pmatrix} & \textrm{\parbox[t]{0.3\textwidth}{in the upper part of the lens around $\R$,}} \\
\begin{pmatrix}
1 &0&0&0 \\
0& 1 & 0 & 0\\
0&0&1&0\\
0&0&e^{n(\lambda_3-\lambda_4)}&1
\end{pmatrix} & \textrm{\parbox[t]{0.3\textwidth}{in the lower part of the lens around $\R$,}} \\
\end{cases}
\end{align} and
\begin{align}\label{eq:defT3}
T=U \quad \textrm{ elsewhere.}
\end{align}
\begin{lemma}
Suppose that $n \equiv 0 \mod 3$. Then $T$ has the following properties \begin{itemize}
\item[\rm (1)] $T$ is analytic in $\C \setminus \Sigma_T$, where $\Sigma_T$ is the contour consisting of the real and imaginary axes and the lips  of the unbounded lenses ; as shown in Figure \ref{fig: global lenses};
\item[\rm (2)] $T_+=T_-J_T$ on $\Sigma_T$ where $J_T$ is specified in \eqref{eq: jump for T structure}--\eqref{eq: jump for T 3};
\item[\rm (3)] $\displaystyle{ T(z)=(I+\mathcal O(z^{-1/3})) \diag \left( 1 , z^{1/3} , 1 , z^{-1/3} \right) \begin{pmatrix} 1 & 0 \\ 0 & A_j \end{pmatrix}}$ as $z \to \infty$ in the $j$-th quadrant.
\end{itemize}
The jump matrices have the following structure
\begin{equation} \label{eq: jump for T structure}
J_T=\begin{cases}
\begin{pmatrix}
(J_T)_1 & 0 \\
0 & (J_T)_3
\end{pmatrix} & \textrm{\parbox[t]{0.4\textwidth}{on $\R$ and on the lips of the lens around $\R$,}} \\
\begin{pmatrix}
1 & 0 & 0 \\
0 & (J_T)_2 & 0 \\
0 & 0 & 1
\end{pmatrix} & \textrm{\parbox[t]{0.4\textwidth}{on $i\R$ and on the lips of the lens around $i\R$,}} \\
\end{cases}
\end{equation}
where the $2 \times 2$ blocks are given by
\begin{align}
(J_T)_1 & =\begin{cases}
\begin{pmatrix}e^{n(\lambda_{1,+}-\lambda_{1,-})} & e^{n(\lambda_{2,+}-\lambda_{1,-})} \\ 0 & e^{n(\lambda_{2,+}-\lambda_{2,-})}  \end{pmatrix} & \textrm{on $\R$,} \\
I_2 & \textrm{\parbox[t]{0.3\textwidth}{on the lips of the lens around $\R$,}}
\end{cases} \label{eq: jump for T 1}\end{align}
\begin{align}
(J_T)_2 & =\begin{cases}
\begin{pmatrix} 0 & -1 \\ 1 & 0 \end{pmatrix} & \textrm{on $i\R$,} \\
\begin{pmatrix} 1 & e^{n(\lambda_3-\lambda_2)} \\ 0 & 1 \end{pmatrix} & \textrm{\parbox[t]{0.5\textwidth}{on the lips of the lens around $i\R$,}}
\end{cases} \label{eq: jump for T 2}\end{align} \begin{align}
(J_T)_3 & =\begin{cases}
\begin{pmatrix} 0 & 1 \\ -1 & 0 \end{pmatrix} & \textrm{on $\R$,} \\
\begin{pmatrix} 1 & 0 \\ e^{n(\lambda_3-\lambda_4)} & 1 \end{pmatrix} & \textrm{\parbox[t]{0.5\textwidth}{on the lips of the lens around $\R$.}}
\end{cases} \label{eq: jump for T 3}
\end{align}
\end{lemma}
\begin{proof}
The proof follows by straightforward computations and is similar to \cite[Sec. 7.1--7.3]{DKM}. 
\end{proof}
The precise location of the lips of the lens is given in the following lemma. 
\begin{lemma} \label{lemma: estimates L2L3star}
We can and do choose the unbounded lenses $L_2$ around $i\R$ and $L_3$ around $\R$ such that
\begin{align}
\Re(\lambda^*_3-\lambda^*_2) & \leq -c_3\min \left(|z|^{3/2},|z|^{4/3}\right) && \text{for $z$ on the lips of $L_2$}, \label{eq: estimate lambda 32 star} \\
\Re(\lambda^*_3-\lambda^*_4) & \leq -c_3\min \left(|z|^{3/2},|z|^{4/3}\right) && \text{for $z$ on the lips of $L_3$}, \label{eq: estimate lambda 34 star}
\end{align}
for some constant $c_3>0$.
\end{lemma}
\begin{proof}
We only prove the statement for the unbounded lens $L_3$ around the real line. The other estimate follows similarly.

For $x \in \R \setminus \{0\}$ we know, see \eqref{eq: sym lambda 5}--\eqref{eq: sym lambda 6} and \eqref{eq: jumps lambda 6}--\eqref{eq: jumps lambda 7}, that $\Re(\lambda^*_{3,\pm}(x)-\lambda^*_{4,\pm}(x))=0$. Then it is clear from the Cauchy-Riemann equations and Lemma \ref{lemma: xi on cuts} that the lens can be opened such that
\[
\Re(\lambda^*_3-\lambda^*_4)<0, \qquad z \in L_3 \setminus \R.
\]
The fact that also the stronger estimate \eqref{eq: estimate lambda 34 star} holds is due to the asymptotics in Lemma \ref{lemma: asymptotics lambda around infinity}, Lemma \ref{lemma: asymptotics lambda around 0}, and \eqref{eq: asym theta 1}--\eqref{eq: asym theta 3}.
\end{proof}

Concluding, we have chosen the lips of the lenses such that, for the critical case $\alpha=-1$ and $\tau=1$, the off-diagonal entries for the jump matrices $J_T$ on the lips $L_2$ and $L_3$ are exponentially small as $n$ tends to infinity. However, we also like  this to be true in the non-critical situation. This is too much to ask for, but it does hold  away from the origin as we will show. 

We exclude a shrinking  disk
\[
\D = \left\{z \in \C \mid |z|< n^{-\delta} \right\}
\] 
centered at 0 with radius $n^{-\delta}>0$ and 
\begin{equation} \label{eq: bounds kappa}
\frac16<\delta<\frac13.
\end{equation}
For the purpose of opening lenses, it will be crucial that the excluded disk $\D$ is large enough and for this reason we need $\delta <\tfrac13$. However, when we construct the local parametrix in $\D$ we will need that the disk is shrinking with a sufficiently large rate and therefore we also need $\delta>\tfrac16$. The precise value $\tfrac16$ of the lower bound will be motivated later on, however the inequality $\delta<\tfrac13$ yields the  following estimate. 
\begin{lemma} \label{lemma: estimate global lenses}
There exists a constant $c_4>0$ such that
\begin{align}
\Re(\lambda_3(z)-\lambda_2(z)) & \leq -c_4n^{-1/2} \max( 1,|z|^{4/3}), \qquad \text{\parbox[t]{.2 \textwidth}{for $z$ on the lips of $L_2$ but outside $\D$,}} \label{eq: estimate 2}\\
\Re(\lambda_3(z)-\lambda_4(z)) & \leq -c_4 n^{-1/2} \max(1,|z|^{4/3}), \qquad  \text{\parbox[t]{.2 \textwidth}{for $z$ on the lips of $L_3$ but outside $\D$,}} \label{eq: estimate 3}
\end{align}
for sufficiently large $n$. 
\end{lemma}

\begin{proof}
We prove \eqref{eq: estimate 3}. Choose $z$ on the lips of $L_3$ but outside $\D$. First assume $|z|\geq 1$. Combining Lemma \ref{lemma: estimates on lambda minus lambdastar} with Lemma \ref{lemma: estimates L2L3star} gives
\[
\Re(\lambda_3(z)-\lambda_4(z)) \leq 2c_2n^{-1/3}|z|^{4/3}-c_3|z|^{4/3}.
\]
Hence, there exists a constant $\tilde c_1>0$ such that for sufficiently large $n$
\[
\Re(\lambda_3(z)-\lambda_4(z)) \leq -n^{-1/2}\tilde c_1 |z|^{4/3},
\]
which proves \eqref{eq: estimate 3} for $|z|\geq 1$. Next suppose that $n^{-\delta} \leq |z| \leq 1$. Then we obtain
\[
\Re(\lambda_3(z)-\lambda_4(z)) \leq |z|^{1/2}\left(2c_2n^{-1/3}-c_3|z|\right).
\]
Now using that $|z| \geq n^{-\delta}$ and $\delta<1/3$ we find that
\[
\Re(\lambda_3(z)-\lambda_4(z)) \leq -\tilde c_2 n^{-\delta}|z|^{1/2},
\]
for a certain constant $\tilde c_2>0$ and large $n$. Applying the inequality $|z| \geq n^{-\delta} $ once again finishes the proof of \eqref{eq: estimate 3}. The proof of \eqref{eq: estimate 2} goes analogously. 
\end{proof}

\subsection{Fourth transformation $T \mapsto S$: opening of local lenses}

In this section we open the local lens $L_1$ around $[-c,c]$. We let four lips emanate from zero, coincide with the lips of the unbounded lens $L_3$ in a small fixed disk $D(0,\rho_0)$ around the origin, and end in $\pm c$ going through the upper/lower half plane. This is shown in Figure \ref{fig: local lenses}. We will be more precise on the definition of $L_1$ later.

\begin{figure}[t]
\centering
\begin{tikzpicture}
\draw[->] (-5,0)--(5,0) node[right]{$\R$};
\draw[->]	(0,-5)--(0,5) node[above]{$i\R$};
\begin{scope}[decoration={markings,mark= at position 0.5 with {\arrow{stealth}}}]
\draw[postaction={decorate}]   (-5,0)--(-3,0);
\draw[postaction={decorate}]   (-3,0)--(0,0);
\draw[postaction={decorate}]   (0,0)--(3,0);
\draw[postaction={decorate}]   (3,0)--(5,0);
\draw[postaction={decorate}]   (0,-5)--(0,0);  
\draw[postaction={decorate}]   (0,0)--(0,5);  
\draw[postaction={decorate}]   (0,0) ..  controls (1,2) and (1,4) .. (1.2,5);
\draw[postaction={decorate}]   (0,0) .. controls (-1,2) and (-1,4) .. (-1.2,5);
\draw[postaction={decorate}]   (1.2,-5) .. controls (1,-4) and (1,-2) .. (0,0);
\draw[postaction={decorate}]   (-1.2,-5) .. controls (-1,-4) and (-1,-2) .. (0,0);
\draw[postaction={decorate}]   (0,0) ..  controls (2,1) and (4,1) .. (5,1.2);
\draw[postaction={decorate}]   (0,0) .. controls (2,-1) and (4,-1) .. (5,-1.2);
\draw[postaction={decorate}]   (-5,1.2) .. controls (-4,1) and (-2,1) .. (0,0);
\draw[postaction={decorate}]   (-5,-1.2) .. controls (-4,-1) and (-2,-1) .. (0,0);
\draw[very thick, postaction={decorate}] (0,0) ..node[midway, below]{$L_1$} controls (1,0.5) and (2,0.5) .. (3,0);
\draw[very thick, postaction={decorate}] (0,0) .. controls (1,-0.5) and (2,-0.5) .. (3,0);
\draw[very thick, postaction={decorate}] (-3,0) ..node[midway, below]{$L_1$} controls (-2,0.5) and (-1,0.5) .. (0,0);
\draw[very thick, postaction={decorate}] (-3,0) .. controls (-2,-0.5) and (-1,-0.5) .. (0,0);
\end{scope}
\draw[dashed] (0,0) circle (0.5);
\fill (3,0) circle (2pt) node[below]{$c$}
			(-3,0) circle (2pt) node[below]{$-c$};
\fill (0,0) circle (2pt);
\end{tikzpicture}
\caption{Jump contour $\Sigma_S$ for $S$ consisting of the lips of the local lenses around $[-c,c]$, the lips of the global lenses and the axes. Note that the lips of the local lens coincide with the lips $L_3$ of the unbounded lens in a small fixed disk $D(0,\rho_0)$ around zero (dashed).}
\label{fig: local lenses}
\end{figure}
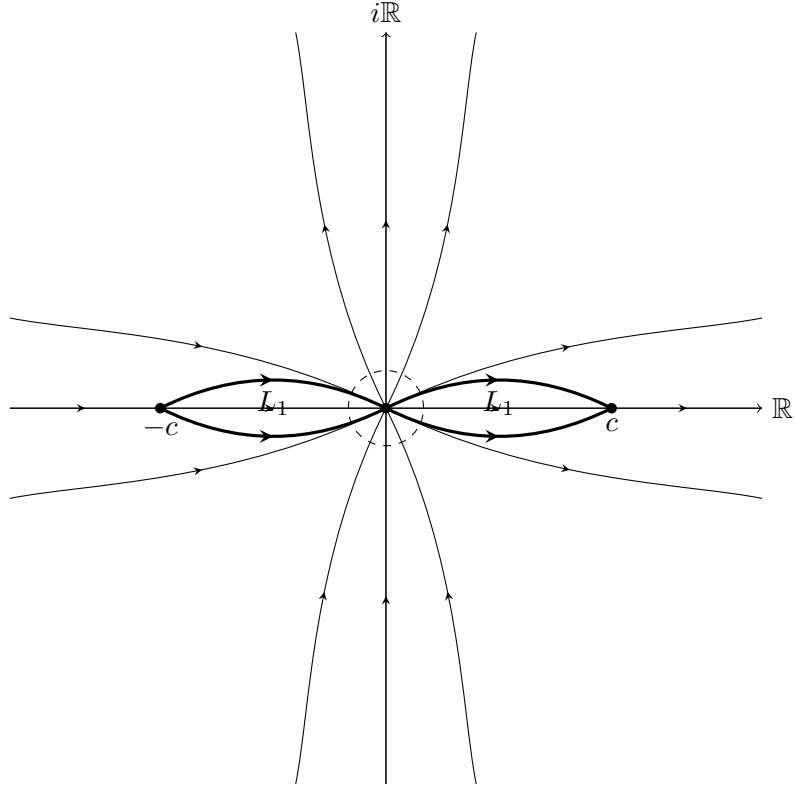

We define
\begin{align}\label{eq:defS}
S=T \times \begin{cases}
\begin{pmatrix}
1 &0&0&0 \\
-e^{n(\lambda_1-\lambda_2)} & 1 &0&0 \\
0&0&1&0\\
0&0&0&1 
\end{pmatrix}
& \textrm{\parbox[t]{0.4\textwidth}{ in the upper part of the lenses around $[-c,0]$ and $[0,c]$,}} \\
\begin{pmatrix}
1 &0&0&0 \\
e^{n(\lambda_1-\lambda_2)} & 1 &0&0 \\
0&0&1&0\\
0&0&0&1 
\end{pmatrix}
& \textrm{\parbox[t]{0.4\textwidth}{ in the lower part of the lenses around $[-c,0]$ and $[0,c]$,}} \end{cases}
\end{align}
and
\[
S=T, \qquad \textrm{ elsewhere.}
\]
Then $S$ is defined and analytic in $\C \setminus \Sigma_S$ where $\Sigma_S$ is the contour consisting of the previous contour $\Sigma_T$ and the lips of the lenses around $[-c,0]$ and $[0,c]$.

Then $S$ satisfies the following RH problem.
\begin{lemma} \label{rhp: S}
Suppose that $n \equiv 0 \mod 3$. Then $S$ has the following properties.
\begin{itemize}
\item[\rm (1)] $S$ is analytic in $\C \setminus \Sigma_S$;
\item[\rm (2)] $S_+=S_-J_S$ on $\Sigma_S$, where $J_S$ will be specified in \eqref{eq: jump for S structure}--\eqref{eq: jumps for S 2};
\item[\rm (3)] $\displaystyle{ S(z)=(I+\mathcal O(z^{-1/3})) \diag \left( 1 , z^{1/3} , 1 , z^{-1/3} \right) \begin{pmatrix} 1 & 0 \\ 0 & A_j \end{pmatrix}}$ as $z \to \infty$ in the $j$th quadrant.
\end{itemize}

The jump matrices $J_S$ are as follows.
\begin{itemize}
\item For $z \in i\R$ and $z$ on the lips of $L_2$ we have $J_S=J_T$.
\item For $z \in \R$ and $z$ on the lips of $L_1$ and $L_3$ the jump matrices have the structure
\begin{equation} \label{eq: jump for S structure}
J_S=
\begin{pmatrix}
(J_S)_1 & 0 \\
0 & (J_S)_3
\end{pmatrix},
\end{equation}
where the $2 \times 2$ blocks are given by
\begin{equation} \label{eq: jumps for S 1}
(J_S)_1  =\begin{cases}
\begin{pmatrix}
0 & 1 \\ -1 & 0
\end{pmatrix} & z \in [-c, c], \\
\begin{pmatrix}
1 & 0 \\ e^{n(\lambda_1-\lambda_2)} & 1 
\end{pmatrix} & \text{on the lips of $L_1$}, \\
(J_T)_1, & \text{\parbox[t]{0.5\textwidth}{on $\R \setminus [-c,c]$ and the remaining part of the lips of $L_3$.}}
\end{cases} 
\end{equation}
The block $(J_S)_3$ is given by
\begin{equation} \label{eq: jumps for S 2}
(J_S)_3  =\begin{cases}
I_2 & \textrm{\parbox[t]{0.5\textwidth}{ on the parts of the lips of $L_1$ that do not coincide with the lips of $L_3$,}} \\
(J_T)_3 & \textrm{\parbox[t]{0.5\textwidth}{on $\R$ and the lips of $L_3$.}}
\end{cases}
\end{equation}
\end{itemize}
\end{lemma}
\begin{proof}
Again the proof follows by straightforward computations and is similar to \cite[Sec. 7.4--7.5]{DKM}.
\end{proof}

We choose the lips of the lens $L_1$ such that we have pointwise exponential decay for the jumps as $n\to \infty$ in the critical situation $\alpha=-1$ and $\tau=1$.

\begin{lemma} \label{lemma: estimates L1star}
Let $\epsilon>0$ be a small constant. We can and do choose the  lenses $L_1$ around $[-c^*,c^*]$ such that their lips coincide with the lips of $L_3$ inside a disk $D(0,\rho_0)$, and such that
\begin{equation} \label{eq: estimate lambda12star}
\Re(\lambda_1^*(z)-\lambda_2^*(z))<-c_5|z|^{3/2},
\end{equation}
for $z\in L_1 \setminus \left( \R \cup D(-c^*,\epsilon) \cup D(c^*,\epsilon) \right)$ and some constant $c_5>0$.
\end{lemma}
\begin{proof}
The proof is similar to that of Lemma \ref{lemma: estimates L2L3star}.

For $x \in (-c^*,0) \cup (0,c^*)$ we know, see \eqref{eq: jumps lambda 1}--\eqref{eq: jumps lambda 2} and \eqref{eq: sym lambda 1}--\eqref{eq: sym lambda 2}, that $\Re(\lambda^*_{1,\pm}(x)-\lambda_{2^*,\pm}(x))=0$. Then it is clear from the Cauchy-Riemann equations and Lemma \ref{lemma: xi on cuts} that the lens can be opened such that 
\[
\Re(\lambda_1^*(z)-\lambda_2^*(z))<0, \qquad z \in L_1 \setminus \R.
\]
Using Lemma \ref{lemma: asymptotics lambda around 0} we see that we can choose $L_1$ such that we have \eqref{eq: estimate lambda12star}. The fact that we can let the lips of the local lens coincide with the lips of $L_3$ in a neighborhood $D(0,\rho_0)$ of zero is also due to the asymptotics in Lemma \ref{lemma: asymptotics lambda around 0}. Here we have to choose $\rho_0$ sufficiently small.
\end{proof}

The following lemma ensures that the jumps on the lenses decay uniformly outside the disks $D(c^*,\epsilon)$, $D(-c^*,\epsilon)$, and $\D$ around $c^*$, $-c^*$, and $0$, even for the non-critical situation.

\begin{lemma} \label{lemma: estimate on local lenses}
There is a constant $c_6>0$ such that for every sufficiently large $n$ we have
\begin{equation} \label{eq: estimate 1}
\Re (\lambda_1(z)-\lambda_2(z)) \leq -c_6 n^{-1/2},  
\end{equation}
for $z$ on the parts of the lips of $L_1$ lying outside the disks $D(c^*,\epsilon)$, $D(-c^*,\epsilon)$, and $\D$.
\end{lemma}

\begin{proof}
The proof is similar to that of Lemma \ref{lemma: estimate global lenses}. Because of Lemma \ref{lemma: estimates on lambda minus lambdastar} and Lemma \ref{lemma: estimates L1star} there exists a constant $\tilde c_1>0$ such that for $n$ large enough
\[
\Re (\lambda_1(z)-\lambda_2(z)) \leq |z|^{1/2}\left(\tilde c_1 n^{-1/3}-c_5|z|\right).
\]
Now we use that $|z|\geq  n^{-\delta}$ to get
\[
\Re (\lambda_1(z)-\lambda_2(z)) \leq -\tilde c_2 |z|^{1/2}n^{-\delta} \leq -\tilde c_2 n^{-1/2},
\]
for a certain constant $\tilde c_2>0$ and sufficiently large $n$. This proves the lemma.
\end{proof}

Now let us stop for a moment and discuss the asymptotic behavior of the jump matrices $J_S$ as $n\to \infty$. From \eqref{eq: jumps for S 1}, \eqref{eq: jumps for S 2},  \eqref{eq: jump for T 1}--\eqref{eq: jump for T 3},  and Lemmas \ref{lemma: estimate global lenses} and  \ref{lemma: estimate on local lenses} it is clear that the jump matrices on the lips of the lenses $L_1,L_2$ and $L_3$ are exponentially close to the identity. The same holds true for the jump matrices on $(-\infty,-c)$ and $(c,\infty)$. We leave it to the reader to check this. In fact, if we stay away from the points $0,\pm c$ these approximations are uniform and in $L^2$ sense. This implies that these jumps can be ignored in the asymptotic analysis.  The final steps in the steepest descent therefore are the following. First, we ignore all exponentially small entries in the jump matrices and solve the RH problem that we thus obtain.  This \emph{global parametrix} will be a good approximation to the solution to the RH problem for $S$ away from the endpoints $0,\pm c$.  Near these endpoints we define \emph{local parametrices} separately. By combining the local and global parametrices we thus construct a function that is a  good approximation to $S$ everywhere.  
\subsection{Global parametrix $P^{(\infty)}$}

If we ignore all exponentially decaying entries of $J_S$ we obtain the following RH problem:
\begin{rhp} \label{rhp:pinfty}
We search for a $4 \times 4$ matrix-valued function $\Pinf(z)$, also depending on the parameters $\alpha$ and $\tau$, satisfying
\begin{itemize}
\item[\rm (1)] $\Pinf$ is analytic in $\C \setminus (\R \cup i \R)$;
\item[\rm (2)] $\Pinf_+=\Pinf_-J_\infty$ on $\R \cup i\R$, where $J_\infty$ is specified in \eqref{eq: jump global parametrix};
\item[\rm (3)] $\displaystyle{ \Pinf(z)=\left( I+\mathcal O\left(z^{-1}\right)\right) \diag \left( 1 , z^{1/3} , 1 , z^{-1/3} \right) \begin{pmatrix} 1 & 0 \\ 0 & A_j \end{pmatrix}}$ as $z \to \infty$ in the $j$-th quadrant;
\item[\rm (4)] $\Pinf(z)=\mathcal O\left(z^{-1/4} \right)$ as $z \to 0$.
\end{itemize}
The jump matrices are specified as follows
\begin{equation} \label{eq: jump global parametrix}
J_\infty = \begin{cases}
\begin{pmatrix}
 0 & 1 &0&0 \\ -1 & 0 &0&0 \\
0&0& 0 & 1 \\ 0&0&-1 & 0 
\end{pmatrix}, & \text{ on }[-c, c], \\
 \begin{pmatrix}
1& 0 &0&0 \\
0 &1& 0 & 0\\
0& 0 & 0 & 1 \\ 
0&0&-1 & 0 
\end{pmatrix}, & \text{ on } \R \setminus [- c, c], \\
 \begin{pmatrix}
1&0&0 & 0\\
0 &  0 & -1 &0 \\ 
0 &1 & 0 & 0  \\
0 & 0 &0 &1 \end{pmatrix},
& \text{ on } i\R.
\end{cases}
\end{equation}
\end{rhp}
This RH problem can be explicitly solved as we show in the next lemma. As the explicit formulas are not relevant for our purposes, we will give a short proof only.
\begin{lemma} \label{lemma: global parametrix}
RH problem \ref{rhp:pinfty} has the solution
\begin{multline} \label{eq: solution global parametrix}
\Pinf =\begin{pmatrix}
w_1^2 & w_2^2 & w_3^2 & w_4^2\\
\gamma^2 w_1+\gamma^5 \frac{1}{w_1} & \gamma^2 w_2+\gamma^5\frac{1}{w_2} & \gamma^2 w_3+\gamma^5 \frac{1}{w_3} & \gamma^2 w_4+\gamma^5 \frac{1}{w_4}\\
\gamma^3 & \gamma^3 & \gamma^3 & \gamma^3 \\
\gamma w_1 & \gamma w_2 & \gamma w_3 & \gamma w_4
\end{pmatrix}
\\ \times \diag \left( \kappa(w_1),\kappa(w_2),\kappa(w_3),\kappa(w_4)\right),
\end{multline}
where $\kappa(w) = \left((w^2+\gamma^3)(w^2-3\gamma^3)\right)^{-1/2}$ and the square root is taken such that $\eta$ is analytic in the complex $w$ plane with a cut along ${w_1}_+([-c, c]) \cup {w_2}_-(i\R) \cup {w_3}_+(\R)$, shown in Figure \ref{fig: cuts kappa}  and such that $$\kappa(w)=w^{-2}(1+\mathcal O(1/w))$$ as $w \to \infty$. Here, $w_1,\ldots,w_4$ are the functions of $z$ defined by \eqref{eq: algebraic equation for w} and \eqref{eq: order}. 
\end{lemma}

\begin{proof}
The jump behavior follows from Lemma \ref{lemma: w meromorf} and the jumps of $\kappa$. A careful asymptotic analysis of \eqref{eq: solution global parametrix}, based on Lemmas \ref{lemma: asymptotics w around 0} and \ref{lemma: asymptotics w around infinity} respectively leads to conditions (3) and (4) in the RH problem.
\end{proof}

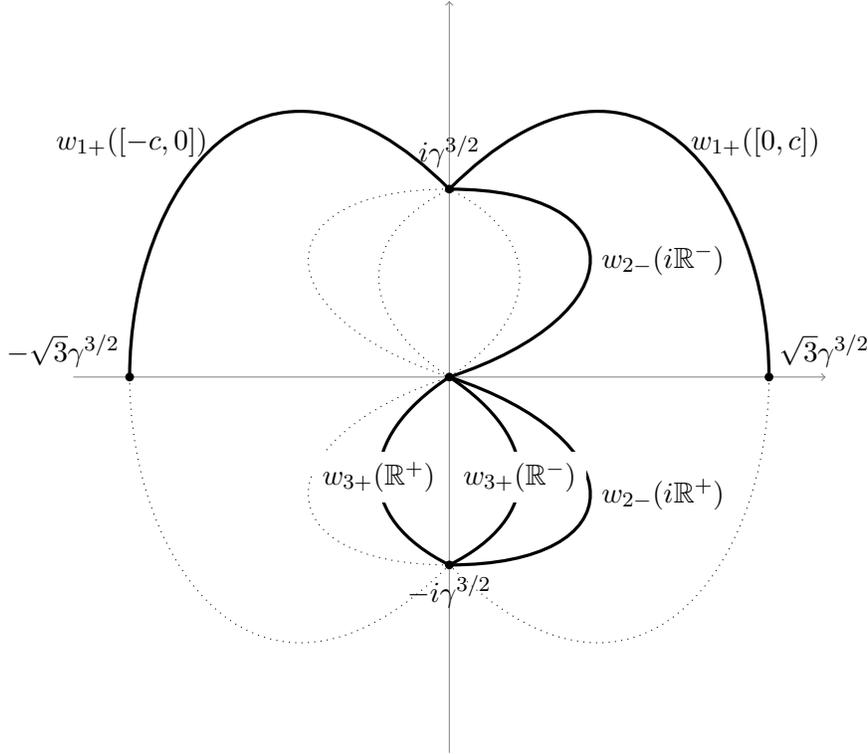
\begin{figure}[t]
\centering
\begin{tikzpicture}[scale=2.5]
\draw[help lines, ->] (-2,0)--(2,0); \draw[help lines, ->] (0,-2)--(0,2);
\draw[very thick] (-1.7,0) .. controls (-1.7,1) and (-1,2).. (0,1) node[midway, above, left]{$w_{1+}([-c,0])$}
      (1.7,0) .. controls (1.7,1) and (1,2).. (0,1) node[midway, above, right]{$w_{1+}([0,c])$};
\draw[dotted] (-1.7,0) .. controls (-1.7,-1) and (-1,-2).. (0,-1)
      (1.7,0) .. controls (1.7,-1) and (1,-2).. (0,-1);
\draw[very thick] (0,1).. controls (1,1) and (1,0.33)..(0,0) node[midway,right]{$w_{2-}(i\R^-)$}
                  (0,-1).. controls (1,-1) and (1,-0.33)..(0,0) node[midway,right]{$w_{2-}(i\R^+)$}; 
\draw[dotted] (-0,1).. controls (-1,1) and (-1,0.33)..(0,0)
                  (0,-1).. controls (-1,-1) and (-1,-0.33)..(0,0);
\draw[dotted] (0,0) .. controls (0.5,0.33) and (0.5,0.75).. (0,1)
              (0,0) .. controls (-0.5,0.33) and (-0.5,0.75).. (0,1); 
\draw[very thick] (0,0) .. controls (0.5,-0.33) and (0.5,-0.75).. (0,-1) node[midway,fill=white]{$w_{3+}(\R^-)$}
                  (0,0) .. controls (-0.5,-0.33) and (-0.5,-0.75).. (0,-1)node[midway,fill=white]{$w_{3+}(\R^+)$};
\fill (0,0) circle (0.66pt)
      (-1.7,0) circle (0.66pt) node[above left]{$-\sqrt 3 \gamma^{3/2}$} 
      (1.7,0) circle (0.66pt) node[above right]{$\sqrt 3 \gamma^{3/2}$}
      (0,1) circle (0.66pt) 
      (0,-1) circle (0.66pt) node[below]{$-i \gamma^{3/2}$};
\draw (0,1.2) node{$i \gamma^{3/2}$};
\end{tikzpicture}
\caption{Image of the map $w: \mathcal R \to \overline{\C}$ in the $w$ plane. The indicated lines (plain and dotted) are the images of the cuts in the Riemann surface $\mathcal R$ under this map. The solid lines are the contours where the function $\kappa: \C \to \C: w \mapsto \kappa(w)$ changes sign.}
\label{fig: cuts kappa}
\end{figure}

The global parametrix $\Pinf$ fails to be a good approximation near $z=-c^*,0,c^*$. Near these points we need to construct local parametrices, which we will do next.

\subsection{Local parametrices around $\pm c^*$}

We now briefly discuss the construction of local parametrices $P^{(-c^*)}$ and $P^{(c^*)}$ in  disks $D(-c^*,\epsilon)$ and $D(c^*,\epsilon)$ around $\pm c^*$ for sufficiently small $\epsilon>0$. The idea is to construct $P^{(\pm c^*)}$ such that it makes exactly the same jumps in the disk $D(\pm c^*,\epsilon)$ as $S$, and such that
\begin{equation} \label{eq: matching on fixed disks}
P^{(\pm c^*)}(z)= \left( I+ \mathcal O \left( n^{-1} \right) \right) P^{(\infty)}(z), \qquad \text{uniformly for $|z \mp c^*|=\epsilon$,}
\end{equation}
as $n \to \infty$. This construction can be done in exactly the same way as in \cite[Sec. 9]{DKM} using Airy functions. As this construction is standard and the explicit formulas are irrelevant to the proof of Theorem \ref{th: main theorem}, we will omit it here.

 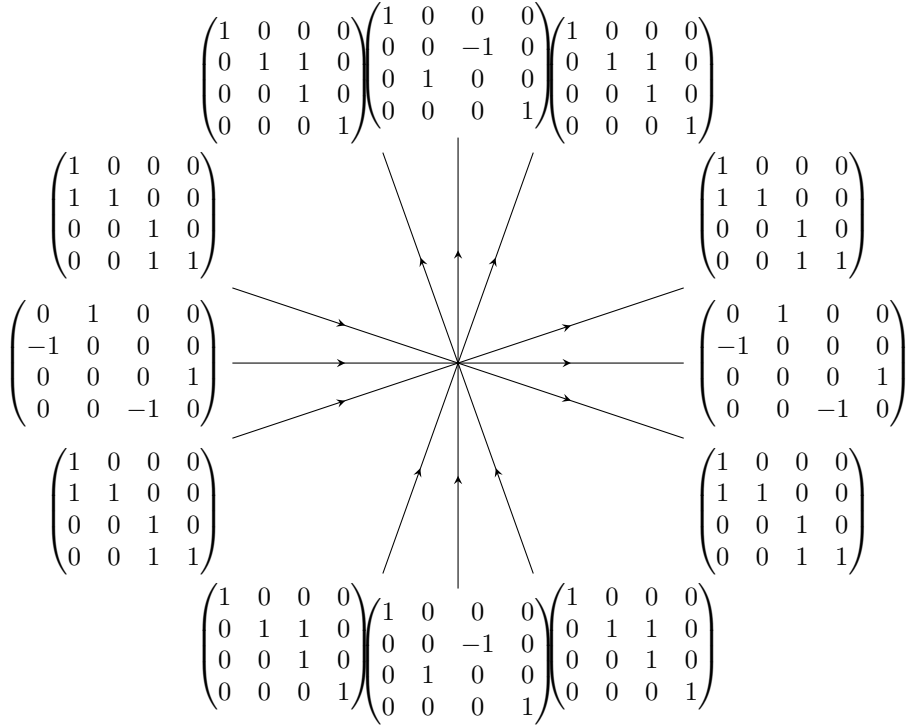
\begin{figure}
 \centering{
\small{
\begin{tikzpicture}[scale=1]
\begin{scope}[decoration={markings,mark= at position 0.5 with {\arrow{stealth}}}]
\draw[postaction={decorate}]      (0,0)--(3,0) node[right]{$\begin{pmatrix} 0 & 1 & 0&0\\ -1 & 0 &0&0 \\ 0&0& 0 & 1 \\0&0 &-1 & 0 \end{pmatrix}$};
\draw[postaction={decorate}]      (0,0)--(3,1) node[above right]{$\begin{pmatrix}  1& 0 &0&0 \\ 1 & 1 &0&0 \\ 0&0& 1& 0 \\0&0& 1 & 1 \end{pmatrix}$};
\draw[postaction={decorate}]      (0,0)--(1,2.8) node[above right]{$\begin{pmatrix} 1&0&0 &0 \\0& 1 & 1&0 \\0& 0 & 1 &0  \\0&0&0&1 \end{pmatrix}$};
\draw[postaction={decorate}]      (0,0)--(0,3) node[above]{$\begin{pmatrix} 1& 0&0&0 \\0& 0 & -1&0 \\0& 1 & 0 &0 \\0&0&0&1 \end{pmatrix}$};
\draw[postaction={decorate}]      (0,0)--(-1,2.8) node[above left]{$\begin{pmatrix} 1&0&0 &0 \\0& 1 & 1&0 \\0& 0 & 1 &0  \\0&0&0&1 \end{pmatrix}$};
\draw[postaction={decorate}]      (-3,1)node[above left]{$\begin{pmatrix}  1& 0 &0&0 \\ 1 & 1 &0&0 \\ 0&0& 1& 0 \\0&0& 1 & 1 \end{pmatrix}$}--(0,0);
\draw[postaction={decorate}]      (-3,0) node[left]{$\begin{pmatrix} 0&1&0&0\\-1&0&0&0 \\ 0&0&0 & 1 \\0&0 &-1 &0\end{pmatrix}$}--(0,0);
\draw[postaction={decorate}]      (-3,-1)node[below left]{$\begin{pmatrix}  1& 0 &0&0 \\ 1 & 1 &0&0 \\ 0&0& 1& 0 \\0&0& 1 & 1 \end{pmatrix}$}--(0,0);
\draw[postaction={decorate}]      (-1,-2.8) node[below left]{$\begin{pmatrix} 1&0&0 &0 \\0& 1 & 1&0 \\0& 0 & 1 &0  \\0&0&0&1 \end{pmatrix}$}--(0,0);
\draw[postaction={decorate}]      (0,-3) node[below]{$\begin{pmatrix} 1& 0&0&0 \\0& 0 & -1&0 \\0& 1 & 0 &0 \\0&0&0&1 \end{pmatrix}$}--(0,0);
\draw[postaction={decorate}]      (1,-2.8) node[below right]{$\begin{pmatrix} 1&0&0 &0 \\0& 1 & 1&0 \\0& 0 & 1 &0  \\0&0&0&1 \end{pmatrix}$}--(0,0);
\draw[postaction={decorate}]      (0,0)--(3,-1)node[below right]{$\begin{pmatrix}  1& 0 &0&0 \\ 1 & 1 &0&0 \\ 0&0& 1& 0 \\0&0& 1 & 1 \end{pmatrix}$};
\end{scope}
\end{tikzpicture}}
\caption{Contour $\Sigma_{\widetilde M}$ together with the jumps on each ray.}
\label{fig: contour Mbar}}
\end{figure}

\subsection{Model RH problem for the local parametrix around 0}
\label{sec:modelRHP}

In Section \ref{sec:P0} we will construct the local parametrix around $0$ based on the solution $M$ to RH problem \ref{rhp: tacnode rhp}. We formulated RH problem \ref{rhp: tacnode rhp} in the  natural form as it appeared also in \cite{DKZ}. For our purposes, however, we need to transform $M$ into $\widetilde M$ as in \eqref{eq: Mbar definition}.

Then, $\widetilde{M}(\zeta)$ is the unique solution to the following RH problem.
\begin{lemma}\label{lem:transformRH}
The function $\widetilde{M}(\zeta)$ has the following properties
\begin{itemize}
\item[\rm (1)] $\widetilde M$ is analytic in $\C \setminus \Sigma_{\widetilde M}$, where $\Sigma_{\widetilde M}$ is shown in Figure \ref{fig: contour Mbar};
\item[\rm (2)] $\widetilde M_+=\widetilde M_-J_{\widetilde M}$ on $\Sigma_{\widetilde M}$, where $J_{\widetilde M}$ is specified in Figure \ref{fig: contour Mbar};
\item[\rm (3)] as $\zeta \to \infty$ we have
\begin{align*}
\widetilde M(\zeta)=\widetilde B^{-1}(\zeta) A^{-1} C_\pm  \left( I+ \mathcal O(\zeta^{-1/2}) \right)D(\zeta), \quad \pm \Im \zeta>0,
\end{align*}
where  $\widetilde B(\zeta)=B(i \zeta)$  and $A,B$ as in RH problem \ref{rhp: tacnode rhp}, and 
\begin{align} \label{eq:defD}
D(\zeta)=\begin{cases}
\diag \left(e^{\eta_2(\zeta)-it\zeta},e^{\eta_1(\zeta)+it\zeta},e^{-\eta_1(\zeta)+it\zeta},e^{-\eta_2(\zeta)-it\zeta}\right), & \zeta \in I,\\
\diag \left(e^{\eta_2(\zeta)-it\zeta},e^{-\eta_1(\zeta)+it\zeta},e^{\eta_1(\zeta)+it\zeta},e^{-\eta_2(\zeta)-it\zeta}\right), & \zeta \in II,\\
\diag \left(e^{\eta_1(\zeta)+it\zeta},e^{\eta_2(\zeta)-it\zeta},e^{-\eta_2(\zeta)-it\zeta},e^{-\eta_1(\zeta)+it\zeta}\right), & \zeta \in III,\\
\diag \left(e^{\eta_1(\zeta)+it\zeta},e^{-\eta_2(\zeta)-it\zeta},e^{\eta_2(\zeta)-it\zeta},e^{-\eta_1(\zeta)+it\zeta}\right), & \zeta \in IV,
\end{cases}
\end{align}
with 
\begin{align}
\label{eq:defeta}
\begin{split}
\eta_1(\zeta )&=\tfrac{2}{3} r_1 e^{3\pi i /4} \zeta ^{3/2}+2 s_1 e^{\pi i /4} \zeta ^{1/2},\\
\eta_2(\zeta )&=\tfrac{2}{3} r_2 e^{3\pi i /4} (-\zeta )^{3/2}+2 s_2 e^{\pi i /4} (-\zeta )^{1/2};
\end{split}
\end{align}
The fractional powers are taken in the same way as in RH problem \ref{rhp: tacnode rhp}.
\item[\rm (4)] $\widetilde M$ is bounded near zero.
\end{itemize}
\end{lemma}
\begin{proof} 
The jump contours and the jumps are straightforward. For the asymptotic condition, we first compute $\psi_1(i \zeta)$ and $\psi_2(i\zeta)$ for $\zeta$ in the different quadrants and $\psi_1,\psi_2$ as in \eqref{eq: def psi1}--\eqref{eq: def psi2}
\begin{align*}
\psi_1(i \zeta)&= \begin{cases}\eta_1(\zeta), &\zeta \in I \cup III \cup IV,\\
-\eta_1(\zeta), & \zeta \in II,
\end{cases} \\
\psi_2(i \zeta)&= \begin{cases}\eta_2(\zeta), &\zeta  \in I \cup II \cup III,\\
-\eta_2(\zeta), & \zeta \in IV.
\end{cases}
\end{align*}
Hence we compute  that $\widetilde M$ satisfies
\[
\widetilde M(\zeta)= \left( I+ \mathcal O(\zeta^{-1}) \right) \widetilde B^{-T}(\zeta) A^{-T} C_\pm D(\zeta),  \quad \pm \Im \zeta>0,
\]
as $\zeta \to \infty$. Note that $A^T=A$ and $\widetilde B^T=\widetilde B$. Moreover, by putting $\widetilde B^{-1}A^{-1}C_\pm$ in front, the order term turns into $\mathcal O(\zeta^{-1/2})$. This gives the asymptotic condition.
 \end{proof}

\subsection{Local parametrix $\P0$ around 0} \label{sec:P0}

In this section we construct a local parametrix $\P0$ in the shrinking disk $\D$ where we recall that $\delta$ satisfies the inequality in \eqref{eq: bounds kappa}. We want the local parametrix to be a good approximation to $S$ in this disk, i.e. we want
\begin{itemize}
\item[(C1)] $\P0$ is analytic in $\D \setminus \Sigma_S$;
\item[(C2)] $\P0$ makes the same jumps as $S$ on $\D \cap \Sigma_S$;
\item[(C3)] $\P0$ matches with the global parametrix $\Pinf$ on the boundary of the disk, i.e. 
\[
\P0(z)=\left( I+\mathcal O\left(n^{\delta-1/3} \right) \right) \Pinf(z), \qquad |z|= n^{-\delta},
\]
as $n\to \infty$.
\end{itemize}

We search for such a $\P0$ in the form
\begin{multline} \label{eq: structure local parametrix}
\P0(z)=E_n(z) \widetilde M(n^{2/3} f(z); r_1,r_2,n^{2/3} s_1,n^{2/3}s_2,n^{1/3} t)\\ \times \diag \left( e^{n \lambda_1(z)},e^{n \lambda_2(z)},e^{n \lambda_3(z)},e^{n \lambda_4(z)}\right),
\end{multline}
where $\widetilde M$ is the solution to the model RH problem and $f$ is a conformal map on a neighborhood of the origin. The function $E_n $ is an analytic $4 \times 4$ matrix valued function. Defined in this way, $\P0$ can already be seen to satisfy condition (C1) and (C2) above. Indeed, (C1) is immediate, possibly after slightly deforming some contours. The jumps follow from Lemma \ref{lemma: lambda meromorf}, \eqref{eq: structure local parametrix}, and the jumps for $\widetilde M$ indicated in Figure \ref{fig: contour Mbar}. Finally, we have the conformal map $f(z)$, the analytic prefactor $E_n$, and the parameters $r_j$, $s_j$ and $t$ at our disposal to make $\P0$ also satisfy the matching condition (C3). It will be necessary to let the parameters depend on $z$ and $n$. To retain the analyticity of $\P0$ the dependence on $z$ will have to be analytic. We will argue later that also for these complex parameters RH problem \ref{rhp: tacnode rhp} is still solvable as long as $n$ is sufficiently large. 

We now come to the precise definitions, starting with the conformal map $f$. We recall that $F$, $G,$ $H,$ and $K$ defined in Lemma  \ref{lemma: asymptotics lambda around 0} are even functions that are analytic in a neighborhood of the origin.  We define 
\begin{align}
f(z)&=z,\label{eq: def f} \\
r_1(z)&=\frac32 e^{-\pi i/4}\left(H(z)-(F(z)-F(0))z^{-1} \right), \label{eq: def r1} \\
r_2(z)&=\frac32 e^{-\pi i/4}\left(H(z)+(F(z)-F(0))z^{-1} \right), \label{eq: def r2} \\
s(z)  &=s_1(z)=s_2(z)=-\tfrac12 e^{\pi i/4}F(0), \label{eq: def s1} \\
t(z)  &=-iG(z). \label{eq: def t}
\end{align}
Note that $r_{1,2}$, $s$, and $t$ are analytic in a sufficiently small neighborhood of the origin.
In the construction of the parametrix, we will use that the functions $r_j,s$, and $t$ are analytic in $\D$  and have the asymptotic behavior as $n\to \infty$ indicated in the following lemma. Note that $r_j$, $s$ and $t$ depend on $n$, although we have suppressed that dependence in the notation. For this lemma the lower bound in \eqref{eq: bounds kappa} is essential.

\begin{lemma}\label{lemma: rst} 
Let $\frac16 <\delta <\frac13$. The functions $r_j(z)$, $s(z)$, and $t(z)$ are analytic for $z\in \D$ and $n$ sufficiently large. Moreover, we have the following limits
\begin{align}\label{eq:limitr0s0t0}
\begin{cases}
\lim_{n \to \infty} r_j(n^{-\delta}z)=1,\\
\lim_{n \to \infty} n^{2/3} s( n^{-\delta}z)=\frac14(a^2-5b),\\
\lim_{n \to \infty}n^{1/3} t( n^{-\delta}z)=-a,
\end{cases}\end{align}
for $|z|< 1$ and $j=1,2$. \end{lemma}
\begin{proof}
It was discussed above that these functions are analytic in $\D$ for $n$ sufficiently large. As $F(z)$ is even, \eqref{eq: def r1}--\eqref{eq: def r2} and \eqref{eq: H} lead to
\[
r_j(z)=\tfrac32 e^{-\pi i/4}H(0)+\mathcal O(z)=\tfrac32 \gamma^{-5/4} \left( \tfrac12 \gamma^2- \tfrac{1}{12\gamma}+\tfrac14 \tau^{4/3} \gamma \right)+\mathcal O(z),
\] 
as $z \to 0$. The limiting behavior for $r_j$ then follows recalling that $\gamma,\tau \to 1$ as $n \to \infty$.

For $s(z)$ we follow the same strategy. Using \eqref{eq: def s1} and \eqref{eq: F} we obtain 
\[ 
n^{2/3} s( n^{-\delta}z)=n^{2/3} \gamma^{1/4} \left( -2 \gamma^2+ \tfrac{1}{\gamma} + \tau^{4/3}\gamma \right),
\]
which in combination with \eqref{eq: scaling gamma} and  \eqref{eq: scaling t tau} yields the statement.

The limiting behavior for $t$ can be found by realizing that 
\[
\lim_{n \to \infty}n^{1/3} t( n^{-\delta}z)=-i\lim_{n \to \infty}n^{1/3} G( n^{-\delta}z)=\lim_{n \to \infty}n^{1/3} \gamma^{-1/2}\left( \tfrac32 \gamma^2-\tfrac{1}{4 \gamma}-\tfrac54 \tau^{4/3}\gamma \right),
\]
which follows from \eqref{eq: def t}, \eqref{eq: G}, and the first inequality in \eqref{eq: bounds kappa}. Inserting the limiting behavior for $\gamma$ and $\tau$ as given in \eqref{eq: scaling gamma} and  \eqref{eq: scaling t tau}, finishes the proof. 
\end{proof}

From Lemma \ref{lemma: rst} it follows that \eqref{eq: structure local parametrix} is well-defined (postponing the definition of $E_n$ for a moment). Indeed, it follows from standard arguments that if the solution to the RH problem  in Lemma \ref{lem:transformRH} exists, it depends analytically on the parameters $r_j,s$ and $t$. Combining this observation with Lemma \ref{lemma: rst} and Theorem \ref{th: tacnode rhp} we see that for the choice of $r_j$, $s$, and $t$ that we made, the solution to RH problem exists and hence \eqref{eq: structure local parametrix} is well-defined for sufficiently large $n$.  Moreover, we defined $r_j$, $s$, and $t$ in order to have the following equalities, which are the key to obtain (C3). 

\begin{lemma}
We have that
\begin{multline}\label{eq:etaf}
\eta_1\left(n^{2/3}z;r_1(z),n^{2/3}s(z)\right)\pm i n t(z) z \\ =\begin{cases}
n\left((-z)^{3/2} H(z)+(-z)^{1/2} F(z)\pm z G(z)\right), & z\in I,II,\\
n\left(-(-z)^{3/2} H(z)-(-z)^{1/2} F(z)\pm z G(z)\right), & z\in III, IV,
\end{cases}
\end{multline}
and
\begin{multline}\label{eq:etaminusgf}
\eta_2\left(n^{2/3}z;r_2(z),n^{2/3} s(z)\right)\pm i n t(z) z \\ =\begin{cases}
n\left(-z^{3/2} H(z)-z^{1/2} F(z)\pm z G(z)\right), & z\in I,II,\\
n\left(z^{3/2} H(z)+z^{1/2} F(z) \pm z G(z)\right), & z\in III, IV,
\end{cases}
\end{multline}
for $z\in \D$.
\end{lemma}

\begin{proof} 
We use  the definitions of $\eta_j$, $r_j(z)$, $s(z)$, and $t(z)$ as given in \eqref{eq:defeta}, and  \eqref{eq: def r1}--\eqref{eq: def t}. The first step in the proof is to realize that 
\begin{multline*}
\frac23 r_1 e^{3\pi i/4}z^{3/2} +2  e^{\pi i/4} s(z) z^{1/2}=  i z^{3/2} H(z) - i z^{1/2}F(z)\\= \begin{cases}
(-z)^{3/2} H(z)+(-z)^{1/2} F(z), & z\in I,II,\\
-(-z)^{3/2} H(z)-(-z)^{1/2} F(z), & z\in III, IV,\end{cases}\end{multline*}
which follows from the definitions by a straightforward check. Then \eqref{eq:etaf} easily follows combining this with \eqref{eq:defeta}.

The statements in \eqref{eq:etaminusgf} follow in the same way. Indeed,
\begin{multline*}
\frac23 r_2(z) e^{3\pi i/4}(-z)^{3/2}+ 2  e^{\pi i/4} s(z) (-z)^{1/2}=  i (-z)^{3/2}\left( H(z) +z^{-1} F(z) \right)\\= \begin{cases}
-z^{3/2} H(z)-z^{1/2} F(z), & z\in I,II,\\
z^{3/2} H(z)+z^{1/2} F(z), & z\in III, IV. 
\end{cases}\end{multline*}
\end{proof}

Finally, we define the analytic prefactor $E_n$ by
\begin{equation} \label{eq: analytic prefactor}
E_n(z)=e^{-nz^2K(z)}\Pinf(z) C_\pm^{-1} A \widetilde B \left(n^{2/3}z \right), \qquad \pm \Im z >0.
\end{equation}

\begin{lemma}
The prefactor $E_n$ is analytic in a neighborhood of the origin.
\end{lemma}
\begin{proof}
Trivially, $E_n$ is analytic within each quadrant. Moreover, it follows from the behavior of $\Pinf$ around zero, stated in Lemma \ref{lemma: global parametrix} and \eqref{eq: B} that $E_n(z)=\mathcal O \left(z^{-1/2} \right)$ as $z \to 0$. Therefore, it is sufficient to prove that $E_n$ has no jumps on the real and imaginary axes. To this end, we note that  by some elementary computations one can check that near the origin $\widetilde B(z)^{-1} A^{-1} C_\pm$ satisfies the same jump conditions as $\Pinf$. Therefore, $E_n$ has indeed no jumps and hence it is analytic.
\end{proof}
We have the following result on the local parametrix $\P0$. 

\begin{proposition}
If $n\equiv 0 \mod 2$, then the function $\P0$ as defined in \eqref{eq: structure local parametrix} satisfies (C1)-(C3).\end{proposition}
\begin{proof}
In the discussion right below the conditions (C1)--(C3) we already argued that the conditions (C1) and (C2) are automatically satisfied, so it remains to show that  we have  (C3).

By inserting \eqref{eq:etaf} and \eqref{eq:etaminusgf} in \eqref{eq:defD} and using Lemma \ref{lemma: asymptotics lambda around 0} and $n\equiv 0 \mod 2$, we obtain 
\begin{align}
D(n^{2/3}f(z))=e^{nz^2K(z)}\diag\left(e^{-n\lambda_1(z)},e^{-n\lambda_2(z)},e^{-n\lambda_3(z)},e^{-n\lambda_4(z)}\right).
\end{align}
Note that we need $n\equiv 0 \mod 2$ because \eqref{eq:lambda1inI} and \eqref{eq:etaf}--\eqref{eq:etaminusgf} differ  by  a term $\pi i$ at some places. Recall that $\Pinf(z)=\mathcal O(z^{-1/4})$ and $\Pinf(z)^{-1}=\mathcal O(z^{-1/4})$ as $z \to 0$. Then we have 
\[\P0(z)\Pinf(z)^{-1}=\Pinf(z)\left(1+\mathcal O(n^{-1/3}z^{-1/2})\right)\Pinf(z)^{-1}=I+\mathcal O(n^{\delta-1/3}),\]
where the order is uniform for $|z|= n^{-\delta}$. 
\end{proof}

\subsection{Final transformation: $S \mapsto R$}

We define the final transformation as
\begin{equation} \label{eq: S to R}
R(z)= \begin{cases}
S(z) \left(\P0\right)^{-1}(z), & \text{for $z \in \D$, } \\
S(z) \left(P^{(\pm c^*)}\right)^{-1}(z), & 
\text{\parbox[t]{0.4\textwidth}{for $z\in D(\pm c^*,\epsilon)$,} } \\
S(z) {\Pinf}^{-1}(z), & \text{elsewhere. } 
\end{cases}
\end{equation}
Then $R$ is analytic in $\C \setminus \Sigma_R$ where $\Sigma_R$ is the contour consisting of the three circles around 0, $-c^*$, and $c^*$, and the parts of the real and imaginary axes and the lips of the lenses that lie outside these circles. The circles are oriented clockwise.

Then $R$ satisfies the following RH problem.
\begin{lemma}
Suppose that $n \equiv 0 \mod 6$. Then
\begin{itemize}
\item[\rm (1)] $R$ is analytic in $\C \setminus \Sigma_{R}$;
\item[\rm (2)] $R_+=R_-J_R$ on $\Sigma_R$ with jump matrices
\[
J_R= \begin{cases}
\P0 \left(\Pinf \right)^{-1} & \text{for $|z|=n^{-\delta}$,}\\
P^{(\pm c^*)} \left(\Pinf \right)^{-1} & \text{for $|z\mp c^*|=\epsilon$,}\\
\left( \Pinf \right)_- J_S \left(\Pinf \right)_+^{-1} & \text{elsewhere on $\Sigma_R$;}
\end{cases}
\]
\item[\rm (3)] as $z \to \infty$ we have
\[
R(z)=I+\mathcal O(1/z).
\]
\end{itemize}
\end{lemma}

Now, all jump matrices are close to the identity matrix as $n$ gets large. On the segments of the real and imaginary axes belonging to $\Sigma_R$ the jump vanishes. On the shrinking circle around 0  the matching condition (C3) in the construction of $\P0$ in Section \ref{sec:P0} yields
\[
J_R=I+\mathcal O \left( n^{\delta-1/3}\right), \qquad \text{as $n \to \infty$,}
\]
uniformly for $|z|=n^{-\delta}$. From \eqref{eq: matching on fixed disks} and \eqref{eq: estimate 1} we find
\[
J_R=I+\mathcal O \left( n^{-1}\right), \qquad \text{as $n \to \infty$,}
\]
for $z$ on the other bounded parts of the contour, i.e. on the circles around $-c^*$ and $c^*$ and on the remainders of the lips of the bounded lens $L_1$. Moreover, in view of \eqref{eq: estimate 2} and \eqref{eq: estimate 3}, there exists a constant $d >0$ such that on the unbounded parts of the contour $\Sigma_R$ the jump matrices satisfy
\[
J_R= I + \mathcal O \left( e^{-dn^{1/2}\max(1,|z|^{4/3})} \right), \qquad \text{as $n \to \infty$,}
\]
uniformly for $z$ on the unbounded arcs of $\Sigma_R$.

We have now achieved the goal of the steepest descent analysis for the RH problem for $Y$. $R(z)$ tends to the identity matrix as $z \to \infty$ and the jump matrices for $R$ tend to the identity matrix as $n \to \infty$, both uniformly on $\Sigma_R$ and in $L^2(\Sigma_R)$. By standard arguments, see \cite{Deift} and in particular \cite{BK} for the case of a moving contour, this leads to the conclusion of our steepest descent analysis.
\begin{lemma}
Let $R$ be as in \eqref{eq: S to R} and $n\equiv 0 \mod 6.$ Then 
\begin{equation} \label{eq: conclusion steepest descent}
R(z)=I+\mathcal O \left( \frac{n^{\delta-1/3}}{(1+|z|)}\right),
\end{equation}
as $n \to \infty$, uniformly for $z \in \C \setminus \Sigma_R$.
\end{lemma}

\subsection{Proof of Theorem \ref{th: main theorem}} \label{sec: proof of main theorem}

\begin{proof}[Proof of Theorem \ref{th: main theorem}.]
We start with the formula \eqref{eq: K in terms of X} that expresses $K_n$ in terms of $X$.
The idea is to write this expression in terms of $R$ instead of $X$ by applying all respective transformations
\[
X \mapsto U \mapsto T \mapsto S \mapsto R
\] introduced in the steepest descent analysis. Meanwhile, we let $n$ tend to infinity, so that we can exploit the conclusion of the steepest descent analysis \eqref{eq: conclusion steepest descent}. 

Let $x,y\in \R$. First, unfolding the transformation $X \mapsto U$, given in \eqref{eq: X to U} yields
\[
K_n(x,y)=\frac{e^{-n(y^2/2-x^2/2-\lambda_{2,+}(y)+\lambda_{1,+}(x))}}{2\pi i (x-y)}\begin{pmatrix} 0 & 1  & 0 & 0 \end{pmatrix} U_+^{-1}(y)U_+(x) \begin{pmatrix} 1\\0\\0\\0 \end{pmatrix}.
\]
The opening of global lenses $U \mapsto T$ in \eqref{eq:defT1}--\eqref{eq:defT3} does not effect the above expression. Indeed, as $x$ and $y$ are real, they fall outside the global lenses around the imaginary axis. Also, the opening of the global lens around the real axis does not affect this expression. Thus,
\[
K_n(x,y)=\frac{e^{-n(y^2/2-x^2/2-\lambda_{2,+}(y)+\lambda_{1,+}(x))}}{2\pi i (x-y)}\begin{pmatrix} 0 & 1  & 0 & 0 \end{pmatrix} T_+^{-1}(y)T_+(x) \begin{pmatrix} 1\\0\\0\\0 \end{pmatrix}.
\]
 The opening of the local lens $T \mapsto S$ in \eqref{eq:defS}, however, does have impact:
\begin{equation} \label{eq: K in terms of S}
K_n(x,y)=\frac{e^{-n(y^2/2-x^2/2)}}{2\pi i (x-y)}\begin{pmatrix} -e^{n \lambda_{1,+}(y)} & e^{n \lambda_{2,+}(y)}  & 0 & 0 \end{pmatrix} S_+^{-1}(y)S_+(x) \begin{pmatrix} e^{-n \lambda_{1,+}(x)}\\e^{-n \lambda_{2,+}(x)}\\0\\0 \end{pmatrix},
\end{equation}
if we assume $-c<x,y<c$.
 Moreover, for $x,y\in \D$, we have 
\[
S(z)=R(z) \P0(z), \qquad z=x,y,
\]
by \eqref{eq: S to R}.
It then follows that
\begin{multline} \label{eq: K in terms of R}
K_n(x,y)=\frac{e^{-n(y^2/2-x^2/2)}}{2\pi i (x-y)}  \begin{pmatrix} -1 & 1  & 0 & 0 \end{pmatrix} \\  \widetilde M^{-1}_+\left(n^{2/3}y; r_1(y),r_2(y),n^{2/3}s(y),n^{2/3}s(y),n^{1/3}t(y)\right)  E_n^{-1}(y)  R^{-1}(y)R(x) E_n(x) \\ \widetilde M_+\left(n^{2/3}x; r_1(x),r_2(x),n^{2/3}s(x),n^{2/3}s(x),n^{1/3}t(x)\right)\begin{pmatrix} 1 \\ 1 \\ 0 \\ 0 \end{pmatrix}.
\end{multline} 
Now we scale $x$ and $y$ with $n$ such that
\[
x=\frac{u}{n^{2/3}} \qquad \textrm{ and } \qquad y=\frac{v}{n^{2/3}},
\]
and $u,v \in  \R$. Then for large $n$, $x$ and $y$ belong the the disk $\D$, so that \eqref{eq: K in terms of R} holds. We want to take the limit as $n \to \infty$. Note that under these conditions
\[
\lim_{n \to \infty} e^{-n(y^2/2-x^2/2)}=1,
\]
and by \eqref{eq:limitr0s0t0} we have
\begin{align*}
r_j(z) &\to 1, && j=1,2, \\
n^{2/3}s_j(z) &\to \tfrac14 (a^2-5b), && j=1,2, \\
n^{1/3}t(z) & \to -a, 
\end{align*}
as $n \to \infty$ and $z=x,y$.
Furthermore, it follows from \eqref{eq: conclusion steepest descent} and Cauchy's formula that
\begin{equation}\label{eq:cauchyonR}
\begin{aligned}  
R^{-1}(y)(R(y)-R(x))&= \frac{(y-x)}{2\pi i} \oint \frac{R^{-1}(y)(R(w)-I)}{(w-x)(w-y)} \ud w \\ 
                    &= \mathcal O \left( \frac{x-y}{n^{1/3-\delta}}\right)=\mathcal O \left(n^{\delta-1}\right).
\end{aligned}
\end{equation}
as $n\to \infty$ where the constant is uniform for $u,v$ in compact subsets of $\R$. 

We estimate the factor $E_n$ given in \eqref{eq: analytic prefactor} in the same way. First note that we have $E_n(n^{-2/3} u)=\mathcal O\left(n^{1/6} \right)$ and $E_n(n^{-2/3} v)=\mathcal O\left(n^{1/6}\right)$ as $n\to \infty$ uniformly for $u,v$ in compact subsets. Then we obtain
\begin{align*} 
E_{n}^{-1}(y)(E_n(y)-E_{n}(x)) & = \frac{(y-x)}{2\pi i} \oint \left(\frac{E_n^{-1}(y)(E_n(w)-I)}{(w-x)(w-y)}\right) \ud w \\
                               & = \mathcal O \left( (x-y)n^{1/3}\right)=\mathcal O \left( n^{-1/3} \right),
\end{align*}
as $n \to \infty$ where the order term is uniform for $u,v$ in compact subsets. 
This results in
\[
\lim_{n \to \infty} E_n^{-1}(y) R_+^{-1}(y)R_+(x) E_n(x)=I,
\]
uniformly for $u,v$ in compact subsets of $\R$.
Inserting the estimates for $R$ and $E_n$  into \eqref{eq: K in terms of R} yields
\begin{multline*}
\lim_{n \to \infty} \frac{1}{n^{2/3}} K_n\left( \frac{u}{n^{2/3}},\frac{v}{n^{2/3}}\right) = \frac{1}{2 \pi i(u-v)} \begin{pmatrix} -1 & 1  & 0 & 0 \end{pmatrix}  \\  \times \widetilde M^{-1}_+\left(v;1,1,\tfrac14 (a^2-5b),\tfrac14 (a^2-5b),-a\right) \widetilde M_+\left(u; 1,1,\tfrac14 (a^2-5b),\tfrac14 (a^2-5b),-a\right)\begin{pmatrix} 1 \\ 1 \\ 0 \\ 0 \end{pmatrix}.
\end{multline*}
This proves  the statement.\end{proof}

 \section{Solvability for RH problem \ref{rhp: tacnode rhp}}\label{sec:existence}

The aim of this section is to prove Theorem \ref{th: tacnode rhp}. Note that by a simple rescaling it is sufficient to consider the case $r_1=r_2=1$ only. We will assume this throughout this section. 

\subsection{Lax pair} \label{sec: lax pair}
Consider first the system of equations
\begin{equation}\label{eq:laxpair}
\begin{cases}
\frac{\partial }{\partial \zeta} M(\zeta,t)=U(\zeta,t) M(\zeta,t),\\
\frac{\partial }{\partial t} M(\zeta,t)=W(\zeta,t) M(\zeta,t).
\end{cases}
\end{equation}
where the matrices $U$ and $W$ are explicitly defined as follows. We need  the Hastings-McLeod solution $q(\nu)$ to the Painlev\'e II  equation 
as defined in \eqref{eq:PIIeq} and \eqref{eq:HScharach}.
We will also use the Hamiltonian
\begin{equation} \label{eq: def u}
u(\sigma)=(q'(\sigma))^2-\sigma q^2(\sigma)-q^4(\sigma).
\end{equation}
Note that $u'(\sigma)=-q^2(\sigma)$. Then we define $U$ and $W$ as 
\begin{align}\label{eq: defW}
W(\zeta,t)=
\begin{pmatrix}
\zeta & -2 b & 0 & -2 i d\\
-2b&-\zeta & 2 i d & 0 \\
0 & -2 if & \zeta & -2h\\
2i f & 0 & -2h & -\zeta
 \end{pmatrix},
\end{align}
and
\begin{align}
\label{eq: defU}
U(\zeta,t)=
\begin{pmatrix}
t-c & d & i & 0\\ 
-d &c-t & 0 & i\\
i(\zeta-s)+i k & -i (h+b) & t+c& d\\
- i (h+b) & -i (\zeta+s)+i k & - d & -(t+c)
\end{pmatrix}.
\end{align}
Here
\begin{equation} \label{eq: def dc}
\begin{cases}
d=2^{-1/3} q\left(2^{2/3}(2s-t^2)\right),\\
c=-2^{-1/3} u\left(2^{2/3}(2s-t^2)\right) +s^2,\\
b=\frac{1}{4t}\frac{\partial d}{\partial  t} +dc +td,\\
f=2dt^2-c\frac{1}{2t}\frac{\partial  d}{\partial  t}-dc^2-d^3-2d s,\\
h=\frac{1}{4t}\frac{\partial  d}{\partial  t}+dc-td,\\
k=c^2-d^2-s.
\end{cases}
\end{equation}
Since $q$ has no poles on the real axis \cite{HMcL}, $U$ and $W$ are defined for all $t\in \R$.

Note that the Lax pair is an overdetermined system of equations for $\Psi$. However, the compatibility condition is satisfied.
\begin{lemma}
With $U$ and $W$ defined as in \eqref{eq: defW} and \eqref{eq: defU} we have
\begin{equation}\label{eq:comp}
\frac{\partial }{\partial  \zeta} W=\frac{\partial }{\partial  t}U + [U,W].
\end{equation} 
In other words, $U$ and $W$ solve the compatibility condition for \eqref{eq:laxpair}.
\end{lemma}
\begin{proof}
This is intrinsic to the construction of $U$ and $W$ and follows by direct verification. Plugging \eqref{eq: defW} and \eqref{eq: defU} into \eqref{eq:comp} it remains to check the following equalities
\begin{align*}
c' &= 2(h-b)d, \\
d' &= 4b(t-c)-2ds+2dk-2f, \\
   &= 4h(t+c)+2ds-2dk+2f, \\
b-h &= 2dt, \\
k' &= 2(h^2-b^2), \\
h'+b' &= -4ft+2(h-b)(k-s).
\end{align*}
These are all valid by virtue of \eqref{eq: def dc}.
\end{proof}

The main idea of our proof of Theorem \ref{th: tacnode rhp} is the following. By standard arguments of isomonodromic deformation theory, we construct a RH problem out of the Lax pair \eqref{eq:laxpair} for which the jump matrices do not depend on the parameters.  It is crucial that the Hastings-McLeod solution to the Painlev\'e II equation has no real poles, so that we get a solution for this RH problem for every $s,t\in \R$. We then prove that this solution  also solves the RH problem \ref{rhp: tacnode rhp}, which finishes the proof of Theorem \ref{th: tacnode rhp}.  To this end we need to check that both the jump conditions and the asymptotic condition are satisfied. The asymptotic condition follows by construction and more specifically the top equation in \eqref{eq:laxpair}. For the jump conditions we use the fact that they are independent of $t$  which allows us to simply  refer to   \cite{DKZ} for the special case $t=0$.

\subsection{Asymptotics}

We will now show that the top equation of \eqref{eq:laxpair} has a unique formal solution near the irregular singularity $\infty $ that has the asymptotic behavior given in RH problem \ref{rhp: tacnode rhp}. It turns out that this equation exhibits a Stokes phenomenon. This means that to find a genuine solution of the equation with the same asymptotic behavior as the formal solution, we need to restrict ourselves to certain sectors in the complex plane. To ensure existence of such solutions, these sectors should be sufficiently small. On the other hand, to get uniqueness the sector needs to be sufficiently large.

Motivated by the asymptotic behavior in RH problem \ref{rhp: tacnode rhp} we define the function $\mathcal A$ by
\begin{equation} \label{eq: cal A}
\mathcal A(\zeta):=B(\zeta) A \diag \left( e^{-\psi(-\zeta)+t\zeta},e^{-\psi(\zeta)-t \zeta}, e^{\psi(-\zeta)+t \zeta},e^{\psi(\zeta)-t\zeta} \right),
\end{equation}
where $A$ and $B(\zeta)$ are given by \eqref{eq: A} and \eqref{eq: B} and
\[
\psi(\zeta)=\tfrac23 \zeta^{3/2}+2s \zeta^{1/2}.
\]
Note that $\mathcal A$ has branch cuts along the positive and the negative real line.  

Now consider the sets
\[
\ell_{i,j}=\{ \zeta \in \C \mid \Re (\widetilde \psi_i(\zeta)-\widetilde \psi_j(\zeta))=0\}, \qquad 1 \leq i < j \leq 4,
\]
where 
\[
\widetilde \psi_1(\zeta)=-(-\zeta)^{3/2}, \quad \widetilde \psi_2(\zeta)=-\zeta^{3/2}, \quad \widetilde \psi_3(\zeta)=(-\zeta)^{3/2}, \quad \widetilde \psi_4(\zeta)=\zeta^{3/2}.
\]
Then each $\ell_{i,j}$ is a union of rays that connect the origin with $\infty$, see also Figure \ref{fig: stokes rays}. These rays are called the Stokes lines associated with $\A(\zeta)$. They separate the complex plane into sectors of angle $\pi/6$. The scheme of dominance is indicated by the numbers between successive Stokes lines, e.g. the sequence $4,1,3,2$ in the sector $0<\arg \zeta < \pi/6$ means that $\Re \widetilde \psi_4(\zeta)>\Re\widetilde \psi_1(\zeta)>\Re\widetilde \psi_3(\zeta)>\Re\widetilde \psi_2(\zeta)$ for $\zeta$ in that sector. Note that the Stokes lines on $\R$ are special in the sense that the order of dominance does not change when crossing these lines. This is due to the branch cut of $\A(\zeta)$ along the real line.

We now define sectors $\Omega_k$, $k=0,\ldots 5$ in the following way
\begin{equation} \label{eq: sectors}
\Omega_k =\left\{\zeta \in \C \setminus \{0\} \mid -\frac{\pi}{12}+k\frac{\pi}{3}<\arg \zeta < \frac{7\pi}{12}+k\frac{\pi}{3} \right\}.
\end{equation}
Note that these sectors cover $\C \setminus \{0\}$ and that there is overlap between $\Omega_k$ and $\Omega_{k+1}$. See also Figure \ref{fig: stokes rays}. Note also that the sectors $\Omega_j$, $j=0,2,3,5$ contain the positive or negative real line on which $\A(\zeta)$ has a branch cut.

In order to deal with the technicalities coming from the branch cuts in the definition of $\mathcal A$ we introduce some extra definitions. First,  we define $\A^+(\zeta)$ as the analytic continuation of $\mathcal A(\zeta)$ in the upper half plane to $\C \setminus [0,-i\infty)$ and $\mathcal A^-(\zeta)$ as the analytic continuation of $\mathcal A(\zeta)$ in the lower half plane to $\C \setminus [0,i\infty)$. In the same way as we did before for $\A$ we can construct Stokes lines and a scheme of dominance.  For $\A^+(\zeta)$ we find
\[
\ell^+_{i,j}=\{ \zeta \in \C \mid \Re (\widetilde \psi^+_i(\zeta)-\widetilde \psi^+_j(\zeta))=0\}, \qquad 1 \leq i < j \leq 4,
\]
where $\widetilde \psi_j^+(\zeta)$ is the analytic continuation of $\widetilde \psi_j(\zeta)$ in the upper half plane to $\C \setminus [0,-i \infty)$. The scheme of dominance for $\A^+(\zeta)$ is shown in Figure \ref{fig: stokes rays +}. We leave the details for $\A^-(\zeta)$ to the reader.  Note that the sectors $\Omega_0$, $\Omega_1$, and $\Omega_2$ contain precisely one Stokes ray from $\ell^+_{i,j}$ for all pairs $i<j$. We refer to such sectors as Stokes sectors associated with $\A^+$. Analogously $\Omega_3$, $\Omega_4$, and $\Omega_5$ contain precisely one Stokes ray from $\ell^-_{i,j}$ for all pairs $i<j$. 

We then have the following result on the existence and uniqueness of solutions of the top equation in \eqref{eq:laxpair} with the desired asymptotic behavior.

\begin{lemma} \label{lemma: asymptotics Y}
Fix $t\in \R$ and let $U$ be as defined in \eqref{eq: defU}. Moreover, let $\Omega_k$  be one of the sectors defined in \eqref{eq: sectors}. Then the equation
\begin{equation}\label{eq: ode U}
\frac{{\partial  }N}{\partial  \zeta}(\zeta)=U(\zeta,t) N(\zeta),
\end{equation}
has a unique fundamental solution $N$ in the sector $\Omega_k$ such that $N(\zeta)$ has the following uniform asymptotics as $\zeta \to \infty$ within this sector
\begin{equation} \label{eq: asymptotics N}
N(\zeta)=\begin{cases} \left( I+\frac{N_1}{\zeta}+\mathcal O(\zeta^{-2})\right) \A^+(\zeta) & k=0,1,2, \\
                       \left( I+\frac{N_1}{\zeta}+\mathcal O(\zeta^{-2})\right) \A^-(\zeta) & k=3,4,5.
\end{cases} \end{equation}                            
Here, $N_1$ is given by
\begin{equation} \label{eq: N1}
N_1= \begin{pmatrix} * & b & * & id \\
-b & * & id & * \\
* & if & * & h \\
if & * & -h & * \end{pmatrix},
\end{equation}
where $*$ denotes (for our purposes) unimportant entries.  
\end{lemma}

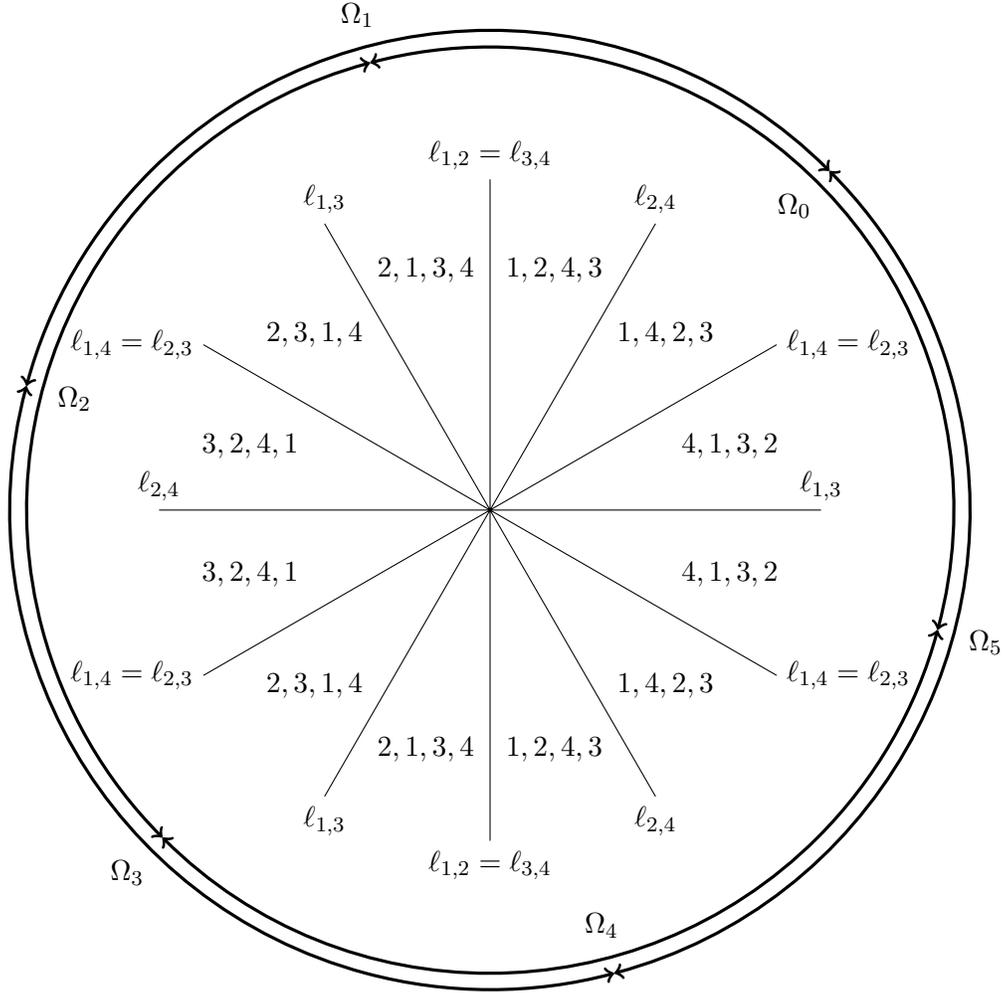
\begin{figure}
\begin{center}
\begin{tikzpicture}[scale=2.2]
\draw[<->,very thick] (-15:2.8) arc (-15:105:2.8);
\draw[<->,very thick] (105:2.8) arc (105:225:2.8);
\draw[<->,very thick] (225:2.8) arc (225:345:2.8);
\draw[<->,very thick] (45:2.9) arc (45:165:2.9);
\draw[<->,very thick] (165:2.9) arc (165:285:2.9);
\draw[<->,very thick] (-75:2.9) arc (-75:45:2.9); 
\draw (0,0) (15:1.5) node[fill=white]{$4, 1,3,2$}
      (0,0) (45:1.5) node[fill=white]{$1,4,2,3$}
      (0,0) (75:1.5) node[fill=white]{$1,2,4,3$} 
      (0,0) (105:1.5) node[fill=white]{$2,1,3,4$} 
      (0,0) (135:1.5) node[fill=white]{$2,3,1,4$} 
      (0,0) (165:1.5) node[fill=white]{$3,2,4,1$}
      (0,0) (-15:1.5) node[fill=white]{$4, 1,3,2$}
      (0,0) (-45:1.5) node[fill=white]{$1,4,2,3$}
      (0,0) (-75:1.5) node[fill=white]{$1,2,4,3$} 
      (0,0) (-105:1.5) node[fill=white]{$2,1,3,4$} 
      (0,0) (-135:1.5) node[fill=white]{$2,3,1,4$} 
      (0,0) (-165:1.5) node[fill=white]{$3,2,4,1$};

\draw (0,0)--(2,0)  node[above]{$\ell_{1,3}$}
      (0,0)--(30:2) node[right]{$\ell_{1,4}=\ell_{2,3}$}
      (0,0)--(60:2) node[above]{$\ell_{2,4}$}
      (0,0)--(90:2) node[above]{$\ell_{1,2}=\ell_{3,4}$}
      (0,0)--(120:2) node[above]{$\ell_{1,3}$}
      (0,0)--(150:2) node[left]{$\ell_{1,4}=\ell_{2,3}$}
      (0,0)--(180:2) node[above]{$\ell_{2,4}$}
      (0,0)--(-30:2) node[right]{$\ell_{1,4}=\ell_{2,3}$}
      (0,0)--(-60:2) node[below]{$\ell_{2,4}$}
      (0,0)--(-90:2) node[below]{$\ell_{1,2}=\ell_{3,4}$}
      (0,0)--(-120:2) node[below]{$\ell_{1,3}$}
      (0,0)--(-150:2) node[left]{$\ell_{1,4}=\ell_{2,3}$};
\draw (45:2.6) node{$\Omega_0$}
      (165:2.6) node{$\Omega_2$}
      (285:2.6) node{$\Omega_4$}
      (105:3.1) node{$\Omega_1$}
      (225:3.1) node{$\Omega_3$}
      (-15:3.1) node{$\Omega_5$};
\end{tikzpicture}
\end{center}
\caption{Stokes diagram for the irregular singular point $\infty$ of the differential equation \eqref{eq: ode U} associated with $\A(\zeta)$. The stokes lines $\ell_{i,j}$ separate the complex plane into sectors of angle $\pi/6$. The scheme of dominance is indicated by the numbers between successive Stokes lines, e.g. the sequence $4,1,3,2$ in the sector $0<\arg \zeta < \pi/6$ means that $\Re\widetilde \psi_4(\zeta)>\Re\widetilde \psi_1(\zeta)>\Re\widetilde \psi_3(\zeta)>\Re \widetilde \psi_2(\zeta)$ for $\zeta$ in that sector.}
\label{fig: stokes rays}
\end{figure}

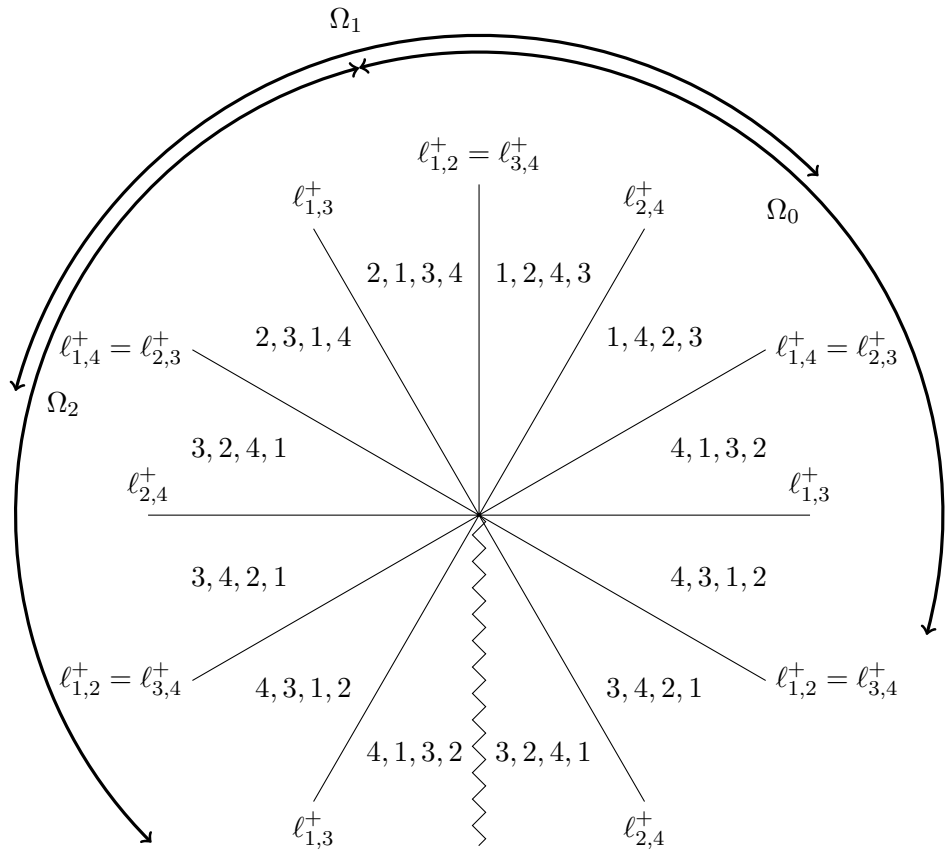
\begin{figure}
\begin{center}
\begin{tikzpicture}[scale=2.2]
\draw[<->,very thick] (-15:2.8) arc (-15:105:2.8);
\draw[<->,very thick] (105:2.8) arc (105:225:2.8);
\draw[<->,very thick] (45:2.9) arc (45:165:2.9);
\draw (0,0) (15:1.5) node[fill=white]{$4, 1,3,2$}
      (0,0) (45:1.5) node[fill=white]{$1,4,2,3$}
      (0,0) (75:1.5) node[fill=white]{$1,2,4,3$} 
      (0,0) (105:1.5) node[fill=white]{$2,1,3,4$} 
      (0,0) (135:1.5) node[fill=white]{$2,3,1,4$} 
      (0,0) (165:1.5) node[fill=white]{$3,2,4,1$}
      (0,0) (-15:1.5) node[fill=white]{$4, 3,1,2$}
      (0,0) (-45:1.5) node[fill=white]{$3,4,2,1$}
      (0,0) (-75:1.5) node[fill=white]{$3,2,4,1$} 
      (0,0) (-105:1.5) node[fill=white]{$4,1,3,2$} 
      (0,0) (-135:1.5) node[fill=white]{$4,3,1,2$} 
      (0,0) (-165:1.5) node[fill=white]{$3,4,2,1$};
\draw (0,0)--(2,0)  node[above]{$\ell^+_{1,3}$}
      (0,0)--(30:2) node[right]{$\ell^+_{1,4}=\ell^+_{2,3}$}
      (0,0)--(60:2) node[above]{$\ell^+_{2,4}$}
      (0,0)--(90:2) node[above]{$\ell^+_{1,2}=\ell^+_{3,4}$}
      (0,0)--(120:2) node[above]{$\ell^+_{1,3}$}
      (0,0)--(150:2) node[left]{$\ell^+_{1,4}=\ell^+_{2,3}$}
      (0,0)--(180:2) node[above]{$\ell^+_{2,4}$}
      (0,0)--(-30:2) node[right]{$\ell^+_{1,2}=\ell^+_{3,4}$}
      (0,0)--(-60:2) node[below]{$\ell^+_{2,4}$}
      (0,0)--(-120:2) node[below]{$\ell^+_{1,3}$}
      (0,0)--(-150:2) node[left]{$\ell^+_{1,2}=\ell^+_{3,4}$};
\draw[decorate,decoration=zigzag](0,0)--(-90:2);
\draw (45:2.6) node{$\Omega_0$}
      (165:2.6) node{$\Omega_2$}
      (105:3.1) node{$\Omega_1$};
\end{tikzpicture}
\end{center}
\caption{Stokes diagram for the irregular singular point $\infty$ of the differential equation \eqref{eq: ode U} associated with $\A^+(\zeta)$. The stokes lines $\ell^+_{i,j}$ separate the complex plane into sectors of angle $\pi/6$. The scheme of dominance is indicated by the numbers between successive Stokes lines. The zigzag line indicates the branch cut of $\A^+(\zeta)$.}
\label{fig: stokes rays +}
\end{figure}

\begin{proof}
We prove the lemma for $k=0,1,2$ and leave the other case to the reader.

Infinity is an irregular singular point of the ODE \eqref{eq: ode U}. Asymptotics of solutions to equations of this type can always be computed following an algorithm presented in the monograph \cite{Wa}. When we consider $U$ as a linear polynomial in $\zeta$, its leading coefficient matrix is nilpotent. First we will transform this system in order to obtain a leading coefficient matrix which has four different eigenvalues. As soon as we are in that situation we can invoke a general result, see e.g. \cite{Wa} or \cite[Th. 1.4]{FIKN}, to obtain asymptotics of the transformed system. 

Thus, we start with the transformation 
\begin{equation} \label{eq: transfo NZ}
N(\zeta)= B^+(\zeta) A \tilde Z(\zeta),
\end{equation}
where $B^+(\zeta)$ is the analytic continuation of $B(\zeta)$, see \eqref{eq: B}, in the upper half plane to $\C \setminus [0,-i\infty)$.
We supplement the transformation by the change of variables $\zeta=z^2$, where we choose $z$ such that $-\pi/4<\arg z < 3\pi/4$, and call $Z(z)=\tilde Z(\zeta)$. In this way we obtain the equivalent equation
\begin{equation} \label{eq: ode E}
\frac{1}{z^2} \frac{\partial Z}{\partial z}=E(z)Z(z),
\end{equation}
where
\[
E(z)=E_0+\frac{1}{z}E_1+\frac{1}{z^2}E_2+\frac{1}{z^3}E_3. 
\]
The  matrices $E_j$, $j=0,1,2,3$ can be explicitly found by a lengthy computation and are constant
\begin{align*}
 E_0 &=\diag \left( -2 i, -2, 2i,2\right)=\diag \left( \lambda_1,\lambda_2,\lambda_3,\lambda_4 \right), \end{align*}\begin{align*}
 E_1 &= \begin{pmatrix}
2t &  -\sqrt 2 id &  2ic & \sqrt 2 id \\
-\sqrt 2 id & -2t &  \sqrt 2 id & 2ic \\
-2ic & \sqrt 2 id &  2t &  \sqrt 2 id \\
\sqrt 2 id &-2ic &  \sqrt 2 id & -2t \end{pmatrix},  \end{align*}\begin{align*}
 E_2 &= \begin{pmatrix}
 i(s-k) & \exp( \pi i/4) (b+h) &  s-k &  i\exp( \pi i/4) (b+h) \\
-\exp( \pi i/4) (b+h) & k-s &  i\exp( \pi i/4) (b+h) & i(k-s) \\
s-k &  -i\exp( \pi i/4) (b+h) &  i(k-s) &  \exp( \pi i/4) (b+h) \\
-i\exp( \pi i/4) (b+h) & i(k-s) & -\exp( \pi i/4) (b+h) &  s-k
\end{pmatrix},  \end{align*}\begin{align*}
 E_3 &= \begin{pmatrix}
0 & 0& -i/2 & 0\\
0&0&0&i/2 \\
i/2 & 0 & 0 & 0 \\
0 & -i/2 & 0 & 0
\end{pmatrix}.
\end{align*}
Hence $E_0$ indeed has four different eigenvalues, which was precisely the goal of the transformation.

The ODE \eqref{eq: ode E} is of the correct form to apply \cite[Prop. 1.1]{FIKN}. This proposition states that \eqref{eq: ode E} admits a unique formal fundamental solution
\begin{equation} \label{eq: formal solution Z}
Z_f(z)=\left( \sum_{k=0}^\infty Z_k z^{-k} \right) z^{\Lambda_0}e^{\Lambda_3 \frac{z^3}{3}+\Lambda_2 \frac{z^2}{2}+\Lambda_1 z},
\end{equation}
where $Z_0=I$, $\Lambda_k$, $k=0,1,2,3$, are diagonal, and $\Lambda_3=E_0$. Moreover, all coefficients $\Lambda_k$, $k=0,1,2$, and $Z_k$, $k=1,2,3,\ldots$, can be determined recursively using the explicit form of $E_j$, $j=1,\ldots,4$. 

If we unfold the transformation \eqref{eq: transfo NZ} this means that 
\eqref{eq: ode U} has a unique formal solution for $\zeta \in \C \setminus [0,-i \infty)$ of the form
\begin{equation} \label{eq: formal solution N}
N_f(\zeta)=\left(\sum_{k=0}^\infty N^{(k)} \zeta^{-k/2}\right) B^+(\zeta) A \zeta^{\Lambda_0/2}
e^{\Lambda_3 \frac{\zeta^{3/2}}{3}+\Lambda_2 \frac{\zeta}{2}+\Lambda_1 \zeta^{1/2}}.
\end{equation}
Here the $N_k$ are constant matrices, that can be recursively found from the matrices $Z_k$.

To determine all coefficients,  we rewrite (following \cite{FIKN}) \eqref{eq: formal solution Z} in the equivalent form
\begin{equation} \label{eq: formal solution 2}
Z_f(z)=\left( \sum_{k=0}^\infty Y_k z^{-k} \right) e^{\Delta(z)}, \qquad Y_0=I,
\end{equation}
where all matrices $Y_k$, $k=1,2,3,\ldots$ are off-diagonal and the diagonal matrix $\Delta(z)$ is given by the series
\[
\Lambda_3 \frac{z^3}{3}+\Lambda_2 \frac{z^2}{2}+\Lambda_1 z+\log \Lambda_0(z)+\sum_{k=1}^\infty \Lambda_{-k}\frac{z^{-k}}{-k}.
\]
Now substituting \eqref{eq: formal solution 2} into \eqref{eq: ode E} and equating coefficients of the same powers of $z$ yields a system of equations for the diagonal matrices $\Lambda_k$ and the off-diagonal matrices $Y_k$. This system can be recursively solved, which leads to
\begin{equation} \label{eq: Lambda Y}
\Lambda_{3-k}= \diag F_{3-k}, \qquad \left( Y_k \right)_{i,j}=\frac{\left(F_{3-k}\right)_{i,j}}{\lambda_j-\lambda_i}, \quad i \neq j,
\end{equation}
where $F_{3-k}$ are defined by the equations
\begin{equation} \label{eq: FF}
F_{3-k}=E_k+\sum_{m=1}^{k-1}\left(E_{k-m}Y_m-Y_m\Lambda_{3-k+m}\right)-(k-3)Y_{k-3}, \qquad k=0,1,2,\ldots.
\end{equation}
Here we have put $E_k=0$, $k>3$, and $Y_k=0$, $k<0$. In this way we find e.g.
\begin{align*}
\Lambda_3 &= E_0, \\ 
\Lambda_2 &= \diag \left(2t,-2t,2t,-2t\right),  \\
\Lambda_1 &=  \diag \left( 2 i s,-2s,2 i s,2s \right), \\
\Lambda_0 &= 0,
\end{align*}
where we used the last line of \eqref{eq: def dc}.

By  continuing this approach we can in principle obtain all coefficients in \eqref{eq: formal solution N}. We claim that all coefficients $N^{(2l+1)}$, $l=0,1,2,\ldots$ vanish. To see this, observe that the matrices $AE_kA^{-1}$ are block diagonal (with blocks of size $2 \times 2$ ) when $k$ is odd and block anti-diagonal for even $k$. It then follows by \eqref{eq: Lambda Y}--\eqref{eq: FF} that $AY_kA^{-1}$ and $A\Lambda_kA^{-1}$ are block diagonal for even $k$ and block anti-diagonal for odd $k$. In the context of \eqref{eq: formal solution Z} this means that also $Z_k$ is block diagonal for even $k$ and block anti-diagonal for odd $k$. Taking this into account when unfolding the transformation \eqref{eq: transfo NZ} proves the claim.

Now we proved that \eqref{eq: ode U} admits a formal solution of the correct form. The fact that in each sector $\Omega_k$, $k=0,1,2$ this formal solution actually has the asymptotic meaning stated in the lemma follows from the choice of the sectors $\Omega_k$, see  \cite[Th. 1.4]{FIKN}.

It remains to prove that $N^{(2)}=N_1$ as in \eqref{eq: N1}. This follows from this approach by detailed but straightforward calculations in which we also use \eqref{eq: def dc}. Alternatively, \eqref{eq: N1} can be checked by substituting the asymptotic formula \eqref{eq: asymptotics N} into \eqref{eq: ode U}.
\end{proof}

\subsection{Proof of Theorem \ref{th: tacnode rhp}}

We are now ready to prove Theorem \ref{th: tacnode rhp}. 
\begin{proof}[Proof of Theorem \ref{th: tacnode rhp}]
The proof follows by standard isomonodromy deformation arguments and the results in \cite{DKZ} for the case $t=0$.  For more details on the general theory we refer to \cite{FIKN}.

First we remark that without loss of generality we may assume that the angles in RH problem \ref{rhp: tacnode rhp} coincide $\varphi_1=\varphi_2=\pi/3$ such that $\Sigma_M=\cup_{k=0}^5 e^{k\pi i/3}[0,\infty)$ with constant jump matrices on each of these rays. The rays separate six sectors
\[
\Delta_k=\{\zeta \in \C \setminus \{0\} \mid k\tfrac{\pi}{3} < \arg \zeta < (k+1)\tfrac{\pi}{3}\}.
\]

From Lemma \ref{lemma: asymptotics Y} we learn that in each sector $\Omega_k$  we can find  a unique solution $\Psi_k(\zeta;t)$ to the top equation of \eqref{eq:laxpair} that has the asymptotics given in \eqref{eq: asymptotics N} as $\zeta \to \infty$ within $\Omega_k$. Note that if we let $\zeta \to \infty$ in the smaller sector $\Delta_k$ this coincides with the (uniform) asymptotics in the RH problem. Moreover these functions can be analytically continued on the full complex plane. The solutions associated with different sectors differ by a matrix $S_k(t)=\Psi_{k}^{-1} \Psi_{k+1}$, $k=0,\ldots,5$, that is independent of $\zeta$. Finally, we define a function $M(\zeta;t)$ by
\[
M(\zeta;t)=\Psi_k(\zeta;t), \qquad \zeta \in \Delta_k.
\]
Summarizing, we have constructed a function $M(\zeta;t)$ with a jump contour on $\Sigma_M$ that satisfies the asymptotic condition as $\zeta \to \infty$.  lt remains to show that the jump matrices $S_k(t)$ on $\Sigma_M$ coincide with the jump matrices in the RH problem. 

For $t=0$ this follows from the results in \cite{DKZ}. We know from $\cite{DKZ}$ that  RH problem \ref{rhp: tacnode rhp} has a solution $\widehat M(\zeta)$ for $t=0$ and that this solution satisfies \eqref{eq: ode U}. Now define $\widehat{\Psi}_k(\zeta)$ as the analytic continuation of $\widehat{M}(\zeta)$ in the sector $\Delta_k$ to the full complex plane. It can then be checked that $\widehat{\Psi}_k(\zeta)$ uniformly satisfies the asymptotic condition stated in \eqref{eq: asymptotics N} as $\zeta \to \infty$ within $\Omega_k$, and as a result we have $\widehat \Psi_k(\zeta)=\Psi_k(\zeta;0)$ and hence $\widehat M(\zeta)=M(\zeta;0)$. This means that for $t=0$, we indeed have that $S_k(0)$ are the jump matrices in the RH problem \ref{rhp: tacnode rhp}. 

By using standard arguments from the theory of isomonodromy deformations we will now show that the Stokes matrices $S_k$ are independent of $t$ and that finishes the proof for general $t\in \R$. To this end we note that from  \eqref{eq:comp} it follows that $\tilde \Psi_k=\partial \Psi_k/\partial t -W \Psi_k$ also solves the top equation of \eqref{eq:laxpair}. This implies that there exists a matrix $C(t)$, not depending on $\zeta$, such that $\tilde \Psi_k=\Psi_k C(t)$ and hence 
\begin{align*}
C(t)=\Psi^{-1}_k \partial \Psi_k /\partial t-\Psi_k^{-1} W \Psi_k.
\end{align*}
We then show that $C(t)$ vanishes. We assume here $k=0,1,2$. The adaptations to cover $k=3,4,5$ are straightforward.
Computing the asymptotic behavior for $\zeta \to \infty$ of the right-hand side, using also \eqref{eq: N1}, yields
\begin{multline*}
C(t)=\diag \left( e^{\psi(-\zeta)-t\zeta},e^{\psi(\zeta)+t \zeta}, e^{-\psi(-\zeta)-t \zeta},e^{-\psi(\zeta)+t\zeta} \right)^+ \mathcal O \left( \zeta^{-1/2} \right)\\\times \diag \left( e^{-\psi(-\zeta)+t\zeta},e^{-\psi(\zeta)-t \zeta}, e^{\psi(-\zeta)+t \zeta},e^{\psi(\zeta)-t\zeta} \right)^+, \qquad \text{as $\zeta \to \infty$ in $\Omega_k$,}
\end{multline*}
where (again) the superscript $+$ indicates the analytic continuation of the function in the upper half plane to $\C \setminus [0,-i \infty)$.
Letting $\zeta \to \infty$ within $\Omega_k$ this already shows that $C(t)$ has zero diagonal. To see that also the $(i,j)$ off-diagonal entry of $C(t)$ is zero it is crucial that $\Omega_k$ has a Stokes ray in $\ell^+_{i,j}$ in its interior, see Figure \ref{fig: stokes rays}. Therefore $\Omega_k$ contains a subsector in which $\Re(\widetilde \psi_j^+(\zeta)-\widetilde \psi_i^+(\zeta))<0$. Letting $\zeta$ approach infinity in this sector proves that $C(t)_{i,j}=0$. Hence $C=0$.

We  thus proved that $\Psi_k$ also solves the bottom equation of \eqref{eq:laxpair}. Therefore  the Stokes matrices are independent of $t$. Indeed
\begin{align}
\begin{split}
\partial \Psi_k /\partial t \Psi_{k}^{-1}&=W(\zeta,t)= \partial \Psi_{k+1} /\partial t \Psi_{k+1}^{-1}\\&= \partial \Psi_k /\partial t \Psi_{k}^{-1}+\Psi_k \partial S_k  /\partial t S_k^{-1} \Psi_{k}^{-1},
\end{split}
\end{align}
and hence $\partial S_k/\partial t=0$. This finishes the proof.
\end{proof}

\section{A double scaling limit for $K_{\rm cr}$}

In this section we prove Theorem \ref{th: PII}. The proof is again based on a Deift/Zhou steepest descent analysis on the RH problem \ref{rhp: tacnode rhp} for $M$.  Since the analysis is  much simpler than the one in Section 4, we will allow ourselves to be brief. We will give all the transformations, but we will  be short about the proofs.

The steepest descent analysis consists of  a sequence of transformations
\begin{align*}
M \mapsto M^{(1)}\mapsto M^{(2)} \mapsto M^{(3)} \mapsto M^{(4)}.
\end{align*}
In the end, the asymptotics for  $M^{(4)}$ (as $a\to \infty$) can be easily found. The proof of the statement then follows by updating the kernel with each transformation. The core of the proof is that at a certain point, we need to construct a local parametrix, which is done using RH problem \ref{rhp: PII rhp}. 

\subsection{First transformation: $M\mapsto M^{(1)}$}

In Theorem \ref{th: PII} the solution  $M(z)=M(z;s,t)$ to RH problem \ref{rhp: tacnode rhp} appears in the critical kernel for the values of parameters
\begin{align*}
\begin{cases}
r_1=r_2=1, \\ t=-a\left(1-\frac{\sigma}{a^2}\right),\\  s=s_1=s_2=\frac{a^2}{2},
\end{cases}
\end{align*}
where $\sigma\in \R$ is fixed.

We define the first transformation as
\begin{equation}\label{eq:defM(1)}
M^{(1)}(z;a)=\diag \left( a^{1/2},a^{1/2},a^{-1/2},a^{-1/2}\right) M\left(a^2z;\frac{a^2}{2},-a\left(1-\frac{\sigma}{a^2}\right)\right).
\end{equation}
Then from RH problem \ref{rhp: tacnode rhp} for $M$, we straightforwardly check that $M^{(1)}$ satisfies the following RH problem.
\begin{lemma}
The function  $M^{(1)}$ defined in \eqref{eq:defM(1)} has the following properties
\begin{itemize}
\item[\rm (1)] $M^{(1)}$ is analytic for  $\zeta \in \C \setminus \Sigma_M$;
\item[\rm (2)] $M^{(1)}_+(z)=M^{(1)}_-(z)J_k$, for $z \in \Gamma_k,$ $k=0,\ldots,9$;
\item[\rm (3)] As $z\to \infty$ we have
\begin{multline*} 
M^{(1)}(z;a)=\left(I+\mathcal O \left(z^{-1}\right) \right) \diag \left((-z)^{-1/4},z^{-1/4},(-z)^{1/4},z^{1/4} \right)A \\
\times  \diag \left( e^{a^3\left(- \widetilde \psi (-z)-(1-\sigma/a^2)z \right)},e^{a^3\left(- \widetilde \psi (z)+(1-\sigma/a^2)z \right)},e^{ a^3\left(\widetilde \psi (-z)-(1-\sigma/a^2)z \right)},e^{a^3\left(\widetilde \psi (z)+(1-\sigma/a^2)z \right)}\right),
\end{multline*}
where
\begin{align*}
\widetilde \psi(z)=\frac23z^{3/2}+ z^{1/2};
\end{align*}
\item[\rm (4)] $M^{(1)}(z)$ is bounded near $z=0$.
\end{itemize}
\end{lemma}

\subsection{Second transformation: $M^{(1)}\mapsto M^{(2)}$}

In the second transformation we deform the  contours. We shift the rays $\Gamma_1$ and $\Gamma_9$ of the contour $\Sigma_M$ to the right by one, while shifting $\Gamma_4$ and $\Gamma_6$ to the left by one. More precisely, we introduce the new rays 
\[
\widetilde \Gamma_{1,9}=\Gamma_{1,9}+1, \text{ and } \widetilde \Gamma_{4,6} =\Gamma_{4,6}-1.\]
Moreover, for $k=0,2,3,5,7,8$ we set $\widetilde \Gamma_k=\Gamma_k$ and define
\begin{align*}
\Sigma_{M^{(2)}}=  \bigcup_{k=0}^9 \widetilde \Gamma_k. 
\end{align*}
We define the function $M^{(2)}$ as follows. For $k=1,4,6,9$ it is defined by
\begin{align}\label{eq:defM(2)}
M^{(2)}(z)=\begin{cases}
M^{(1)}(z) J_k & k=1,6, \\
M^{(1)}(z) J_k^{-1} & k=4,9, \end{cases}
\end{align}
 for $z$ in the region enclosed by $\Gamma_k$, $\widetilde \Gamma_k$ and $\R$. In the other regions we simply set $M^{(2)}=M^{(1)}$. 
 
 \begin{lemma}
The function  $M^{(2)}$ as defined \eqref{eq:defM(2)} has the following properties
\begin{itemize}
\item[\rm (1)] $M^{(2)}$ is analytic for  $\zeta \in \C \setminus \Sigma_{M^{(2)}}$;
\item[\rm (2)] $M^{(2)}_+(z)=M^{(2)}_-(z)J_{M^{(2)}}$, for $z \in\Sigma_{M^{(2)}}$;
\item[\rm (3)] $M^{(2)}(z)$ has the same asymptotics as $M^{(1)}$ as $z\to \infty$;
\item[\rm (4)] $M^{(2)}(z)$ is bounded near $z=0$.
\end{itemize}
See  Figure \ref{fig: contour D} for the jump matrices $J_{M^{(2)}}$. Note that we reversed the orientation of some rays.
\end{lemma}

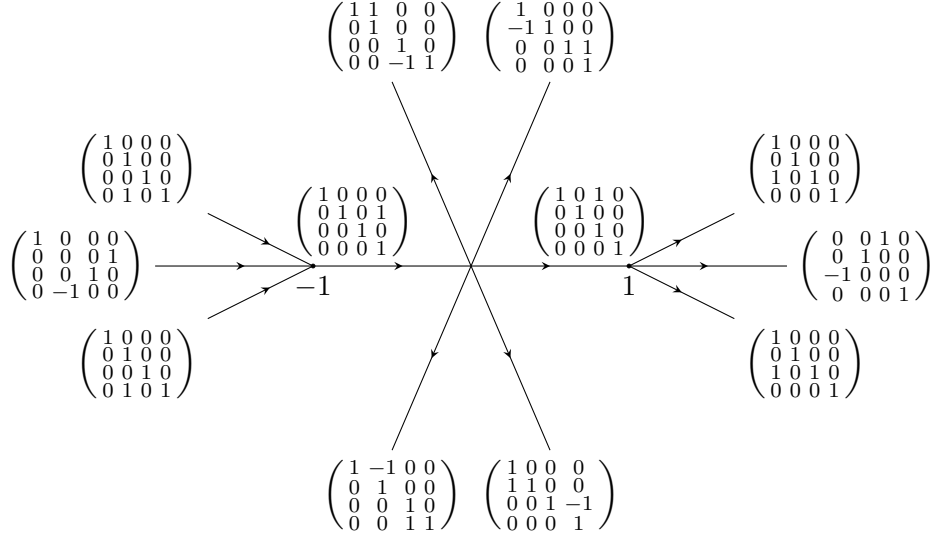
\begin{figure}[t]
\centering
\begin{tikzpicture}[scale=0.7]
\begin{scope}[decoration={markings,mark= at position 0.5 with {\arrow{stealth}}}]
\draw[postaction={decorate}]      (0,0)--node[near end, above]{$\left(\begin{smallmatrix} 1&0&1&0 \\0&1&0&0\\0&0&1&0\\0&0&0&1 \end{smallmatrix}\right)$}(3,0);
\draw[postaction={decorate}]      (3,0)--(6,0) node[right]{$\left(\begin{smallmatrix} 0&0&1&0\\0&1&0&0\\-1&0&0&0\\0&0&0&1 \end{smallmatrix}\right)$};
\draw[postaction={decorate}]      (3,0)--(5,1) node[above right]{$\left(\begin{smallmatrix} 1&0&0&0\\0&1&0&0\\1&0&1&0\\0&0&0&1 \end{smallmatrix}\right)$};
\draw[postaction={decorate}]      (0,0)--(1.5,3.5) node[above]{$\left(\begin{smallmatrix} 1&0&0&0\\-1&1&0&0\\0&0&1&1\\0&0&0&1 \end{smallmatrix}\right)$};
\draw[postaction={decorate}]      (0,0)--(-1.5,3.5) node[above]{$\left(\begin{smallmatrix} 1&1&0&0\\0&1&0&0\\0&0&1&0\\0&0&-1&1 \end{smallmatrix}\right)$};
\draw[postaction={decorate}]      (0,0)--(-1.5,-3.5)node[below]{$\left(\begin{smallmatrix} 1&-1&0&0\\0&1&0&0\\0&0&1&0\\0&0&1&1 \end{smallmatrix}\right)$};
\draw[postaction={decorate}]      (0,0)--(1.5,-3.5)node[below]{$\left(\begin{smallmatrix} 1&0&0&0\\1&1&0&0\\0&0&1&-1\\0&0&0&1 \end{smallmatrix}\right)$};
\draw[postaction={decorate}]      (3,0)--(5,-1)node[below right]{$\left(\begin{smallmatrix} 1&0&0&0\\0&1&0&0\\1&0&1&0\\0&0&0&1 \end{smallmatrix}\right)$};
\end{scope}
\begin{scope}[decoration={markings,mark= at position 0.5 with {\arrowreversed{stealth}}}]
\draw[postaction={decorate}]      (-3,0)--(-5,1) node[above left]{$\left(\begin{smallmatrix} 1&0&0&0\\0&1&0&0\\0&0&1&0\\0&1&0&1 \end{smallmatrix}\right)$} ;
\draw[postaction={decorate}]      (0,0)--node[near end, above]{$\left(\begin{smallmatrix} 1&0&0&0 \\0&1&0&1\\0&0&1&0\\0&0&0&1 \end{smallmatrix}\right)$}(-3,0);
\draw[postaction={decorate}]      (-3,0)--(-6,0) node[left]{$\left(\begin{smallmatrix} 1&0&0&0\\0&0&0&1\\0&0&1&0\\0&-1&0&0 \end{smallmatrix}\right)$};
\draw[postaction={decorate}]      (-3,0)--(-5,-1) node[below left]{$\left(\begin{smallmatrix} 1&0&0&0\\0&1&0&0\\0&0&1&0\\0&1&0&1 \end{smallmatrix}\right)$} ;
\end{scope}
\filldraw (3,0) circle (1pt) node[below]{$1$} (-3,0) circle (1pt) node[below]{$-1$};
\end{tikzpicture}
\caption{The jump contour and jump matrices for $M^{(2)}(z;a)$.}
\label{fig: contour D}
\end{figure}

\subsection{Third transformation: $M^{(2)}\mapsto M^{(3)}$}
In the third transformation we use $g$-functions to normalize the behavior at infinity. The $g$-functions are defined by
\begin{equation}\label{eq:defgjsteep2}
\begin{aligned}
g_1(z)&=-\frac23(1-z)^{3/2}-(1-\sigma/a^2) z, & g_2(z)&=-\frac23(1+z)^{3/2}+(1-\sigma/a^2) z, \\ 
g_3(z)&=\frac23(1-z)^{3/2}-(1-\sigma/a^2) z, & g_4(z)&=\frac23(1+z)^{3/2}+(1-\sigma/a^2) z.
\end{aligned}
\end{equation}
These functions satisfy the following properties;
\begin{itemize}
\item $g_1$ and $g_3$ are analytic on $\C \setminus [1,+\infty)$ and $g_{1\pm}(z)=g_{3\mp}(z)$ for $z \in [1,+\infty)$;
\item $g_2$ and $g_4$ are analytic on $\C \setminus (-\infty,-1]$ and $g_{2\pm}(z)=g_{4\mp}(z)$ for $z \in (-\infty,-1]$;
\item as $z \to \infty$,
      \begin{align*}
      g_1(z)&=-\widetilde \psi(-z)-( 1-\sigma/a^2)z-\frac14(-z)^{-1/2}+\mathcal O \left( z^{-3/2} \right), \\
      g_2(z)&=-\widetilde \psi(z)+( 1-\sigma/a^2)z-\frac14z^{-1/2}+\mathcal O \left( z^{-3/2} \right), \\
      g_3(z)&=\widetilde \psi(-z)-( 1-\sigma/a^2)z+\frac14(-z)^{-1/2}+\mathcal O \left( z^{-3/2} \right), \\
      g_4(z)&=\widetilde \psi(z)+( 1-\sigma/a^2)z+\frac14z^{-1/2}+\mathcal O \left( z^{-3/2} \right).
      \end{align*}
\end{itemize}
We define $M^{(3)}$ by 
\begin{equation}\label{eq:defM(3)}
M^{(3)}(z;a)=E_0 M^{(2)}(z;a) \diag \left( e^{-a^3g_1(z)},e^{-a^3g_2(z)},e^{-a^3g_3(z)},e^{-a^3g_4(z)}\right),
\end{equation}
where $E_0$ is the constant matrix
\[
E_0=I_4+\frac{i}{4}E_{3,1}-\frac{i}{4}E_{4,2}.
\]

Then $M^{(3)}$ satisfies the following RH problem with $\Sigma_{M^{(3)}}:=\Sigma_{M^{(2)}}$.

\begin{lemma} The function  $M^{(3)}$ as defined in \eqref{eq:defM(3)} has the following properties 
\begin{itemize}
\item[\rm (1)] $M^{(3)}$ is analytic for  $z\in \C \setminus \Sigma_{M^{(3)}}$;
\item[\rm (2)] $M^{(3)}_+(z)=M^{(3)}_-(z)J_{M^{(3)}}$, for $z \in \Sigma_{M^{(3)}}$;
\item[\rm (3)] As $z \to \infty$ we have
\begin{align*}
M^{(3)}(z;a)=\left(I+\mathcal O \left(z^{-1}\right) \right) \diag \left((-z)^{-1/4},z^{-1/4},(-z)^{1/4},z^{1/4} \right)A.
\end{align*}
\end{itemize}
The jump matrices $J_{M^{(3)}}$ are given by 
\begin{multline*}
J_{M^{(3)}}(z)=\diag \left( e^{a^3g_{1-}(z)},e^{a^3g_{2-}(z)},e^{a^3g_{3-}(z)},e^{a^3g_{4-}(z)}\right)\\
\times J_{M^{(2)}}\diag \left( e^{-a^3g_{1+}(z)},e^{-a^3g_{2+}(z)},e^{-a^3g_{3+}(z)},e^{-a^3g_{4+}(z)}\right).
\end{multline*}
And more explicitly,
\begin{align*}
J_{M^{(3)}}(z)=\begin{pmatrix}
1 & 0 & 0 & 0\\
0 & 0 & 0 & 1\\
0 & 0 & 1 & 0\\
0 & -1 & 0 & 0
\end{pmatrix}, & \quad z \in (-\infty,-1), &
J_{M^{(3)}}(z)=\begin{pmatrix}
0 & 0 & 1 & 0\\
0 & 1 & 0 & 0\\
-1 & 0 & 0 & 0\\
0 & 0 & 0 & 1
\end{pmatrix},  \quad z\in (1,\infty), \end{align*}
\begin{align*}
J_{M^{(3)}}(z)=\begin{pmatrix}
1 & 0 & 0 & 0\\
0 & 1 & 0 & e^{a^3(g_2-g_4)}\\
0 & 0 & 1 & 0\\
0 & 0 & 0 & 1
\end{pmatrix}, &  \quad z\in (-1,0),  \end{align*}
\begin{align*}
J_{M^{(3)}}(z)=\begin{pmatrix}
1 & 0 & e^{a^3 (g_1-g_3)} & 0\\
0 & 1 & 0 & 0\\
0 & 0 & 1 & 0\\
0 & 0 & 0 & 1
\end{pmatrix}, &  \quad z\in (0,1),  \end{align*}
\begin{align*}
J_{M^{(3)}}(z)=\begin{pmatrix}
1 & 0 & 0 & 0\\
0 & 1 & 0 & 0\\
0 & 0 & 1 & 0\\
0 & e^{a^3(g_4-g_2)} & 0 & 1
\end{pmatrix}, &  \quad \arg (z+1)=\pm (\pi-\varphi_1),  \end{align*}
\begin{align*}
J_{M^{(3)}}(z)=\begin{pmatrix}
1 & 0 & 0 & 0\\
0 & 1 & 0 & 0\\
e^{a^3(g_3-g_1)}  & 0 & 1 & 0\\
0 & 0 & 0 & 1
\end{pmatrix}, &  \quad \arg (z-1)=\pm \varphi_1, \end{align*}
\begin{align*}
J_{M^{(3)}}(z)=\begin{pmatrix}
1 & 0 & 0 & 0\\
\mp e^{a^3 (g_2-g_1)} & 1 & 0 & 0\\
0 & 0 & 1 &  \pm e^{a^3 (g_3-g_4)}\\
0 & 0 & 0 & 1
\end{pmatrix}, &  \quad \arg z=\pm \varphi_2, \end{align*}
and, finally,
\begin{align*}
J_{M^{(3)}}(z)\begin{pmatrix}
1 & \pm e^{a^3 (g_1-g_2)}  & 0 & 0\\
0 & 1 & 0 & \\
0 & 0 & 1 &  0 \\
0 & 0 & \mp e^{a^3 (g_4-g_3)} & 1
\end{pmatrix}, &  \quad \arg z=\pm(\pi-\varphi_2).
\end{align*}

\end{lemma}

Let us discuss the limiting behavior of the jump matrices $J_{M^{(3)}}$  as $a\to +\infty$. First assume that $\sigma=0$. In that case, it is not hard to check from \eqref{eq:defgjsteep2} that in the limit $a\to +\infty$, the jump matrices on the contours other than $(-\infty,-1)$ and $(1,\infty)$ are exponentially close to the identity. To have exponential decay in the upper left block of the jump matrices on $\widetilde \Gamma_k$, $k=2,3,7,8$ we also need $\varphi_2 \geq \pi/6$. This can always be achieved deforming these rays if necessary.  The exponential decay holds uniformly, when we stay away from $0$ and $\pm 1$, suggesting the construction of a global parametrix that solves the limiting RH problem and local parametrices near $0$ and $\pm1$. Now, if $\sigma\neq 0$ but $a$ large, then near the origin it may happen that we have exponential growth of some of the entries of $J_{M^{(3)}}$ instead of decay. However, this only happens very close to the origin where we are going to construct a local parametrix anyway. 

Summarizing, in the next steps we will construct first a global parametrix and then local parametrices around $0$ and $\pm 1$.  For our purposes, the local parametrix at the origin is most important and  it is precisely here that RH problem \ref{rhp: PII rhp} for the Hastings-McLeod solution to the Painlev\'e II equation appears.

\subsection{Construction of the global parametrix $P^{(\infty)}$}

In this part we will construct an explicit solution to the RH problem for $M^{(3)}$ after taking the limit $a \to \infty$. That is, we want to solve the following RH problem.

We look for a $4 \times 4$ matrix-valued function $P^{(\infty)}(z)$ satisfying
\begin{itemize}
\item[\rm (1)] $P^{(\infty)}(z)$ is analytic for  $z \in \C \setminus \left( (-\infty,-1] \cup [1,\infty) \right) $;
\item[\rm (2)] $P^{(\infty)}$ has the following jumps
\begin{align}
P^{(\infty)}_+(x)=P^{(\infty)}_-(x)\begin{pmatrix} 0&0&1&0\\0&1&0&0\\-1&0&0&0\\0&0&0&1 \end{pmatrix}, & \qquad x \in (1,\infty), \\
P^{(\infty)}_+(x)=P^{(\infty)}_-(x)\begin{pmatrix} 1&0&0&0\\0&0&0&1\\0&0&1&0\\0&-1&0&0 \end{pmatrix}, & \qquad x \in (-\infty,-1), 
\end{align}
where the contour is oriented from left to right;
\item[\rm (3)] As $z \to \infty$ we have 
\[
P^{(\infty)}(z)= \left(I+\mathcal O \left(z^{-1}\right) \right)\diag \left((-z)^{-1/4},z^{-1/4},(-z)^{1/4},z^{1/4} \right)A, \qquad \text{as }z \to \infty.
\]
\end{itemize}

It is straightforward to check that
\[
P^{(\infty)}(z)=\diag \left( (1-z)^{-1/4},(1+z)^{-1/4},(1-z)^{1/4},(1+z)^{1/4}\right)A
\]
solves the RH problem.
\subsection{Local parametrices $P^{(\pm 1)}$ around $\pm 1$}

Local parametrices  $P^{(\pm 1)}$ in a neighborhood $U_1$ of $1$ and a neighborhood $U_{-1}$ of $-1$ can be constructed using Airy functions. This construction is well-known in the literature and will therefore be left to the reader. Also, we will not need the exact form of the solution for our analysis.

\subsection{Local parametrix $P^{(0)}$ around 0}

In this part we will construct a local parametrix around zero out of solutions for two $2 \times 2$ size model RH problems. As can be expected from the structure of the Riemann surface on the critical curve $\tau=\sqrt{\alpha+2}$, one of these will be related to the PII model RH problem (RH problem \ref{rhp: PII rhp}), while the other one is constructed from the sine function. More precisely we use the following model RH problems.

\paragraph{Rotated PII model RH problem}
The first model RH problem will be nothing more than a rotated version of RH problem \ref{rhp: PII rhp}. It is immediate that 
\[
\widetilde \Psi (\zeta)=\Psi(i\zeta),
\]
satisfies the following RH problem.

\begin{rhp} \label{rhp: PII rhp rotated}
We look for a $2 \times 2$ matrix-valued function $\widetilde \Psi(\zeta)$ satisfying
\begin{itemize}
\item[\rm (1)] $\widetilde \Psi(\zeta)$ is analytic for  $\zeta \in \C \setminus \Sigma_{\widetilde \Psi}$;
\item[\rm (2)] $\widetilde \Psi_+(\zeta)=\widetilde \Psi_-(\zeta)J_{\widetilde \Psi}$, for $\zeta \in \Sigma_{\widetilde \Psi}$;
\item[\rm (3)] As $\zeta \to \infty$ we have
\[
\widetilde \Psi(\zeta)=\left( I+ \mathcal O(\zeta^{-1}) \right) \begin{pmatrix} e^{- \frac43 \zeta^3+\nu \zeta } & 0 \\ 0& e^{ \frac43 \zeta^3-\nu \zeta} \end{pmatrix};
\]
\item[\rm (4)] $\Psi(\zeta)$ is bounded near $\zeta=0$.
\end{itemize}
\end{rhp}

The contour $\Sigma_{\widetilde \Psi}$ consists of the rays $\arg(\zeta)=\pm \pi /3$ and $\arg(\zeta)=\pm 2\pi /3$. The jumps on these rays are presented in Figure \ref{fig: contour rotated Psi}.

\begin{figure}[t]
\centering
\begin{tikzpicture}[scale=1]
\begin{scope}[decoration={markings,mark= at position 0.5 with {\arrow{stealth}}}]
\draw[postaction={decorate}]      (0,0)--(1,-1.73) node[right]{$\begin{pmatrix} 1&0 \\ 1&1 \end{pmatrix}$};
\draw[postaction={decorate}]      (0,0)--(1,1.73) node[right]{$\begin{pmatrix} 1&0 \\ -1&1 \end{pmatrix}$};
\draw[postaction={decorate}]      (0,0)--(-1,1.73) node[left]{$\begin{pmatrix} 1&1 \\ 0&1 \end{pmatrix}$};
\draw[postaction={decorate}]      (0,0)--(-1,-1.73) node[left]{$\begin{pmatrix} 1&-1 \\ 0&1 \end{pmatrix}$};
\end{scope}
\end{tikzpicture}
\caption{The jump contour $\Sigma_{\widetilde \Psi}$ in the complex $\zeta$-plane and the constant jump matrices $J_{\widetilde \Psi}$ on each of the rays.}
\label{fig: contour rotated Psi}
\end{figure}
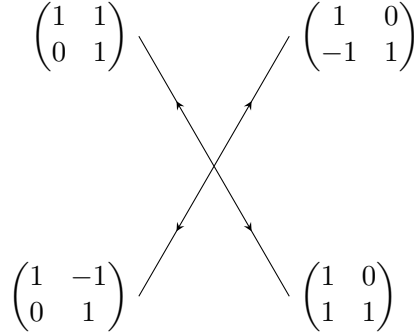

\paragraph{Rotated sine model RH problem}

The second model RH problem that we need is almost a trivial one. It is related to the sine kernel, but since also here we use a rotated version exponential functions arise. 

\begin{rhp} \label{rhp: sine rhp}
We look for a $2 \times 2$ matrix-valued function $\Theta(\zeta)$ satisfying
\begin{itemize}
\item[\rm (1)] $\Theta(\zeta)$ is analytic for  $\zeta \in \C \setminus \Sigma_{\Theta}$;
\item[\rm (2)] $\Theta_+(\zeta)=\Theta_-(\zeta) J_\Theta$, for $\zeta \in \Sigma_{\Theta}$;
\item[\rm (3)] As $\zeta \to \infty$ we have
\[
\Theta(\zeta)=\left( I+ \mathcal O(\zeta^{-1}) \right) \begin{pmatrix} e^{-\zeta } & 0 \\ 0& e^{\zeta} \end{pmatrix};
\]
\item[\rm (4)] $\Theta(\zeta)$ is bounded near $\zeta=0$.
\end{itemize}
The contour $\Sigma_{\Theta}$ consists of the rays $\arg(\zeta)=\pm \pi /3$ and $\arg(\zeta)=\pm 2\pi /3$ and the jumps $J_\Theta$ on these rays are shown in Figure \ref{fig: contour Theta}.
\end{rhp}

\begin{figure}[t]
\centering
\begin{tikzpicture}[scale=1]
\begin{scope}[decoration={markings,mark= at position 0.5 with {\arrow{stealth}}}]
\draw[postaction={decorate}]      (0,0)--(1,-1.73) node[right]{$\begin{pmatrix} 1&-1 \\ 0&1 \end{pmatrix}$};
\draw[postaction={decorate}]      (0,0)--(1,1.73) node[right]{$\begin{pmatrix} 1&1 \\ 0&1 \end{pmatrix}$};
\draw[postaction={decorate}]      (0,0)--(-1,1.73) node[left]{$\begin{pmatrix} 1&0 \\ -1&1 \end{pmatrix}$};
\draw[postaction={decorate}]      (0,0)--(-1,-1.73) node[left]{$\begin{pmatrix} 1&0 \\ 1&1 \end{pmatrix}$};
\end{scope}
\end{tikzpicture}
\caption{The jump contour $\Sigma_{\Theta}$ in the complex $\zeta$-plane and the constant jump matrices on each of the rays. }
\label{fig: contour Theta}
\end{figure}
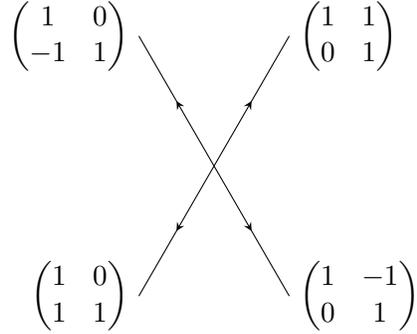

This RH problem has the solution
\[
\Theta(\zeta)=\begin{cases}
\begin{pmatrix} e^{-\zeta} & -e^{-\zeta} \\ 0 & e^{\zeta} \end{pmatrix}, &  |\arg \zeta| < \frac{\pi }{3}, \\
\begin{pmatrix} e^{-\zeta} & 0 \\ 0 & e^{\zeta} \end{pmatrix}, & \frac{\pi }{3} < |\arg \zeta |<\frac{2\pi }{3}, \\
\begin{pmatrix} e^{-\zeta} & 0 \\ -e^{\zeta} & e^{\zeta} \end{pmatrix}, & \frac{2\pi }{3}< |\arg \zeta | \leq \pi , 
\end{cases}
\]
as can be seen by direct verification.

\paragraph{Structure of the local parametrix}

We will define the local parametrix $P^{(0)}(z)$ on a disk $U_0$ of fixed radius $\epsilon$ centered at the origin such that it is analytic inside this disk, except on the contour shown in Figure \ref{fig: contour local parametrix}. On this contour the parametrix needs to have the jumps indicated in the same figure. Moreover, on the circle of radius $\epsilon$ we impose the matching condition
\begin{equation} \label{eq: matching condition}
P^{(0)}(z)\left({P^{(\infty)}}\right)^{-1}(z)=I+ \mathcal O(1/a), \qquad a \to \infty.
\end{equation}

\begin{figure}[t]
\centering
\begin{tikzpicture}[scale=1]
\begin{scope}[decoration={markings,mark= at position 0.5 with {\arrow{stealth}}}]
\draw[postaction={decorate}]      (0,0)--(1,-1.73) node[below right ]{$\begin{pmatrix} 1&0&0&0 \\ e^{a^3(g_2-g_1)}&1&0&0 \\ 0&0&1& -e^{a^3(g_3-g_4)} \\ 0&0&0&1 \end{pmatrix}$};
\draw[postaction={decorate}]      (0,0)--(1,1.73) node[above right ]{$\begin{pmatrix} 1&0&0&0 \\ -e^{a^3(g_2-g_1)}&1&0&0 \\ 0&0&1& e^{a^3(g_3-g_4)} \\ 0&0&0&1 \end{pmatrix}$};
\draw[postaction={decorate}]      (0,0)--(-1,1.73) node[above left ]{$\begin{pmatrix} 1&e^{a^3(g_1-g_2)}&0&0 \\0&1&0&0\\ 0 &0&1&0 \\ 0&0&-e^{a^3(g_4-g_3)} &1 \end{pmatrix}$};
\draw[postaction={decorate}]      (0,0)--(-1,-1.73) node[below left ]{$\begin{pmatrix} 1&-e^{a^3(g_1-g_2)}&0&0 \\0&1&0&0\\ 0 &0&1&0 \\ 0&0&e^{a^3(g_4-g_3)} &1 \end{pmatrix}$};
\draw (0,0) circle (2);
\end{scope}
\end{tikzpicture}
\caption{The jump contour and jump matrices for the local parametrix $P^{(0)}(z)$ around zero. On the circle the matching condition \eqref{eq: matching condition} is imposed.}
\label{fig: contour local parametrix}
\end{figure}
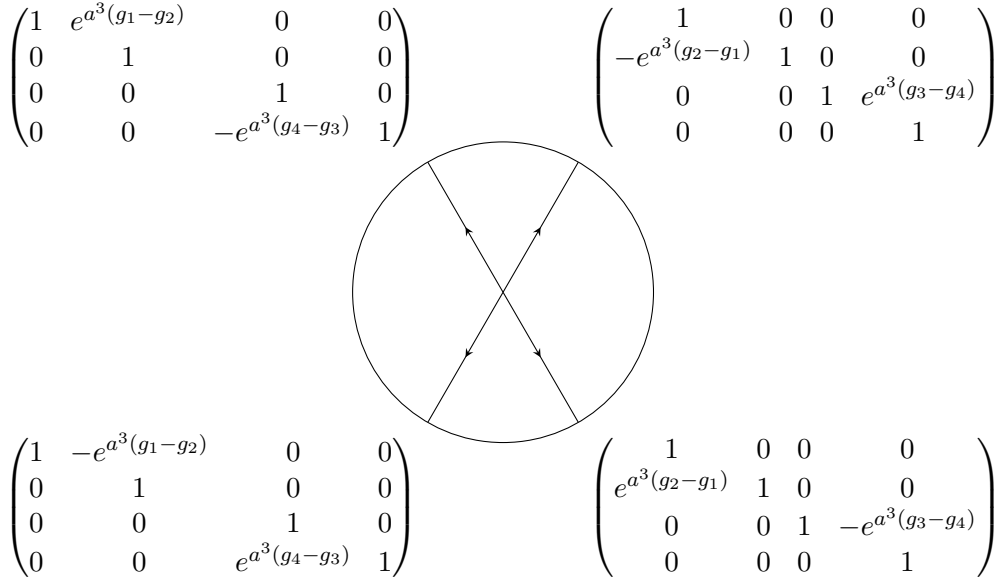

We look for a solution to this problem of the form
\begin{multline*}
P^{(0)}(z)=P^{(\infty)}(z) \begin{pmatrix} \widetilde \Psi(af_1(z);\nu(z)) & 0 \\ 0 & \Theta(a^3f_2(z)) \end{pmatrix}\\ \times \diag \left( e^{-a^3\frac{g_1-g_2}{2}},e^{a^3\frac{g_1-g_2}{2}},e^{-a^3\frac{g_3-g_4}{2}},e^{a^3\frac{g_3-g_4}{2}}\right).
\end{multline*}
We now define the conformal maps $f_1$ and $f_2$ and the map $\nu$, starting with the first. We define
\begin{align*}
f_1(z)&=\left(\frac14\left((1-z)^{3/2}-(1+z)^{3/2}\right)+\frac34z \right)^{1/3},
\end{align*}
for $z$ in a small neighborhood $U_0$ of the origin. A simple Taylor expansion shows that the cubic root can be taken such that $f_1$ is a conformal map satisfying
\begin{align*}
f_1(z)&=2^{-5/3} z+\mathcal O(z^3), \qquad \text{as }z \to 0.
\end{align*}
Then we define $\nu$ by
\begin{align*}
\nu(z)= \frac{\sigma z}{f_1(z)}.
\end{align*}
It is clear that we can choose $U_0$ sufficiently small such that $z\mapsto \nu(z)$ is analytic in $U_0$. Note that $\nu(0)=2^{5/3}\sigma$. By the definition of $g_j$ in \eqref{eq:defgjsteep2} and this choice of $f_1$ and $\nu$, we have
\begin{align} \label{eq: match1}
a^3\frac43 f_1(z)^3- a \nu(z) f_1(z)=a^{3}(g_2(z)-g_1(z))/2.
\end{align}
Finally, we define $f_2$ by
\begin{align*}
f_2(z) &=\frac12 (g_4(z)-g_3(z)).
\end{align*}
Note that a simple Taylor expansions shows that 
\begin{align} \label{eq: match2}
g_3(z)-g_4(z)&=-4z+\frac{1}{12}z^3+\mathcal O(z^5),\qquad \text{as }z \to 0.
\end{align}
Defined in this way, and possibly after some slight deformation of contours in the rotated PII and rotated sine RH problems, $P^{(0)}(z)$ already satisfies the jump conditions. It then only remains to check the matching condition \eqref{eq: matching condition}.

\paragraph{Matching condition}

Let us fix $z$ on the boundary $\partial U_0$ of the disk with radius $\epsilon$. Then $af_1(z)$ and $a^3f_2(z)$ both tend to infinity as $a \to \infty$. We may therefore plug in the asymptotic formula for $\widetilde \Psi$ and $\Theta$. After a quick calculation using \eqref{eq: match1} it turns out that
\begin{align*}
\widetilde \Psi(af_1(z);\nu(z))\diag \left( e^{-a^3\frac{g_1-g_2}{2}},e^{a^3\frac{g_1-g_2}{2}}\right)& =I_2+\mathcal O(1/a), \quad \text{as }a \to +\infty, \\
\Theta(a^3f_2(z))\diag \left(e^{-a^3\frac{g_3-g_4}{2}},e^{a^3\frac{g_3-g_4}{2}}\right)&= I_2+\mathcal O(1/a), \quad \text{as }a \to +\infty.
\end{align*}
It is then clear that also the matching condition  \eqref{eq: matching condition} is satisfied.

\subsection{Final transformation: $M^{(3)}\mapsto M^{(4)}$}
Finally we define a last matrix function $M^{(4)}(z;a)$ as
\begin{align}\label{eq:defM(4)}
M^{(4)}(z;a)=\begin{cases}
M^{(3)}(z;a)\left(P^{(\infty)}\right)^{-1} & \text{for }z \in \C \setminus (\overline{U_0} \cup \overline{U_1} \cup \overline{U_{-1}} \cup \Sigma_{M^{(3)}}), \\
M^{(3)}(z;a)\left(P^{(0)}\right)^{-1} & \text{for }z \in U_0 \setminus  \Sigma_{M^{(3)}}, \\
M^{(3)}(z;a)\left(P^{(1)}\right)^{-1} & \text{for }z \in U_1 \setminus  \Sigma_{M^{(3)}}, \\
M^{(3)}(z;a)\left(P^{(-1)}\right)^{-1} & \text{for }z \in U_{-1} \setminus  \Sigma_{M^{(3)}}.
\end{cases}
\end{align}
This function is analytic in $\C \setminus \Sigma_{M^{(4)}}$, where $\Sigma_{M^{(4)}}$ is obtained from $\Sigma_{M^{(3)}}$ by removing the rays $(-\infty,-1]$ and $[1,\infty)$, adding three small circles around $-1$, 0, and 1 and finally deleting the parts of $\Sigma_{M^{(3)}}$ within these circles. Then all jumps on $ \Sigma_{M^{(4)}}$ tend to the identity matrix as $a \to +\infty$, both uniformly and in $L^2$-sense. Indeed, for the jump on $\partial U_0$ this is a consequence of the matching condition  \eqref{eq: matching condition}. Likewise the jumps on $\partial U_1$ and $\partial U_{-1}$ tend to identity if one constructs the right local parametrices around $\pm 1$. Moreover $M^{(4)}(z;a)$ is normalized at infinity, meaning that
\[
M^{(4)}(z;a)=I+\mathcal O(1/z), \qquad \text{as }z \to \infty.
\]
Then the following result follows by standard methods.
\begin{lemma} As $a\to \infty$ we have
\[
M^{(4)}(z;a)=I+\mathcal O \left( \frac{1}{a(1+|z|)}\right),
\]
uniformly for $z \in \C \setminus \Sigma_{M^{(4)}}$.
\end{lemma}
This is the final goal of the steepest descent analysis and the starting point of the proof of the main theorem.

\subsection{Proof of Theorem \ref{th: PII}}\label{subsec: proof PII}

\begin{proof}[Proof of Theorem \ref{th: PII}] 
We only give the proof for the case $u,v>0$. Other cases can be proved similarly. 
First note that by \eqref{eq:defKcr}, \eqref{eq: Mbar definition}, and taking the transpose  we obtain
\begin{align*}
\Kcr (u,v;s,t)=\frac{1}{2 \pi i (u-v)} \begin{pmatrix} -1 & 1 & 0 & 0 \end{pmatrix}
 M^{-1}\left(iu; s,t\right) M\left(iv; s,t\right) \begin{pmatrix} 1\\1\\0\\0 \end{pmatrix}.
\end{align*}
Fix $u,v>0$ and let $a$ be a large positive number.  Using \eqref{eq:defM(1)} we write
\begin{multline*}
a^2 \Kcr \left(a^2u,a^2v;\tfrac12{a^2},-a(1-\sigma/a^2)\right)=\frac{1}{2 \pi i (u-v)} \begin{pmatrix} -1 & 1 & 0 & 0 \end{pmatrix}\\
\times {M^{(1)}\left(ia^2u; a\right)}^{-1} M^{(1)}\left(ia^2v; a\right) \begin{pmatrix} 1\\1\\0\\0 \end{pmatrix}.
\end{multline*}
By the transformations $M^{(1)}\mapsto M^{(2)}\mapsto M^{(3)}$ in \eqref{eq:defM(2)} and \eqref{eq:defM(3)} we obtain
\begin{multline*}
a^2 \Kcr \left(a^2u,a^2v;a,\tfrac12{a^2},-a(1-\sigma/a^2)\right)=\frac{1}{2 \pi i (u-v)} \begin{pmatrix} -1 & 1 & 0 & 0 \end{pmatrix}\\
\times \diag \left( e^{-a^3g_{1}(iu)},e^{-a^3g_{2}(iu)},e^{-a^3g_{3}(iu)},e^{-a^3g_{4}(iu)}\right)\left(M^{(3)}\right)^{-1}(iu;a) \\ \times M^{(3)}(iv;a) \diag \left( e^{a^3g_{1}(iv)},e^{a^3g_{2}(iv)},e^{a^3g_{3}(iv)},e^{a^3g_{4}(iv)}\right)\begin{pmatrix} 1\\1\\0\\0 \end{pmatrix}.
\end{multline*}
In the next step we unfold the transformation $M^{(3)}\mapsto M^{(4)}$ in \eqref{eq:defM(4)}. 
Assume that $u,v \in U_0$. Then we obtain 
\begin{multline*}
a^2 \frac{e^{a^3(g_1(iu)+g_2(iu))/2}}{e^{a^3(g_1(iv)+g_2(iv))/2}}\Kcr \left(a^2u,a^2v;\tfrac12 a^2,-a(1-\sigma/a^2)\right)=\frac{1}{2 \pi i (u-v)}  \\ 
\times \begin{pmatrix} -1 & 1 & 0 & 0 \end{pmatrix}\begin{pmatrix} \widetilde \Psi(af_1(iu);\nu(iu))^{-1} & 0 \\ 0&0 \end{pmatrix} {P^{(\infty)}}(iu)^{-1}M^{(4)}(iu;a)^{-1} \\ \times M^{(4)}(iv;a)P^{(\infty)}(iv)\begin{pmatrix} \widetilde \Psi(af_1(iv);\nu(iv)) & 0 \\ 0&0 \end{pmatrix}
 \begin{pmatrix} 1\\1\\0\\0 \end{pmatrix}.
\end{multline*}
Now we make the change of variables $u=2^{5/3}x/a$ and $v=2^{5/3}y/a$. We consider $x$ and $y$ as fixed so that $u,v \to 0$ as $a \to +\infty$. Note that after this change of variables
\begin{align*}
af_1(iu)\to ix, \qquad \text{as }a \to +\infty, \\
af_1(iv)\to iy, \qquad \text{as }a \to +\infty.
\end{align*}
Also note that $\nu(iu)\to \nu(0)=2^{5/3} \sigma$.

It follows from standard considerations  (see also \eqref{eq:cauchyonR}) that 
\[
M^{(4)}(iu;a)^{-1} M^{(4)}(iv;a)=I+\mathcal O \left( \frac{u-v}{a}\right)=I+\mathcal O \left( \frac{x-y}{a^2}\right), \qquad \text{as } a \to +\infty,
\]
uniformly for $x$ and $y$ in a compact subsets of $\R$. Observe also that 
\begin{align*}
P^{(\infty)}(iv)=(I+\mathcal O(1/a))A, \qquad \text{as }a \to +\infty, \\
P^{(\infty)}(iu)^{-1}=A^{-1}(I+\mathcal O(1/a)), \qquad \text{as }a \to +\infty.
\end{align*}
Therefore,
\begin{align}\label{eq:laatste}
\begin{aligned}
&\lim_{a\to +\infty} 2^{5/3}a\frac{e^{a^3(g_1(iu)+g_2(iu))/2}}{e^{a^3(g_1(iv)+g_2(iv))/2}}\Kcr \left(2^{5/3}ax,2^{5/3}ay;\tfrac12 a^2,-a(1-\sigma/a)\right) \\  
&\qquad = \frac{1}{2 \pi i (x-y)}   \begin{pmatrix} -1 & 1  \end{pmatrix}\widetilde \Psi(ix;2^{5/3}\sigma)^{-1}  \widetilde \Psi(iy;2^{5/3}\sigma)  \begin{pmatrix} 1\\1 \end{pmatrix} \\ 
&\qquad =\frac{1}{2 \pi i (x-y)}   \begin{pmatrix} -1 & 1  \end{pmatrix} \Psi(-x;2^{5/3}\sigma)^{-1}   \Psi(-y;2^{5/3}\sigma) \begin{pmatrix} 1\\1 \end{pmatrix}.
\end{aligned} 
\end{align}
The convergence is uniform for $x$ and $y$ in compact subsets of $\R$. Now observe that  $\Psi(z)$  satisfies the symmetry
\[
\begin{pmatrix} 0&1 \\ 1&0 \end{pmatrix} \Psi(z;2^{5/3}\sigma) \begin{pmatrix} 0&1 \\ 1&0 \end{pmatrix}=\Psi(-z;2^{5/3}\sigma).
\]
This can be easily verified by checking that both sides solve the same RH problem and noting that the RH problem has a unique solution. Finally, after inserting this symmetry in \eqref{eq:laatste} and using  \eqref{eq: PII kernel} we obtain 
\begin{align*}
\lim_{a\to +\infty} 2^{5/3}a\frac{e^{a^3(g_1(iu)+g_2(iu))/2}}{e^{a^3(g_1(iv)+g_2(iv))/2}}\Kcr \left(2^{5/3}ax,2^{5/3}ay;\tfrac12{a^2},-a(1-\sigma/a)\right) =
 K_{\rm PII}(x,y;2^{5/3}\sigma).
\end{align*} 
This completes the proof of Theorem \ref{th: PII} for $u,v>0$.  
\end{proof}

\appendix

\section{Asymptotic behavior of $K_{\rm cr}$ and $K_{\rm tac}$ along the diagonal}

Here we prove the asymptotic expansions \eqref{eq:tacnodekernelexp} and \eqref{eq:Kcrexp} for the diagonal of  $K_{\rm cr}$ as defined in \eqref{eq:defKcr} and  $K_{\rm tac}$ in \eqref{eq:tacnodekernel} . First we note that that the diagonal terms can be obtained by taking the limit $u\to v$ in \eqref{eq:defKcr} and  \eqref{eq:tacnodekernel}, which leads to \begin{align}
K_{\rm cr}(u,u)&= \frac{1}{2\pi} \begin{pmatrix} 1&-1&0&0 \end{pmatrix} M^{-1}(iu)\frac{\partial M}{\partial \zeta}(iu) \begin{pmatrix} 1\\1\\0\\0 \end{pmatrix}, \\
K_{\rm tac}(u,u)&= \frac{1}{2\pi i} \begin{pmatrix} -1&0&1&0 \end{pmatrix} M_+^{-1}(u)\frac{\partial M_+}{\partial \zeta}(u) \begin{pmatrix} 1\\0\\1\\0 \end{pmatrix}, 
\end{align}
for $u>0$.
Next we write
\[
M(\zeta)=Y(\zeta)Q(\zeta),
\]
where $Y$ has a full and uniform asymptotic expansion in inverse powers of $\zeta$
\begin{equation} \label{eq: asymptotics Y}
Y(\zeta)=I_4+\frac{M_1}{\zeta}+\mathcal O \left( \zeta^{-2} \right), \qquad \text{as } \zeta \to \infty, \quad \zeta \in \C,
\end{equation}
and
\begin{multline*}
Q(\zeta)=\diag ((-\zeta)^{-1/4},\zeta^{-1/4},(-\zeta)^{1/4},\zeta^{1/4}) \\ \times A
\diag \left( e^{-\psi_2(\zeta)+t\zeta},e^{-\psi_1(\zeta)-t\zeta}, e^{\psi_2(\zeta)+t\zeta},e^{\psi_1(\zeta)-t\zeta} \right).
\end{multline*}
Here $\psi_1$ and $\psi_2$ are as defined in \eqref{eq: def psi1}--\eqref{eq: def psi2}.
A straightforward computation shows
\[
Q'(\zeta)=R(\zeta)Q(\zeta),
\]
where
\[
R(\zeta)=\begin{pmatrix}
-\frac{1}{4\zeta}+t & 0 & ir_2-\frac{is_2}{\zeta} & 0 \\ 0 & -\frac{1}{4\zeta}-t & 0 & ir_1+\frac{is_1}{\zeta} \\
ir_2\zeta-is_2 & 0  & \frac{1}{4\zeta}+t & 0 \\ 0 & -ir_1\zeta-is & 0 & \frac{1}{4\zeta}-t 
\end{pmatrix}.
\]
Then
\[
M^{-1}(\zeta)\frac{\partial M}{\partial \zeta}(\zeta)=Q^{-1}(\zeta)\left(Y^{-1}(\zeta)Y'(\zeta)+R(\zeta)\right)Q(\zeta).
\]
When we let $\zeta$ tend to $\infty$ the contribution of the first term is of order $\mathcal O (\zeta^{-2})$ so we will concentrate on the second term. 
For the tacnode one-point correlation function we find
\begin{align}
\begin{split}
K_{\rm tac}(u,u;r_1,r_2,s_1,s_2)&=\frac{1}{2\pi i} \begin{pmatrix} -1&0&1&0 \end{pmatrix} Q^{-1}_+(u)R(u)Q_+(u) \begin{pmatrix} 1\\0\\1\\0 \end{pmatrix}+\mathcal O(u^{-3/2}) \\
&= -\frac{1}{2\pi i}\left(2r_2(-u)^{1/2}+2s_2(-u)^{-1/2}+\frac{i}{4u}\left( e^{2\psi_2(u)}+e^{-2\psi_2(u)}\right) \right)+\mathcal O(u^{-3/2}) \\
&= \frac{r_2\sqrt u}{\pi}-\frac{s_2}{\pi \sqrt u}-\frac{1}{4\pi u}\Re e^{2\psi_2(u)}+\mathcal O(u^{-3/2}),
\end{split}
\end{align}
as $u\to +\infty$.
Note that there is no constant term in the expansion and that the $1/u$ has a highly oscillatory coefficient of modulus 1.
For the two-matrix model we can compute
\begin{align}
\begin{split}
K_{\rm cr}(u,u;s,t)&=\frac{1}{2\pi} \begin{pmatrix} 1&-1&0&0 \end{pmatrix} Q^{-1}(iu)R(iu)Q(iu) \begin{pmatrix} 1\\1\\0\\0 \end{pmatrix}+\mathcal O(u^{-3/2}) \\
&= \frac{1}{2\pi}\left((-iu)^{1/2}+(iu)^{1/2}+2t+s(-iu)^{-1/2}+s(iu)^{-1/2} \right)+\mathcal O(u^{-3/2}) \\
&= \frac{\sqrt u}{\sqrt 2 \pi}+\frac{t}{\pi}+\frac{s}{\sqrt 2 \pi \sqrt u}+\mathcal O(u^{-3/2}),
\end{split}
\end{align}
as $u\to +\infty$. Here we do have a constant contribution, but there is no $1/u$ term. Therefore, these two expansions cannot be the same and hence the kernels must be different.


\begin{thebibliography}{99}

\bibitem{AFM}
	  M. Adler, P. Ferrari and P. van Moerbeke,
	 \emph{Non-intersecting random walks in the neighborhood of a symmetric tacnode}, to appear in Ann. of Prob.,
	 arXiv1007.1163
	 
\bibitem{ABK}
	  A. Aptekarev, P. Bleher and A.B.J. Kuijlaars, 
	  \emph{Large n limit of Gaussian random matrices with external
	  source, part II}, 
	  Comm. Math. Phys. 259 (2005), no. 2, 367--389.

\bibitem{BEH1}
	M. Bertola,  B.Eynard and J. Harnad,
	\emph{The PDEs of biorthogonal polynomials arising in the two-matrix model},
	Math. Phys. Anal. Geom. 9 (2006), 162--212.
 

\bibitem{BEH2}
	M. Bertola,  B.Eynard and J. Harnad,
	\emph{Duality, biorthogonal polynomials and multi-matrix models},
	Comm. Math. Phys. 229 (2002), 73--120.
	  
\bibitem{BEH}
	M. Bertola,  B.Eynard and J. Harnad,
	\emph{Differential systems for biorthogonal polynomials appearing in 2-matrix models and the associated Riemann-Hilbert problem}, Comm. Math. Phys. 243 (2003), 193--240.	  
	
\bibitem{BI}
   	P. Bleher and A. Its,
	\emph{Double scaling limit in the random matrix model: the Riemann-Hilbert approach}, Comm. Pure Appl. Math 56 (2003), 433--516
	
\bibitem{BK} P. Bleher and A.B.J. Kuijlaars,
    \emph{Large $n$ limit of Gaussian random matrices with external source, part III: double scaling limit},
    Comm. Math. Phys. 270 (2007), 481--517.
  
\bibitem{Bor}
    A. Borodin,
    \emph{Biorthogonal ensembles},
    Nucl. Phys. B { 536}, (1999), no. 3, 704-732.
        
\bibitem{BorDet} 
 	  A. Borodin, 
	  \emph{Determinantal point processes},  
	  in: ``Oxford Handbook on Random Matrix theory'', edited by G. Akemann, J. Baik
	  and P. Di Francesco, Oxford University Press, 2011. (arXiv:0911.1153)

\bibitem{BH}
 	  E. Br\'ezin and S. Hikami, 
 	  \emph{Universal singularity at the closure of a gap in a random matrix theory}, Phys.
	  Rev. E. (3) 57 (1998), no. 4. 7176--7185.
	
\bibitem{BH1}
	  E. Br\'ezin and S. Hikami, 
	  \emph{Level spacing of random matrices in an external source},
	  Phys. Rev. E. (3) 58 (1998), no. 6, part A, 4140--4149.
	 
\bibitem{CK}
	  T. Claeys and A.B.J. Kuijlaars, 
	  \emph{Universality of the double scaling limit in random matrix models},
	  Comm. Pure Appl. Math. 59 (2006), 1573--1603.
	
\bibitem{CKV}
	  T. Claeys, A.B.J. Kuijlaars and M. Vanlessen,
	  \emph{Multi-critical unitary random matrix ensembles and the general Painlev\'e II equation}, 
	  Ann. of Math. 168 (2008), 601--642.
	  
	
\bibitem{DaKaKo}
    J.M. Daul, V. Kazakov, and I.K. Kostov,
    \emph{Rational theories of 2D gravity from the two-matrix model},
    {Nucl. Phys. B} { 409} (1993), 311--338.
    
    
\bibitem{Deift} P. Deift,
		\emph{Orthogonal Polynomials and Random Matrices: a Riemann-Hilbert approach},
		Courant Lecture Notes in Mathematics Vol. 3, Amer. Math. Soc., Providence R.I. 1999.
		
\bibitem{DKMVZ1}
 	P. Deift, T. Kriecherbauer, K.T.-R. McLaughlin, S. Venakides, and X. Zhou,
	\emph{Uniform asymptotics for polynomials orthogonal with respect to varying exponential weights and applications to universality questions in random matrix theory},
	Comm. Pure Appl. Math. 52 (1999), 1335--1425.
	

\bibitem{DKMVZ2}
	P. Deift, T. Kriecherbauer, K.T.-R. McLaughlin, S. Venakides, and X. Zhou,
	\emph{Strong asymptotics for polynomials orthogonal with respect to varying exponential weights},
	Comm. Pure Appl. Math. 52 (1999), 1491--1552.
	
\bibitem{DZ}
	P. Deift and X. Zhou, 
	\emph{A steepest descent method for oscillatory Riemann-Hilbert problems. Asymptotics for the MKdV equation},
	Ann. of Math. (2) 137 (1993), 295--368.

\bibitem{DelvauxK}
	S. Delvaux and A.B.J. Kuijlaars,
	\emph{A phase transition for non-intersecting Brownian motions and the Painlev\'e II equation},
	 Int. Math. Res. Not. 2009  no. 19, 3639--3725.

\bibitem{DKZ} 
    S. Delvaux, A.B.J. Kuijlaars and L. Zhang, 
    \emph{Critical behavior of non-intersecting Brownian motions at a tacnode},
     Comm. Pure and Appl. Math 64 (2011), 1305--1383.

\bibitem{DK} 
    K. Deschout and A.B.J. Kuijlaars,
    \emph{Double scaling limit for modified Jacobi-Angelesco polynomials},
    in: ``Notions of Positivity and the Geometry of Polynomials'', edited by P. Brändén, M. Passare and M. Putinar, Trends in Mathematics, Springer, Basel, 2011, 115--161. (arXiv:1102.1349) 

\bibitem{DGK} 
    M. Duits, D. Geudens and A.B.J. Kuijlaars, 
    \emph{A vector equilibrium problem for the two-matrix model in the quartic/quadratic case}, 
    Nonlinearity 24 (2011), no. 3, 951--993.

\bibitem{DK2} M. Duits and A.B.J. Kuijlaars,
    \emph{Universality in the two-matrix model: a Riemann-Hilbert steepest descent analysis},
    Comm. Pure Appl. Math. 62 (2009), 1076--1153.

\bibitem{DKM} M. Duits, A.B.J. Kuijlaars and M. Y. Mo,
   \emph{The Hermitian two-matrix model with an even quartic potential},
   to appear in Memoirs Amer. Math. Soc.
   
\bibitem{Dou}
    M.R. Douglas,
    \emph{The two-matrix model},
    in: ``Random surfaces and quantum gravity'', NATO Adv. Sci. Inst. Ser. B Phys., 262, Plenum,
    New York, (1991), pp. 77--83.
   
\bibitem{EMc} N.M. Ercolani and K.T.-R. McLaughlin,
    \emph{Asymptotics and integrable structures for biorthogonal polynomials
    associated to a random two-matrix model},
    Physica D 152/153 (2001), 232--268.
      
\bibitem{Eyn2}
    B. Eynard,
    \emph{Large-$N$ expansion of the 2 matrix model},
    { J. High Energy Phys.} (2003), no. 1, 051, 38p.
  
\bibitem{EyM} B. Eynard and M.L. Mehta,
    \emph{Matrices coupled in a chain: eigenvalue correlations},
    J. Phys. A 31 (1998) , 4449--4456.

\bibitem{FN}
	 H. Flaschka and A.C. Newell, 
	 \emph{Monodromy and spectrum-preserving deformations I},
	  Comm. Math. Phys. 76 (1980), 65--116.
	  
\bibitem{FIK}
	A.S. Fokas, A. R. Its and A.V. Kitaev, 
	\emph{The isomonodromy approach to matrix models in 2D quantum gravity},
	Comm. Math. Phys. 147 (1992), 395--430.
	
\bibitem{FIKN}
    A.S. Fokas, A.R. Its, A.A. Kapaev and V.Yu. Novokshenov, 
    \emph{Painlev\'e Transcendents: a Riemann-Hilbert Approach}, 
    Mathematical Surveys and Monographs 128, 
    Amer. Math. Soc., Providence R.I. 2006.
   
\bibitem{HMcL}
    S.P. Hastings and J.B. McLeod, \emph{A boundary value problem associated with the second Painlev\'e transcendent 
    and the Korteweg-de Vries equation},
    Arch. Rational Mech. Anal. 73 (1980), 31--51.

\bibitem{HK}
    A. Hardy and A.B.J. Kuijlaars, \emph{Weakly admissible vector equilibrium problems},
    in preparation.

\bibitem{HKPV} 
  	J. B. Hough, M. Krishnapur, Y. Peres and B. Vir\'ag, 
	  \emph{Determinantal processes and independence}, 
	  Prob. Surv. 3 (2006), 206--229.

\bibitem{Jdet} 
	  K. Johansson, 
	  \emph{Random matrices and determinantal processes},  
	  Mathematical Statistical Physics, Elsevier B.V.\ Amsterdam (2006), 1--55.
	
\bibitem{J}
	  K. Johansson, 
	  \emph{Non-colliding Brownian Motions and the extended tacnode process},
	  arXiv:1105.4027.

\bibitem{Kapaev}
	A.A. Kapaev,
	\emph{Riemann-Hilbert problem for bi-orthogonal polynomials},
	J. Phys. A 36 (2003), 4629--4640.
	
\bibitem{KMcL}
    A.B.J. Kuijlaars and K.T-R. McLaughlin,
   \emph{ A Riemann-Hilbert problem for biorthogonal polynomials},
    { J. Comput. Appl. Math.} {\bf 178} (2005), 313--320.

\bibitem{KMW}
    A.B.J. Kuijlaars, A. Martinez-Finkelshtein and F. Wielonsky,
    \emph{Non-intersecting squared Bessel paths: Critical time and double scaling limit},
    Comm. Math. Phys. 308 (2011), 227--279.

\bibitem{K} 
	  W.  K\"onig, 
	  \emph{Orthogonal polynomial ensembles in probability theory}, 
	  Probab. Surveys 2 (2005), 385--447.
	  
	
\bibitem{L} 
	  R. Lyons, 
	  \emph{Determinantal probability measures}, 
	  Publ. Math. Inst. Hautes Etudes Sci. 98  (2003), 167--212.
 
\bibitem{Mo} M.Y. Mo,
    \emph{Universality in the two matrix model with a monomial quartic and a
    general even polynomial potential},
    Comm. Math. Phys. 291 (2009), 863--894.

\bibitem{MS} M. Mehta and P. Shukla,
    \emph{Two coupled matrices: eigenvalue correlations and spacing functions},
    J. Phys. A 27 (1994), 7793--7803.

\bibitem{OR}
		A. Okounkov and N. Reshetikhin, 
		\emph{Random skew plane partitions and the Pearcey process}, 
		Comm. Math. Phys. 269 (2007), no.3, 571--609.

\bibitem{SaffTotik}
	E. B. Saff and V. Totik, 
	\emph{Logarithmic Potentials with External Field},
	Grundlehren der Mathematischen Wissenschaften 316, Springer-Verlag, Berlin, 1997.
	 
\bibitem{Sosh} 
	  A. Soshnikov, 
	  \emph{Determinantal random point fields},  
	  Uspekhi Mat. Nauk 55 (2000), no. 5 (335), 107--160;
	  translation in Russian Math. Surveys 55 (2000), no. 5, 923--975.
		
\bibitem{TW}
	  C. Tracy and H. Widom, 
	  \emph{The Pearcey process}, 
	  Comm. Math. Phys 263 (2006), 381--400.

\bibitem{Wa} W. Wasow,
    \emph{Asymptotic expansions for ordinary differential equations},
    Pure and applied mathematics Vol. XIV, John Wiley \& sons, 1965.

\end{thebibliography}
\end{document}